\documentclass[11 pt, onecolumn, journal]{IEEEtran}

\usepackage{refcheck}
\usepackage{graphicx}
\usepackage{amsmath}
\usepackage{amssymb}
\usepackage{graphics}
\usepackage{latexsym}
\usepackage[dvips]{epsfig}
\usepackage[nospace]{cite}
\usepackage{tabularx}
\usepackage{subfigure}
\usepackage{epstopdf}
\usepackage{mathdots}
\usepackage{stmaryrd}
\usepackage[T1]{fontenc}
\usepackage[utf8]{inputenc}

\def\withcolors{1}
\def\withnotes{1}

\ifnum\withcolors=1
\else

\fi

\usepackage[colorinlistoftodos,textsize=scriptsize]{todonotes}
\ifnum\withnotes=1

\else

\fi

\newtheorem{Definition}{Definition}
\newtheorem{Lemma}{\bf Lemma}
\newtheorem{Corollary}{\bf Corollary}
\newtheorem{Proposition}[Lemma]{\bf Proposition}
\newtheorem{Theorem}{\bf Theorem}

\newtheorem{Example}{Example}
\newtheorem{Remark}{Remark}

\def\Pr{{\rm \mathbf {Pr}}}
\def\E{{\rm \mathbf  E}}

\def\var{{\rm \bold {var}}}

\def\essup{{\rm \text{ess-sup}}}
\def\Poi{{\mathrm{Poi}}}

\newcommand{\argmax}{\operatornamewithlimits{argmax}}
\newcommand{\argmin}{\operatornamewithlimits{argmin}}
\newcommand{\ml}[2]{\mathcal{L}\left(#1  \!\!  \to  \!\!   #2\right)} 
\newcommand{\mlt}[2]{{L}\left(#1     \!\!  \to  \!\!    #2\right)}
\newcommand{\mcl}[2]{\mathcal{L}^c\left(#1    \!\!  \to  \!\! #2 \right)} 
\newcommand{\cml}[3]{\mathcal{L}\left(#1    \!\!  \to  \!\!    #2 | #3 \right)} 
\newcommand{\mly}[2]{\widetilde{\mathcal{L}} \left(#1    \!\!  \to  \!\!   #2 \right)} 
\newcommand{\kml}[2]{\mathcal{L}^{(k)} \left(#1    \!\!  \to  \!\!   #2 \right)} 
\newcommand{\mlg}[2]{\mathcal{L}_G \left(#1    \!\!  \to  \!\!    #2 \right)} 
\newcommand{\mlr}[2]{\mathcal{L}^r \left(#1  \!\!  \to  \!\! #2 \right)} 

\allowdisplaybreaks

\let\svthefootnote\thefootnote

\begin{document}

\sloppy

\title{An Operational Approach \\to Information Leakage}
 \author{
  Ibrahim Issa, Aaron B. Wagner, and Sudeep Kamath
 }


 \maketitle

\begin{abstract} 
Given two random variables $X$ and $Y$, an operational approach is undertaken to quantify the ``leakage'' of information from $X$ to $Y$. The resulting measure $\ml{X}{Y}$ is called \emph{maximal leakage}, and is defined as the multiplicative increase, upon observing $Y$, of the probability of correctly 
guessing a randomized function of $X$, maximized over all such randomized 
functions. A closed-form expression for $\ml{X}{Y}$  is given for discrete $X$ and $Y$, and it is subsequently generalized to handle a large class of random variables.  
The resulting properties are shown to be consistent with an axiomatic view of a leakage measure, and the definition is shown to be robust to variations in the setup. Moreover, a variant of the Shannon cipher system is studied, in which performance of an encryption scheme is measured using maximal leakage. 
A single-letter characterization of
the optimal limit of (normalized) maximal leakage is derived and asymptotically-optimal
encryption schemes are demonstrated. Furthermore, the sample complexity of estimating maximal leakage from data is characterized up to subpolynomial factors.
Finally, the \emph{guessing} framework used to define maximal leakage is used to give operational interpretations of commonly used leakage measures, such as Shannon capacity, maximal correlation, and local differential privacy.
\end{abstract}

\section{Introduction}
\let\thefootnote\relax\footnote{I. Issa is with the School of Computer and Communication Sciences, Ecole Polytechnique F\'ed\'erale (EPFL), Lausanne, Switzerland (e-mail: ibrahim.issa@epfl.ch), and was with the School of Electrical and Computer Engineering at 
Cornell University, Ithaca, NY, when much of this work was conducted. A. B. Wagner is with the School of Electrical and Computer Engineering, Cornell University, Ithaca, NY (e-mail: wagner@cornell.edu). S. Kamath is with PDT Partners, New York, NY (e-mail: sudeep.kamath@gmail.com). Parts of this work were presented at the 2016 Annual Conference on Information Sciences and Systems, and at the 2016 and 2017 IEEE International Symposium on Information Theory.} 
\addtocounter{footnote}{-1}\let\thefootnote\svthefootnote
How much information does an observation ``leak'' about a quantity on which it depends? This basic question arises in many secrecy and privacy problems in which the quantity of interest is considered sensitive and an observation is available to an adversary. 
The observation could be intentionally provided to the adversary, as occurs
when a curator publishes statistical information about a given 
population. Or the observation could be an inevitable, if
undesirable, consequence of a design. In the latter case, which is
the focus of this paper, we call the observation the output of a 
\emph{side channel}. Some examples of side channels include:
\begin{itemize}
\item When using the Secure Shell (SSH), after the initial handshake, each keystroke is sent immediately to the remote machine, as shown in Figure~\ref{FigSSH}. When communicating
over a wireless network, an eavesdropper can observe the timing
of the packets which are correlated with the timing of the keystrokes, and hence with the input of the user (e.g., the inter-keystroke delay in
'ka' is significantly smaller than that in '9k'~\cite{SSHTiming}).
\begin{figure}[htp]
\centering
\includegraphics[scale=0.6]{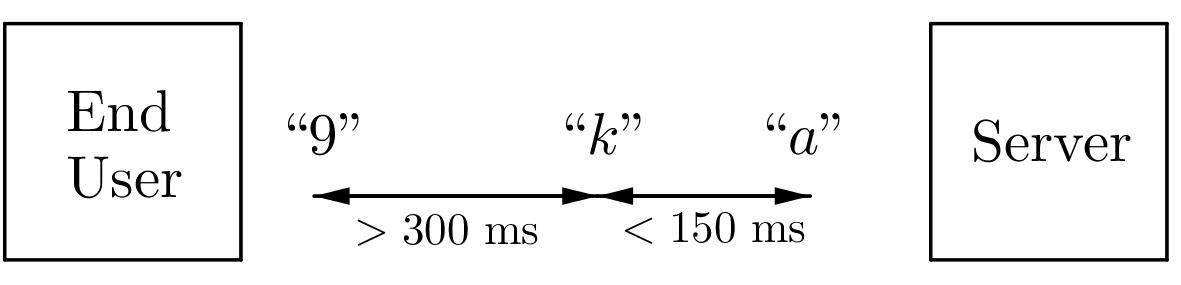}
\caption{The Secure Shell: each keystroke is sent immediately to the remote machine.} \label{FigSSH}
\end{figure}
\item Consider an on-chip network that has several processes running simultaneously, one of which is malicious. Because resources such as memory and buses are shared on the chip, the timing characteristics (e.g., memory access delays) observed by the malicious application are affected by the behavior of the other applications (e.g., memory access patterns) and can leak sensitive information such as keys. Similar phenomena occur when users share links or buffers in a communication network~\cite{NegarCapacityFCFS}.
\item  Consider the Shannon cipher system (shown in Figure~\ref{figureShannonCipher1}) in which a transmitter and a receiver are connected through a public noiseless channel and share a secret key. Unless the key rate is very high, the public message depends on the message~\cite{ShannonSecrecy}. 
\begin{figure}[htp]
\centering
\includegraphics[scale=0.8]{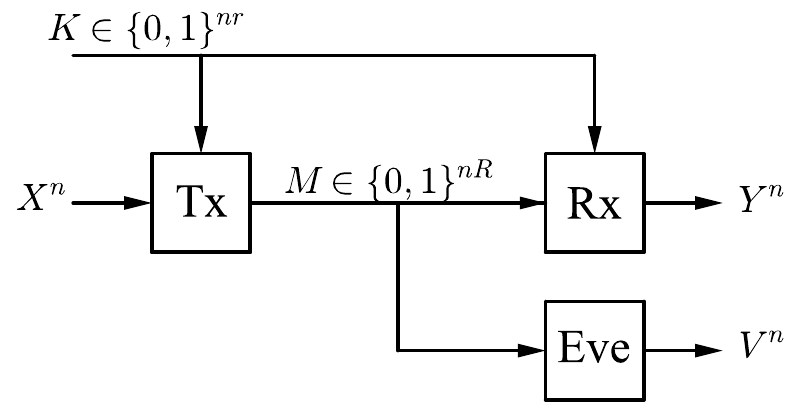}
\caption{The Shannon cipher system.}\label{figureShannonCipher1}
\end{figure}
\item An adversary could ``wiretap'' a communication channel to intercept transmissions. The wiretap channel is typically noisier than the main channel, but its output nevertheless depends on the transmitted message~\cite{WireTapChannel,WiretapGaussian}.
\item Suppose one would like to anonymously transmit a message through a given network (say, a call for protest on a social network). A powerful adversary (say, a government) could learn the spread of the message (i.e., who received it), which is correlated with the identity of its author~\cite{FantiRumor,RumorCulprit}.  
\item  Cellular networks track the locations of its users in order to route calls. Such tracking data might reveal private information of the user (such as their political affiliation, their place of work, etc.)~\cite{WickerLocationPrivacy,KolbeinnSemanticLocation}.
\end{itemize} 
Although at first glance such side-channels may seem innocuous, many works have shown that they pose a significant security threat~\cite{SSHTiming,PeepingTom,TimingSuh,SharedSchedulers,RSAbreak,PowerKocher,ICDoD,koziel2012effects}. For instance, Zhang and Wang~\cite{PeepingTom} show how to use the keystroke timing in SSH to reduce the search space for passwords by a factor of at least 250. Kocher~\cite{RSAbreak} shows how to break implementations of RSA using timing information.
 Ristenpart \emph{et al.}~\cite{amazonattack} show how
 secret keys can be extracted from co-resident virtual machines on production Amazon EC2 servers through microarchitectural timing channels. \\
 
 Addressing such threats first requires an answer to the question posed at the outset. That is, if $X$ is a random variable representing sensitive information and $Y$ is the output of a side-channel with input $X$,
 { \begin{center}
How much information does $Y$ leak about $X$?
\end{center} }
\noindent Let $\mlt{X}{Y}$ denote a potential answer. Before discussing existing approaches, we posit that a good choice of $\mlt{X}{Y}$ should satisfy the following requirements:
\begin{enumerate}
\item[(R1)] It should have a cogent \emph{operational} interpretation. That is, a system designer should be able to explain what guarantees on the system  an upper bound on $\mlt{X}{Y}$ provides. In the context of side channels,
   the design goal is typically to prevent the adversary from guessing 
   sensitive, discrete quantites such as keys and passwords. Thus the 
    leakage measure should be interpretable in terms of the adversary's 
    difficulty in guessing such quantities.
\item[(R2)] Assumptions about the adversary should be minimal (since guarantees are void if any assumption does not hold true). Indeed, one would like to take into account a large family of potential adversaries.
\item[(R3)] It should satisfy axiomatic properties of an information measure:
\begin{enumerate}
\item The data processing inequality: $\mlt{X}{Y} \leq \min\{\mlt{X}{Z},\mlt{Y}{Z}\}$ if $X-Y-Z$ is a Markov chain.
\item The independence property: $\mlt{X}{Y}=0$ if and only if $X$ and $Y$ are independent.
\item The additivity property: if $(X_1,Y_1)$ and $(X_2,Y_2)$ are independent, then $\mlt{X_1^2}{Y_1^2} = \mlt{X_1}{Y_1} + \mlt{X_2}{Y_2}$.
\end{enumerate}
\item[(R4)] It should accord with intuition. That is, it should not mis-characterize the (severity of) information leakage in systems that we understand well.
\end{enumerate}

\subsection{Common Information-Theoretic Approaches}
Notably, many commonly-used information leakage metrics do not satisfy the above requirements. For example, mutual information, which has been frequently used as a leakage measure~\cite{ShannonSecrecy,WiretapGaussian,WireTapChannel}\!\! \cite{SecrecyFading, csiszar1978broadcast,erkippoor2008lossless, PoorUtilityPrivacy},
arguably fails to satisfy both (R1) and (R4). Regarding the latter,
consider the following example proposed by Smith~\cite{SmithQuantInfoFlow}.
\vspace{1mm}
\begin{Example}
\label{Example:Guessing}
Given $n \in \mathbb{N}$, let $\mathcal{X}= \{0,1\}^{8n}$ and $X \sim \mathrm{Unif}(\mathcal{X})$. Now consider the following two conditional distributions:
\begin{align*}
P_{Y|X} = \begin{cases}
X, & \text{if } X ~\mathrm{mod}~ 8 = 0, \\
1, & \text{otherwise}.
\end{cases} \quad \text{ and } \quad
P_{Z|X}= (X_1, X_2, \dots, X_{n+1}). 
\end{align*}
Then the probability of guessing $X$ correctly from $Y$ is at least $1/8$, whereas the probability of guessing $X$ correctly from $Z$ is only $2^{-7n+1}$ for $Z$. However, one can readily verify that $I(X;Y) \approx (n+0.169)\log 2 \leq I(X;Z) = (n+1)\log 2$~\cite{SmithQuantInfoFlow}. 
\end{Example}
\vspace{1mm}

Regarding the former, note that operational interpretations of mutual
information arise in transmission and compression settings, which are different from the security setting at hand. Moreover, in those settings, mutual information arises as part of a computable characterization of the solution, rather than as part of the formulation itself, i.e., the transmission and compression problems are not \emph{defined} in terms of mutual information. 

Mutual information could potentially be justified by appealing to 
rate-distortion theory~\cite[Section~V]{cuff2013secrecycausal}. 
In fact, a number of leakage
measures in the literature are based on rate-distortion theory.
For instance, Yamomoto~\cite{yamamoto1997shannoncipher} introduces a distortion function $d$ and measures the privacy of $P_{Y|X}$ using $\inf_{ \hat{x}(\cdot)} \E[d(X, \hat{x}(Y))]$. Schieler and Cuff~\cite{cuff2013secrecycausal} discuss (and generalize) an example that shows the inadequacy of this approach, if
conventional distortion measures such as Hamming distortion are used.
\vspace{1mm}
\begin{Example}
\label{Example:Hamming}
Given $n \in \mathbb{N}$, let $X^n$ be i.i.d $\sim \mathrm{Ber}(1/2)$ and let $K \sim \mathrm{Ber}(1/2)$ be independent of $X^n$. Suppose $d$ is the Hamming distortion and let $P_{Y|X^n}$ be as follows: if $K=0$, $Y=X^n$; otherwise, $Y=\bar{X}^n$ (i.e., flip all the bits of $X^n$). Then $\inf_{ \hat{x}(\cdot)} \E[d(X, \hat{x}(Y))]=1/2$, which is the maximum distortion the adversary could incur. The proposed scheme is hence optimal from an expected distortion point of view. However, by observing $Y$, the adversary can guess $X^n$ with probability $1/2$. Moreover, they can guess it exactly with two attempts.
\end{Example}
\vspace{1mm}
Similarly to expected distortion, the expected number of guesses to find $X^n$ fails on the fourth requirement: it can label obviously insecure systems as secure (see~\cite{MySecrecySystem} for an example). More generally, rate-distortion-based approaches do not meet the second requirement (R2): they assume there is a \emph{known}  distortion function, and in some cases a particular distortion level, up to which the adversary is interested in reproducing the sensitive information $X$.

\subsection{Contributions}
We introduce a new metric, \emph{maximal leakage}, that meets all the above requirements. To do so, we first describe a threat model that captures the side-channel setup.

{\bf Threat model.} We assume the adversary is interested in a (possibly randomized) function of $X$, called $U$. We restrict $U$ to be discrete, which captures most scenarios of interest (in the side-channel examples above, all functions of interest are discrete, e.g., a password, a message, an identity, etc). However, $P_{U|X}$ is unknown to the system designer. This models the case in which we do not know the adversary's function of interest, and wish to account for a large family of potential adversaries, as in the second requirement of $\mlt{X}{Y}$ above. Even if it is known, $P_{U|X}$ could be so complicated that it might as well be unknown. The adversary observes a random variable $Y$, and the Markov chain $U-X-Y$ holds. They wish to \emph{guess} $U$ and can verify if their guess is correct (if, say, $U$ is a password for a given system, then they can attempt to log in using it). Hence, they would like to maximize the probability of guessing $U$ correctly. Finally, we assume the system designer accepts low risks (i.e., a random event that reveals $U$ is tolerable as long as it has very low probability), or that the leakage is concentrated with respect to $Y$ (i.e., we can average over $P_Y$, which is the case in side-channels where the input and output are running processes).

{\bf Operational, robust measure of information leakage.} We now define maximal leakage, which we denote by $\ml{X}{Y}$, as the (logarithm of the) ratio of the probability of correctly guessing $U$ from $Y$ to the probability of correctly guessing $U$ with no observation, maximized over all $U$ satisfying $U-X-Y$ (cf.~Definition~\ref{defmaxleakage}). The maximization over $U$ guarantees that our definition satisfies the requirement (R2) of making minimal assumptions about the adversary. Moreover, the operational meaning of this quantity is clear: a leak of $\ell$ bits means that for any $U$, the multiplicative increase (upon observing $Y$) in the correct guessing probability is upper-bounded by, but can be arbitrary close to, $2^\ell$.    

So defined, it is not clear \emph{a priori} that maximal leakage is
computable, since it requires maximizing over all auxiliary random
variables $U$. A standard approach to obtaining a computable characterization
in such problems is to bound the necessary alphabet size for $U$ 
in terms of the alphabet size of $X$ using Carath\'{e}odory's 
theorem~(e.g.,~\cite[Lemma~5.4]{korner}). This technique fails for the present problem,
however: even a binary $X$ can require arbitrarily large $U$ in order to
approach the supremum. Nonetheless, 
Theorem~\ref{thmmainthm} provides a simple formula for maximal leakage for the case of discrete $X$ and $Y$. In particular, it shows that $\ml{X}{Y}$ is equal to the Sibson mutual information of order infinity $I_\infty(X;Y)$~\cite{SibsonInfo,Verdualpha}. 
Consequently, maximal leakage meets our third requirement (R3) of satisfying axiomatic properties of an information measure. That is, it is zero if and only if $X$ and $Y$ are independent; it satisfies the data processing inequality; and it is additive over independent pairs $\{(X_i,Y_i)\}_{i=1}^{n}$. Interestingly, it is lower-bounded by $I(X;Y)$, indicating that mutual information underestimates leakage.

Moreover, the definition of maximal leakage is shown to be 
\emph{robust:} the result is unaffected if the adversary picks the function of interest $U$ only after observing $Y$ (cf.~Theorem~\ref{thmmaxleakY}), if they only wish to approximate $U$ (cf.~Theorem~\ref{lemmaconteqdisc}),
if they are allowed several guesses (cf.~Theorem~\ref{lemmakeq1}), or if they wish to maximize some arbitrary gain function (cf.~ Theorem~\ref{thmmaxgains}). We also extend the notion of maximal leakage in two directions. We propose a conditional form of maximal leakage, which attempts to answer the question: how much does $Y$ leak about $X$ when $Z$ is given? Here $Z$ represents side information that is available to the adversary. 
We again provide an operational definition in the guessing framework (cf.~Definition~\ref{defcondmaxleakage}), and derive a simple form for $\cml{X}{Y}{Z}$ (cf.~Theorem~\ref{thmcondmaxleakage}). Moreover, we generalize the computable characterization for maximal leakage to cover a large class of random variables and stochastic processes (cf.~Theorem~\ref{thmgeneralformula}). Both the general and the conditional form retain the axiomatic properties of a leakage measure, and are lower-bounded by mutual information and conditional mutual information, respectively.

{\bf New insights for mechanism design.} The new metric is useful to develop new mechanisms to mitigate leakage, as well as to evaluate existing mechanisms for this purpose. A common approach in such designs is to add independent noise to successive inputs of the system. For example, in the SSH scenario, packets could be passed through an $./M/1$ queue before being sent over the network. We provide examples of such memoryless schemes to show that, roughly speaking, they do not perform well under maximal leakage, and are outperformed by quantization-based schemes. More concretely, we consider the Shannon cipher system with lossy communication and evaluate the performance of an encryption scheme using maximal leakage between the source and the public message (other works have considered this setup under different metrics\cite{MySecrecySystem,yamamoto1997shannoncipher,cuff2014henchman,
merhavsecurelossy}). For a discrete memoryless source, we show that memoryless schemes are strictly suboptimal (cf.~Lemma~\ref{LemmaMemoryless}), whereas optimal schemes correspond to good rate-distortion codes. Moreover, we derive a single-letter characterization of the optimal (normalized) limit (cf.~Theorems~\ref{Thmmain}
and~\ref{Thmvary}).   

{\bf Complexity of estimating maximal leakage.} The computation of maximal leakage might become intractable for complicated schemes. For example, in the setup of multiple processes running on the same chip as described above, what determines the information leakage between processes is the memory controller, the operation of which might depend on many variables. Thus, we consider the problem of estimating maximal leakage from data. We show that this task is feasible only if we know (a lower bound on) the minimum strictly positive probability of a symbol $x \in \mathcal{X}$, denoted by $\theta$.	More specifically, we show that the number of samples needed to estimate $\ml{X}{Y}$ up to $\epsilon-$additive-accuracy is $ \Omega \left( |\mathcal{Y}|/(\theta \log |\mathcal{Y}|) \right)$ (cf.~Theorem~\ref{thmlowerbound}). Note that the lower bound diverges to infinity as $\theta$ tends to zero. On the other hand, we show that $ \mathcal{O} \left( \frac{|\mathcal{Y}| \log |\mathcal{X}|} {\theta} \right)$ samples are sufficient (cf.~Theorem~\ref{thmupperbound}). This suggests that we should take into account amenability to analysis while designing leakage-mitigating mechanisms. 

{\bf Guessing framework to interpret leakage measures.} Finally, we use the guessing framework used to define maximal leakage to give new operational interpretations for different information leakage measures. This provides a common framework with which to compare them, and elucidates in which setups each should be used. More specifically, we study the following commonly used metrics: Shannon capacity, local differential privacy~\cite{LocalDiffPrivacy}, and maximal correlation~\cite{MaxCorrSecrecy}.

We show that 
\begin{enumerate}
\item Shannon capacity captures the multiplicative increase of the probability of correct guessing over the restricted set of functions of $X$ that can be \emph{reliably} reconstructed from $Y$, hence underestimating leakage (cf.~Theorem~\ref{thmcapacitymaxleak});
\item Local differential privacy captures the multiplicative increase of the guessing probability of functions of randomized $X$, maximized over realizations of $Y$ and over distributions $P_X$ (cf.~Theorem~\ref{thmLDP});
Moreover, maximizing over realizations of $Y$ for a fixed $P_X$ yields the maximum information rate (cf.~Theorem~\ref{thmmaxleakrealizable});
\item Maximal correlation captures the multiplicative change in the variance of functions of $X$, rather than the guessing probability (cf.~Theorem~\ref{thmvariance}). We extend this last notion to a new measure we call \emph{maximal cost leakage} (cf.~Definition~\ref{defcostleakage}), which captures the worst-case multiplicative reduction over all cost functions defined on any hidden variable $U$.
\end{enumerate}

\subsection{Related Work}
Calmon \emph{et al.}~\cite{CalmonBoundsInference} and Li and El Gamal~\cite{MaxCorrSecrecy} use maximal correlation, $\rho_m(X;Y)$, as a secrecy measure (Calmon \emph{et al.} also generalize it to $k$-correlation, which is defined as the sum of the $k$ largest principal inertial components of the joint distribution $P_{XY}$). A key motivating result~\cite[Theorem 9]{CalmonPIC} shows that maximal correlation bounds the \emph{additive} increase in the correct guessing probability of any deterministic function of $X$. Although $\rho_m(X;Y)$ is zero only if $X$ and $Y$ are independent, the correct guessing probability of any deterministic function might be unchanged even if $X$ and $Y$ are not independent,
as illustrated in the following example.
\vspace{1mm}
\begin{Example} \label{exampledeterministicfunc}
Suppose $P_{XY}$ satisfies the following condition: there exists $x^\star \in \mathcal{X}$ such that for all $y \in \mathcal{Y}$, $P_{X|Y}(x^\star|y) \geq 1/2$. Then for any deterministic function $f$, $f(x^\star)$ is the adversary's best guess for $f(X)$, both with and without the observation of $Y$. Hence, observing $Y$ does not affect the probability of guessing any deterministic function of $X$. Note, however, that $X$ and $Y$ may be dependent.
\end{Example} 
\vspace{1mm}

The literature on leakage and privacy measures extends beyond information theory to computer security and computer science more generally. The closest to our work in fact comes from computer security~\cite{SmithQuantInfoFlow,CompSecQuantLeakage,LeakageGeneralized,LeakageMinEntropy
,AlvimAdditiveMultiplicative}. In particular, Smith~\cite{SmithQuantInfoFlow} defines leakage from $X$ to $Y$ as the logarithm of the multiplicative increase, upon observing $Y$, of the probability of guessing $X$ \emph{itself} correctly, neglecting that the adversary might be interested in certain functions of $X$. Braun \emph{et al.}~\cite{CompSecQuantLeakage} consider a worst case modification of the metric, and maximize the previous quantity over all distributions on the alphabet of $X$ (while $P_{Y|X}$ is fixed). The resulting quantity turns out to  equal $\ml{X}{Y}$---it is called ``maximal leakage'' in the computer security literature as well. It is denoted by $ML(P_{Y|X})$, and its properties were further studied by Espinoza and Smith~\cite{LeakageMinEntropy} and Alvim \emph{et al.}~\cite{LeakageGeneralized}. The latter also define $g$-leakage by introducing a gain function $g: \mathcal{X} \times \mathcal{\hat{X}} \rightarrow [0,1]$ and considering the normalized maximal gain (for $g$). Alvim \emph{et al.}~\cite{AlvimAdditiveMultiplicative} consider several variants of $g$-leakage (i.e., additive or multiplicative increase, fixing or maximizing over the marginal $P_X$, etc). They show that maximizing $g$-leakage over gain functions $g$ yields maximal leakage. However, no operational significance is attached to the $g$ that achieves the maximum. Moreover, the result is given only as one of many possible computable variations of leakage~\cite{AlvimAdditiveMultiplicative,SmithDevelopment}.

Another connected line of work stems from cryptography, and in particular from the notion of \emph{semantic security}~\cite{SemanticSec} which considers the security of encryption schemes. Goldwasser and Micali~\cite{SemanticSec} define the ``advantage'' for a given  function of the messages as the additive increase of the correct guessing probability
upon observing the encrypted message (i.e., the ciphertext).  
Semantic security then requires that, for an adversary that can work only for a polynomial (in the length of the message) amount of time, the advantage is negligible for all input distributions and for all deterministic functions that are computable in polynomial time.

There are several variants of semantic security. In particular, entropic security~\cite{EntropicSec,EntropicSec2} drops the computational bounds (on the adversary and the considered functions), but restricts its attention to input distributions with high min-entropy. Bellare \emph{et al.}~\cite{WiretapSemantic} introduce semantic security to the wiretap channel, and do not restrict it to computationally-bounded adversaries or to deterministic polynomial-time computable functions. For a given encryption scheme, they  then upper and lower-bound the advantage of semantic security in terms of ``mutual information security advantage'', which is defined as the maximum, over all input distributions, of the mutual information between the message and the output of the channel whose input is the encryption of the message. 
For further discussion of leakage metrics, we refer the reader to Wagner and Eckhoff's work\cite{Eighty}, which categorizes over eighty such metrics.

\subsection{Outline}
We describe our threat model and define maximal leakage in Section~\ref{sec:main}. We also give a closed-form expression of maximal leakage (for discrete $X$ and $Y$), discuss its properties, and compare it to related leakage metrics. In Section~\ref{sec:variations}, we prove the robustness of our definition by considering several variations on the setup, and show that they all lead to the same quantity. Furthermore, we generalize the formula of maximal leakage and analyze a simple model of the SSH side-channel. We also present a conditional form of maximal leakage. In Section~\ref{sec:cipher}, we consider the Shannon cipher system and derive (asymptotically) optimal schemes. We show that memoryless schemes are strictly suboptimal in general. We study the complexity of estimating maximal leakage from data in Section~\ref{sec:learning}. Finally, in Section~\ref{sec:commonframe}, we use the guessing framework to give new operational interpretations for common information leakage metrics, and we introduce a cost-based notion of leakage.

\section{Maximal Leakage} \label{sec:main}

Let $X$ be a random variable representing sensitive information, and $Y$ be the output of a side-channel the input of which is $X$. To give an operational definition of information leakage between $X$ and $Y$, we specify a threat model as follows. 
\begin{itemize}
\item The adversary is interested in a possibly randomized, discrete function of $X$ called $U$.
\item The adversary observes $Y$ and the Markov chain $U-X-Y$ holds.
\item The adversary wishes to \emph{guess} $U$ and can verify if the guess is correct.
\item The distribution $P_{U|X}$ is unknown to the system designer.
\item From the system designer's viewpoint, if the probability of guessing $U$ correctly is high for some realizations of $Y$, then it suffices that the probability of such realizations is suitably small.
\end{itemize} 
To clarify our model, consider how it applies to the SSH side-channel. In this case, $X$ represents the nominal packet timings. Suppose we perturb the packet timings before sending them over the network, in which case $Y$ represents the post-perturbation timings observed by the adversary. $U$ corresponds to the input of the user, e.g., their password. The adversary wishes to guess $U$ (e.g., the password) and can verify their guess (e.g., by attempting to log into the system). So they wish to maximize the probability that the guess is correct, and the system designer wishes to minimize it. The distribution of passwords given packet timings is complicated, so we assume it is unknown to the system designer. Finally, the system designer only requires the probability that the system is compromised to be small. 

One might be tempted to restrict the range of $U$ to deterministic functions of $X$. However, this is too restrictive as implied by Example~\ref{exampledeterministicfunc} in the introduction. Moreover, in most side-channel examples we mentioned, the $U$ of interest is a randomized function of $X$ (passwords given packet timings, key values given memory access patterns, political affiliation given location traces, etc). On the other hand, the restriction to discrete $U$'s still captures most scenarios of interest, as in the above examples. Indeed even when $X$ represents location traces, for instance, the $U$ of interest is typically discrete, e.g., home/work address, political affiliation, etc. Finally, assuming $P_{U|X}$ is unknown allows us to take into account a wide range of adversaries having different objectives. That is, as in our requirement (R2), we do not assume we know the function of interest to the adversary.   

We are now ready to present the definition of maximal leakage. Since the adversary wishes to guess $U$, we consider the maximum advantage in the probability of guessing $U$ from $Y$, as compared with guessing with no observations. Maximal leakage captures the maximum advantage over all $U$'s as in the following definition.

\vspace{1mm}
\begin{Definition}[Maximal Leakage] \label{defmaxleakage}
Given a joint distribution $P_{XY}$ on alphabets $\mathcal{X}$ and $\mathcal{Y}$, the \emph{maximal leakage} from $X$ to $Y$ is defined as
\begin{equation} 
\label{eq:maxleakagedef}
\ml{X}{Y} = \sup_{U - X - Y - \hat{U}} \log \frac{\Pr{\left(U=\hat{U}\right)}}{\max_{u \in \mathcal{U}} P_U(u)},
\end{equation}
where the supremum is over all $U$ and $\hat{U}$ taking values in the same finite, but arbitrary, alphabet.
\end{Definition}
\vspace{1mm}
\begin{Remark}
$\log$ is the natural logarithm so $\ml{X}{Y}$ is in nats. Using $\log_2$ instead gives an answer in bits.
\end{Remark}
\vspace{1mm}
The guarantee that a small leakage provides is as follows. Whatever function $U$ the adversary is interested in, if $\ml{X}{Y} \leq \ell$, then $\sup_{\hat{u}(\cdot)} \Pr \left( U = \hat{u} (Y) \right) \leq e^{\ell} \max_{u} P_U(u)$. Note that the upper bound can be decomposed into two quantities: $\max_{u} P_U(u)$ which is completely outside the control of the system designer, and $e^{\ell}$ which is determined by the designer's choice of $P_{Y|X}$ (which is typically subject to quality constraints related to the performance of the underlying system).  
Moreover, the definition directly implies several important properties of maximal leakage.
\vspace{1mm}
\begin{Lemma} \label{lemmadefprop}
For any joint distribution $P_{XY}$ on alphabets $\mathcal{X}$ and $\mathcal{Y}$,
\begin{enumerate}
\item (\emph{Data Processing Inequality}) If the Markov chain $X-Y-Z$ holds,  $\ml{X}{Z} \leq \min\{\ml{X}{Y}, \ml{Y}{Z} \}$.
\item If $Y$ is discrete, $\ml{X}{Y} \leq \log |\mathrm{supp}(Y)|$. 
\item If $X$ is discrete, $\ml{X}{Y} \leq \log |\mathrm{supp}(X)|$. 
\item $\ml{X}{Y} \geq 0$ with equality if $X$ and $Y$ are independent.
\end{enumerate}
\end{Lemma}
The proof is given in Appendix~\ref{app:lemmadefproof}. Note that properties 1) and 4) were two of our axioms for a leakage measure (R3). Properties 2) and 3) are consistent with intuitive understanding of information. In particular,  a binary variable $Y$ cannot leak more than one bit about any variable $X$. Similarly, a binary variable $X$ has no more than one bit of information to be leaked.

Despite the useful properties of the definition, it involves an infinite-dimensional optimization problem, so it is not clear \emph{a priori} that it is computable. In fact, one can show that it is impossible to bound the cardinality of the alphabet $\mathcal{U}$ in terms of the cardinalities of the alphabets $\mathcal{X}$ and $\mathcal{Y}$.
Nonetheless, we can show that maximal leakage is indeed computable and actually takes a simple form. We focus first on the discrete case and consider general alphabets later.
\vspace{1mm}
\begin{Theorem} \label{thmmainthm}
For any joint distribution $P_{XY}$ on finite alphabets $\mathcal{X}$ and $\mathcal{Y}$, the maximal leakage from $X$ to $Y$ is given by the Sibson mutual information of order infinity, $I_\infty(X;Y)$. That is,
\begin{equation*}
\ml{X}{Y} = \log  \sum_{y \in \mathcal{Y}} \max_{\substack{x \in \mathcal{X}:\\ P_X(x) >0 }} P_{Y|X}(y|x) = I_\infty(X;Y).
\end{equation*}
\end{Theorem} 
\vspace{1mm}
\begin{Remark} 
Sibson's mutual information~\cite{SibsonInfo, Verdualpha} of 
order $\alpha$ ($\alpha \geq 0$, $\alpha \ne 1$),
which can be expressed (in the discrete case) as
\begin{align}
\label{eq:SibsonMI}
I_{\alpha}(X;Y) & = \inf_{Q_Y} D_{\alpha} (P_{XY} || P_X \times Q_Y)
\end{align}
where
\begin{equation}
D_\alpha(P||Q)   = \frac{1}{\alpha - 1} \log
                  \left( \sum_{a} P^\alpha(a)
                      Q^{1 - \alpha}(a)\right),
\end{equation}
is one of several suggested extensions of the concept of Renyi entropy $H_\alpha (X)$ (itself an extension of entropy) and Renyi divergence $D_\alpha (P||Q)$. 
Verd\'{u}~\cite{Verdualpha} argues for the adoption of Sibson's extension, and the above result supports that choice by providing an operational interpretation of 
\begin{align}
\label{eq:SibsonMIinf}
I_{\infty}(X;Y) & = \lim_{\alpha \rightarrow \infty} I_\alpha(X;Y) \\
\label{eq:SibsonMIinfDinf}
                & = \inf_{Q_Y} D_\infty (P_{XY}|| P_X \times Q_Y),
\end{align}
where
\begin{align}
\label{eq:RenyiDivInf}
D_\infty(P||Q) & = \lim_{\alpha \rightarrow \infty} 
                      D_\alpha (P||Q) \\
               & = \log \left(\sup_{a}  \frac{P(a)}{Q(a)} \right)
\end{align}
and the interchange of the limit and infimum in~(\ref{eq:SibsonMIinf}) and~(\ref{eq:SibsonMIinfDinf}) is implied by~\cite[Theorem~4]{Hoconvexconcave} (see also (\ref{eq:optQY}) to follow).
\end{Remark}
\vspace{1mm}
Before proving the theorem, we investigate some of its consequences. 
First, it reveals two of the more useful aspects of maximal leakage from an engineering perspective: minimizing $\mathcal{L}({X \!\! \to \!\! Y})$ over $P_{Y|X}$, for a fixed support of $P_X$, amounts to minimizing a convex function, and $\mathcal{L}({X \!\! \to \!\! Y})$ depends on $P_X$ only through its support. The latter fact is very useful because in practice $P_X$ is typically complicated and outside our control. $P_X$ is also typically used to model the adversary's prior knowledge of $X$, which is not necessarily known to us. The following corollary (the proof of which is given in Appendix~\ref{app:corrprop}) summarizes some useful properties of $\ml{X}{Y}$.
\vspace{1mm} 
\begin{Corollary} \label{corrprop}
For any joint distribution $P_{XY}$ on finite alphabets $\mathcal{X}$ and $\mathcal{Y}$, 
\begin{enumerate}
\item $\ml{X}{Y} = 0$ iff $X$ and $Y$ are independent.
\item (\emph{Additivity}) 
If $\{(X_i,Y_i)\}_{i=1}^{n}$ are mutually independent, then
\begin{equation*}
\ml{X^n}{Y^n} = \sum_{i=1}^n \ml{X_i}{Y_i}.
\end{equation*}
\item $\ml{X}{Y} = \log |\mathcal{X}|$ iff $X$ is a deterministic function of $Y$ (assuming $X$ has full support).
\item $\ml{X}{X} = H_0(X) =\log |\mathrm{supp}(X)|$.
\item $\ml{X}{Y}$ is not symmetric in $X$ and $Y$.
\item $\exp\{\ml{X}{Y}\}$ is convex in $P_{Y|X}$ for fixed support of $P_X$.
\item $\ml{X}{Y}$ is concave in $P_X$ for fixed $P_{Y|X}$.
\end{enumerate}
\end{Corollary}
\vspace{1mm}

Note that properties 1) and  2) along with the data processing inequality are the axioms we stated in (R3). Property 5) reveals a potential ``weakness'' in some suggested leakage metrics, including mutual information. In particular, there is no reason to expect a priori that $X$ leaks about $Y$ as much as $Y$ leaks about $X$. Therefore, metrics that are symmetric by design miss that fact (this is in contrast with R{\'e}nyi's axiom that a \emph{dependence} measure should be symmetric~\cite{RenyiMaxCorr}).  Finally, 
property 6) shows that minimizing maximal leakage, for a fixed support of $P_X$, amounts to minimizing a convex function. That is, one can efficiently solve the problem of finding the randomization mechanism $P_{Y|X}$ that minimizes maximal leakage, subject to a convex constraint. \\

We evaluate $\ml{X}{Y}$ for some special cases.
\begin{Example} \label{exampleBSC}
If $X \sim \text{Ber}(q)$, $0 < q <1$, and $Y$ is the output of a BSC with parameter $p$, $0 \leq p \leq 1/2$, then $\ml{X}{Y} = \log(2(1-p))$. 
\end{Example}
\vspace{1mm}
\begin{Example} \label{exampleasy1}
If $X \sim \text{Ber}(q)$, $0 < q <1$, and $Y$ is the output of a BEC with parameter $\epsilon$, $0 \leq \epsilon < 1$, then $\ml{X}{Y} = \log(2-\epsilon)$, and $\mathcal{L}(Y \!\! \to \!\! X)=\log 2$. 
\end{Example}
\vspace{1mm}
\begin{Example} \label{exampledeterministic}
For any deterministic law $P_{Y|X}$, ${\ml{X}{Y} = \log | \mathrm{supp}(Y)| }$.
\end{Example}
\vspace{1mm}
Consider the examples from the introduction that showed that expected distortion and mutual information do not meet our fourth requirement (R4).
\vspace{1mm}
\begin{Example}[cf.~Example~\ref{Example:Hamming}]
Given $n \in \mathbb{N}$, let $X^n$ be i.i.d $\sim \mathrm{Ber}(1/2)$ and let $K \sim \mathrm{Ber}(1/2)$ be independent of $X^n$. Let $P_{Y|X^n}$ be as follows: if $K=0$, $Y=X^n$; otherwise, $Y=\bar{X}^n$ (i.e., flip all the bits of $X^n$). This scheme is optimal from an expected Hamming distortion viewpoint. On the other hand, $\ml{X^n}{Y^n} = (n-1)\log  2$, which is exactly describing that we know $X^n$ up to 1 bit.
\end{Example}
\vspace{1mm}
\begin{Example}[cf.~Example~\ref{Example:Guessing}]
Given $n \in \mathbb{N}$, let $\mathcal{X}= \{0,1\}^{8n}$ and $X \sim \mathrm{Unif}(\mathcal{X})$. Now consider the following two conditional distributions:
\begin{align*}
P_{Y|X} = \begin{cases}
X, & \text{if } X ~\mathrm{mod}~ 8 = 0, \\
1, & \text{otherwise}.
\end{cases} \quad \text{ and } \quad
P_{Z|X}= (X_1, X_2, \dots, X_{n+1}). 
\end{align*}
Then $\ml{X}{Y} = \log(2^{8n-3}+1) > \ml{X}{Z} = (n+1) \log 2$, whereas $I(X;Y) \approx (n+0.169)\log 2 < I(X;Z) = (n+1)\log 2$.
\end{Example}
\vspace{1mm}
In the next section, we elaborate on the comparison between mutual information and maximal leakage. We also comment on the relation to the computer security and computer science literature, before proving Theorem~\ref{thmmainthm} in Section~\ref{sec:proofmain}.

\subsection{Comparison with Related Metrics}

\subsubsection{Mutual Information}  
We first compare maximal leakage with mutual information in the following lemma. It shows that $\ml{X}{Y}$ upper-bounds $I(X;Y)$, and no scalar multiple
of $I(X;Y)$ can upper-bound $\ml{X}{Y}$.
\vspace{1mm}
\begin{Lemma} \label{lemmacompareMI}
For any joint distribution $P_{XY}$ on finite alphabets $\mathcal{X}$ and $\mathcal{Y}$, 
$\ml{X}{Y} \geq I(X;Y)$. Moreover, for any $c > 0$, there exists $P_{XY}$ such that $\ml{X}{Y} \geq c I(X;Y)$. Furthermore. 
 $ \ml{X}{Y} = I(X;Y)$ if and only if
\begin{enumerate}
\item If $P_{XY}(x,y) > 0$ and $P_{XY}(x',y) > 0$, then $P_{Y|X}(y|x) = P_{Y|X}(y|x')$.
\item For all $y,y' \in \text{supp}(Y)$,
\begin{align*}
\sum_{x: P_{XY}(x,y) >0} P_X(x) = \sum_{x': P_{XY}(x',y')>0} P_X(x').
\end{align*}
\end{enumerate} 
\end{Lemma}
\begin{Remark} \label{remarksing} A joint distribution satisfying condition 1) is called \emph{singular}~\cite{YucelRCBound}. Moreover, if $X$ has full support, $\ml{X}{Y}= I(X;Y) \Rightarrow \ml{X}{Y} = C(P_{Y|X}).$
\end{Remark}
\vspace{1mm}
\begin{proof}
That $I_\infty(X;Y) \geq I(X;Y)$ is already known~\cite{SibsonInfo,Verdualpha}. For the stronger statement, it suffices to show it for binary $\mathcal{X}$ and $\mathcal{Y}$. To that end, let $X \sim \mathrm{Ber}(1/2)$ and let $P_{Y|X}$ be a BSC with parameter $p \in (0,1/2)$. Then $\ml{X}{Y} = \log(2(1-p))$ and $I(X;Y)= \log 2 - H(p)$ (where the entropy function is computed using the natural logarithm). One can readily verify that
\begin{align*}
\lim_{p \rightarrow 1/2} \frac{\log(2(1-p))}{\log 2 - H(p)} = + \infty.
\end{align*}
The conditions for equality can be readily verified, and are included in Appendix~\ref{app:lemmaMIMLproof} for completeness.
\end{proof}

The lemma shows that a small maximal leakage is a more stringent requirement than a small mutual information. Since $\ml{X}{Y} $ depends on $P_X$ only through its support, it follows that maximal leakage is at least the Shannon capacity of the channel $P_{Y|X}$ when $X$ has full support, and this inequality can be strict (as in the BSC example in the proof of the lemma).
This justifies the claim in the introduction that the Shannon
capacity of a side-channel does not necessarily upper-bound its leakage.
The maximization in the definition of maximal leakage hints at the reason 
why. In particular, Shannon capacity is concerned with (the size of) message sets that can be \emph{reliably} reconstructed at the receiver, i.e., $\Pr(U=\hat{U}(Y)) \geq 1 - \epsilon$ for some small $\epsilon$. Leakage, on the other hand, is concerned with the advantage in guessing, without any notion of reliability. This observation is made mathematically precise in Section~\ref{subsec:compareCapacity}. On the other hand, local differential privacy~\cite{LocalDiffPrivacy}, which some regard as too pessimistic~(e.g.,~\cite{LocalDiffPrivacy}), does upper-bound maximal leakage. This is further explored in Section~\ref{subsec:LDP}. 

\subsubsection{g-leakage} For discrete $X$ and $Y$, Braun \emph{et al.}~\cite{CompSecQuantLeakage} define leakage as follows
\begin{align} \label{eq:BraunML}
ML(P_{Y|X}) = \sup_{P_X} \frac{\sup_{X-Y-\hat{X}} \Pr(X=\hat{X}) }{\max_{x \in \mathcal{X}} P_X(x)}.
\end{align}
This definition assumes the adversary wishes to guess $X$ itself, and hence does not meet our second requirement (R2). However, it is equal to $\ml{X}{Y}$ when $X$ has full support. Alvim \emph{et al.}~\cite{LeakageGeneralized} define $g$-leakage by introducing a gain function $g: \mathcal{X} \times \mathcal{\hat{X}} \rightarrow [0,1]$, where $\hat{\mathcal{X}}$ is a finite set. 
Then
\begin{align} \label{eq:gleakage}
\mathcal{ML}_g (P_X, P_{Y|X}) = \frac{ \sup_{X-Y-\hat{X}} \E[g(X,\hat{X})] }{\max_{\hat{x} \in \hat{\mathcal{X}}} \E[g(X,\hat{x})] }.
\end{align} It is shown~\cite{AlvimAdditiveMultiplicative} that
\begin{align} \label{eq:gleakmax}
\sup_{\hat{\mathcal{X}}, g: \mathcal{X} \times \mathcal{\hat{X}} \rightarrow [0,1]} \mathcal{ML}_g (P_X, P_{Y|X}) = \ml{X}{Y}.
\end{align}
This definition however does not explicitly account for random functions of $X$, whereas we have seen that in many cases the adversary could be interested in a hidden random variable $U$. Moreover, there is no operational interpretation attached to the $g$ that achieves the maximum. Nonetheless, these metrics
are quite similar to the one introduced here.

\subsubsection{Semantic Security} The semantic and entropic security literature~\cite{SemanticSec,EntropicSec,EntropicSec2} consider the difference between the guessing probabilities as opposed to the ratio considered in Definition~\ref{defmaxleakage}. Since $U$'s of interest, such as passwords, are typically hard to guess (i.e., $\max_u P_U(u)$ is small), the ratio is arguably the more appropriate measure of the change. It is also the more natural choice when viewing leakage in terms of leaked bits. Nevertheless, the following simple argument bounds the maximum difference in terms of maximal leakage. 
\vspace{1mm}
\begin{Lemma} \label{lemmaadditive}
For any joint distribution $P_{XY}$ on alphabets $\mathcal{X}$ and $\mathcal{Y}$, 
\begin{align*}
\sup_{U: U-X-Y} \left( \sup_{\hat{u}(\cdot)} \Pr( U = \hat{u}(Y)) - \max_{u \in \mathcal{U}} P_U(u) \right) \leq 1- e^{-\mathcal{L}(X \to Y)},
\end{align*}
where the supremum is over all $U$ taking values in a finite, but arbitrary, alphabet.
\end{Lemma}
\vspace{1mm}
\begin{proof}
Consider any $U$ satisfying $U-X-Y$. Then
$$\frac{\sup_{\hat{u}(\cdot)} \Pr( U = \hat{u}(Y))}{\max_{u \in \mathcal{U}} P_U(u)} \leq e^{\mathcal{L}(X \to Y)}.$$ Hence,
\begin{align*}
\sup_{\hat{u}(\cdot)} \Pr( U = \hat{u}(Y)) - \max_{u \in \mathcal{U}} P_U(u) \leq \Big( \sup_{\hat{u}(\cdot)} \Pr( U = \hat{u}(Y)) \Big) \left( 1- e^{-\mathcal{L}(X \to Y)} \right) \leq 1- e^{-\mathcal{L}(X \to Y)}.
\end{align*}
\end{proof}

This bound is nontrivial when $\max_{u} P_U(u) < e^{-\mathcal{L}(X \to Y)}$, i.e., $H_\infty(U) > \ml{X}{Y}$. It is worth noting that Alvim \emph{et al.} showed that 
\[ \sup_{\hat{\mathcal{X}}, g: \mathcal{X} \times \mathcal{\hat{X}} \rightarrow \mathbb{R}} \left( \sup_{X-Y-\hat{X}} \E[g(X,\hat{X})] -\max_{\hat{x} \in \hat{\mathcal{X}}} \E[g(X,\hat{x})] \right),\]
where $g$ is ``1-spanning''~\cite[Definition 3]{AlvimAdditiveMultiplicative}, can be efficiently computed~\cite[Theorem 17, Corollary 18]{AlvimAdditiveMultiplicative}. On the other hand, for a given threshold $t$, it is NP-hard~\cite[Theorem 11]{AlvimAdditiveMultiplicative} to decide whether 
\begin{align}
\sup_{P_X} \left( \sup_{\hat{x}(\cdot)} \Pr( X = \hat{x}(Y)) - \max_{x \in \mathcal{X}} P_X(x) \right)  \geq t.
\end{align}
Lemma~\ref{lemmaadditive}, however, gives a simple bound on the latter quantity.

\subsection{Proof of Theorem~\ref{thmmainthm}} \label{sec:proofmain}
Assume, without loss of generality, that $P_X(x) > 0$ for all $x \in \mathcal{X}$.
To show that $\ml{X}{Y} \leq I_\infty (X ; Y)$, consider any $U$ satisfying $U-X-Y$. Define
\begin{equation} \label{eqdefleakU}
\ml{X}{Y}[U]= \log \frac{\sum_{y \in \mathcal{Y}} \max_{u \in \mathcal{U}} P_{UY}(u,y)  }{\max_{u \in \mathcal{U}} P_U(u)},
\end{equation}
so that $\ml{X}{Y} = \sup_{U: U-X-Y} \ml{X}{Y}[U]$. Then
\begin{align*}
 \sum_{y \in \mathcal{Y}} \max_{u \in \mathcal{U}}  P_{UY}(u,y)   
& = \sum_{y\in \mathcal{Y}} \max_{u \in \mathcal{U}} \sum_{x \in \mathcal{X}} P_X(x) P_{U|X}(u|x) P_{Y|X}(y|x)  \\
& \leq  \sum_{y\in \mathcal{Y}} \max_{u \in \mathcal{U}} \sum_{x \in \mathcal{X}} P_X(x) P_{U|X}(u|x) \max_{x' \in \mathcal{X}} P_{Y|X}(y|x')  \\
& = \sum_{y\in \mathcal{Y}} \left( \max_{x' \in \mathcal{X}} P_{Y|X}(y|x') \right) \max_{u \in \mathcal{U}} \sum_{x \in \mathcal{X}} P_X(x) P_{U|X}(u|x)    \\
& =  \sum_{y\in \mathcal{Y}} \max_{x \in \mathcal{X}} P_{Y|X}(y|x) \max_{u \in \mathcal{U}} P_U(u).
\end{align*}
Therefore, $\ml{X}{Y}[U] \leq I_\infty(X ; Y)$ for all $P_{U|X}$, hence $\ml{X}{Y} \leq I_\infty(X ; Y)$. \\

For the reverse inequality, we construct a $P_{U|X}$ for which $\ml{X}{Y}[U] = I_\infty(X ; Y)$, which we will call the ``shattering'' $P_{U|X}$. To that end, let $p^\star = \min_{x \in \mathcal{X}} P_X(x)$. For each $x \in \mathcal{X}$, let $k(x) = P_X(x)/p^{\star}$, and let $ \mathcal{U} = \bigcup_{x \in \mathcal{X}} \{ (x,1), (x,2), \ldots, (x,\lceil k(x)\rceil) \}$.  For each $u=(i_u,j_u) \in \mathcal{U}$ and $x \in \mathcal{X}$, let $P_{U|X}(u|x)$ be:
\begin{align} \label{eqshattering}
 P_{U|X}((i_u,j_u)|x) = \begin{cases}
\frac{p^{\star}}{P_X(x)}, &  i_u=x, ~~1 \leq j_u \leq \lfloor k(x) \rfloor , \\
1 - \frac{(\lceil k(x) \rceil -1)p^{\star}}{P_X(x)}, & i_u=x,~~ j_u=\lceil k(x) \rceil, \\
0, &  i_u \neq x, ~~1 \leq j_u \leq \lceil k(i_u) \rceil.
\end{cases}
\end{align} 
It is easy to check that if $\lfloor k(x) \rfloor = \lceil k(x) \rceil$, then the corresponding formulas are equal.
Then, for each $((i_u,j_u),x) \in \mathcal{U} \times \mathcal{X}$,
\begin{align} \label{eqshatteringjointX}
P_{UX}((i_u,j_u),x)  = \begin{cases}
p^{\star}, &  i_u=x, ~1 \leq j_u \leq \lfloor k(x) \rfloor , \\
P_X(x)  - ( \lceil k(x) \rceil   -1)p^{\star}, &  i_u=x, ~j_u=\lceil k(x) \rceil, \\
0, &  i_u \neq x, ~1 \leq j_u \leq \lceil k(i_u) \rceil.
\end{cases}
\end{align} 
Note that the supports of $P_{U|X=x}$ is disjoint for each distinct $x$, and it effectively ``shatters'' $x$ into shards of probability $p^\star$. Now note that
\begin{equation}
\label{eqleakdenom}
\max_{u \in \mathcal{U}} P_U(u) =   \max_{(i_u,j_u) \in \mathcal{U}} P_{UX}((i_u,j_u),i_u) = p^{\star}.
\end{equation}
Now consider any $(u,y) \in \mathcal{U} \times \mathcal{Y}$. We have
\begin{align} \label{eqshatteringjointY}
P_{UY} ((i_u,j_u),y) 
& = \sum_{x \in \mathcal{X}} P_X(x) P_{U|X}((i_u,j_u)|x) P_{Y|X}(y|x) \notag \\
& = P_X(i_u) P_{U|X}((i_u,j_u)|i_u) P_{Y|X}(y|i_u) \notag \\
& = \begin{cases} 
p^{\star} P_{Y|X}(y|i_u), &  1 \leq j_u \leq \lfloor k(i_u) \rfloor, \\
 (P_X(x)  -   (\lceil k(x) \rceil -1)p^{\star})P_{Y|X}(y|i_u), &  j_u = \lceil k(i_u) \rceil.
\end{cases}
\end{align}
Then, for a given $y \in \mathcal{Y}$,
\begin{align}
\max_{(i_u,j_u) \in \mathcal{U}} P_{UY}((i_u,j_u),y) &  = \max_{(i_u,1) \in \mathcal{U}} p^{\star} P_{Y|X}(y|i_u)  = \max_{x \in \mathcal{X}} p^{\star} P_{Y|X}(y|x).  \label{eqleaknum}
\end{align}
Finally, we get
\begin{equation*}
\ml{X}{Y} \geq \ml{X}{Y}[U] = \log \sum_{y\in \mathcal{Y}} \max_{x \in \mathcal{X}} P_{Y|X} (y|x),
\end{equation*}
where the  inequality follows from the definition, and the equality follows from equations~\eqref{eqdefleakU},~\eqref{eqleakdenom}, and~\eqref{eqleaknum}. \hfill $\blacksquare$ 

Note that in the above proof, the conditional distribution (given in~\eqref{eqshattering}) that achieves the supremum in~\eqref{eq:maxleakagedef} depends on $P_{XY}$ only through the $X-$marginal, $P_X$. So we get the following proposition.
\vspace{1mm}
\begin{Proposition}~\label{prop:shattering}
Let $\mathcal{X}$ be a finite alphabet and $P_X$ a distribution on $\mathcal{X}$. Then the ``shattering'' $P_{U|X}$ defined in~\eqref{eqshattering} achieves the supremum in~\eqref{eq:maxleakagedef} for all finite alphabets $\mathcal{Y}$ and conditional distributions $P_{Y|X}$.
\end{Proposition}

\section{Maximal Leakage: Variations and Extensions} \label{sec:variations}

We now consider several natural variations to our threat model. In particular, we consider the following scenarios. 
\begin{enumerate}
\item The adversary chooses the variable of interest $U$ \emph{after} observing $Y$, i.e., for different realizations of $Y$, they might attempt to guess different functions of $X$.
\item The adversary only needs their guess to be within a certain distance of the true value of $U$.
\item The adversary can make several guesses.
\item The adversary attempts to maximize a gain function defined on $\mathcal{U} \times \hat{\mathcal{U}}$ for some alphabet $\hat{\mathcal{U}}$.
\end{enumerate}
We modify the definition of maximal leakage accordingly for each scenario. However, for each of these cases, the resulting computable characterization is unchanged. This shows that the definition of maximal leakage is robust, and meets the requirements we presented in the introduction. In particular, it has several useful operational interpretations, and it requires minimal assumptions about the adversary's goal.

Furthermore, we extend the notion of maximal leakage in two directions. First, we propose a conditional form of leakage $\cml{X}{Y}{Z}$, where $Z$ represents side information available at the adversary. Finally, we generalize Theorem~\ref{thmmainthm} to account for a large class of random variables, including point processes. We use the general formula to analyze a simple model of the SSH side-channel.

\subsection{Multiple Functions of Interest}
In our threat model, we assumed that the adversary is interested in a specific randomized function of $X$. However, they could be interested in several functions and choose which one to guess only after seeing the realization of $Y$. To account for this, we modify the definition of maximal leakage as follows.
\vspace{1mm}
\begin{Definition}[Opportunistic Maximal Leakage] \label{defmaxleakY}
Given a joint distribution $P_{XY}$ on alphabets $\mathcal{X}$ and $\mathcal{Y}$, define
\begin{align}
\label{eq:oppML:1}
\mly{X}{Y} & =  \log \sum_{y \in \mathcal{Y}} P_Y(y) \sup_{U: U-X-Y} \frac{\max_{u \in \mathcal{U} }P_{U|Y}(u|y)}{\max_{u \in \mathcal{U} }P_{U}(u)} \\
\label{eq:oppML:2}
       & =  \sup_{(U_y, y \in \mathcal{Y})-X-Y} \log 
                    \sum_{y \in \mathcal{Y}} P_Y(y)   
                  \frac{\max_{u \in \mathcal{U}_y }P_{U_y|Y}(u|y)}
                   {\max_{u \in \mathcal{U}_y }P_{U_y}(u)}  \\
\label{eq:oppML:3}
       & =  \sup_{U} \log \sum_{y \in \mathcal{Y}} P_Y(y) 
             \frac{\max_{u \in \mathcal{U}} \sum_{x \in \mathcal{X}}
                              P_{U|X,Y}(u|x,y)P_{X|Y}(x|y)}
                                 {\max_{u \in \mathcal{U}} 
                            \sum_{x \in \mathcal{X}}P_{U|X,Y}(u|x,y) P_X(x)},
\end{align}
where the $U$ variables in all three suprema take values in finite but
arbitrary alphabets.
\end{Definition}
\vspace{1mm}
The different $U_y, y \in \mathcal{Y}$ in (\ref{eq:oppML:2}) can 
be interpreted as different
secrets that the adversary might attempt to guess. The adversary 
opportunistically attempts to guess secret $U_y$ when it observes $Y = y$.
Notably, allowing the adversary this additional freedom does not
change the result.
\vspace{1mm}
\begin{Theorem} \label{thmmaxleakY}
For any joint distribution $P_{XY}$ on finite alphabets $\mathcal{X}$ and $\mathcal{Y}$, \[ \ml{X}{Y} = \mly{X}{Y}. \]
\end{Theorem}
\vspace{1mm}
\begin{proof}
It follows straightforwardly from the definitions that $\mly{X}{Y} \geq \ml{X}{Y}$. For the reverse direction, consider the following proposition.
\begin{Proposition} \label{prop:guessboundY}
Suppose $U$, $X$, and $Y$ are discrete random variables that satisfy the Markov chain $U-X-Y$. Then for each $y \in \mathrm{supp}(Y)$, 
\begin{align*}
\frac{\max_{u \in \mathcal{U} }P_{U|Y}(u|y)}{\max_{u \in \mathcal{U} }P_{U}(u)} \leq \frac{\max_{x: P_{X|Y}(x|y) > 0 }P_{Y|X}(y|x)}{P_Y(y)}.
\end{align*}
\end{Proposition}
It follows from the proposition that
\begin{align*}
\exp\{\mly{X}{Y}\} \leq \sum_{y \in \mathrm{supp}(Y)} \max_{x: P_{X|Y}(x|y) > 0} P_{Y|X}(y|x) \leq \sum_{y \in \mathcal{Y}} \max_{x: P_{X}(x)>0} P_{Y|X}(y|x).
\end{align*}
Then it remains to prove Proposition~\ref{prop:guessboundY}. To that end, consider a triple of discrete random variables $U$, $X$, and $Y$ satisfying $U-X-Y$, and fix $y \in \mathrm{supp}(Y)$. Then
\begin{align*}
\max_{u \in \mathcal{U}} P_{U|Y}(u|y) & = \max_{u \in \mathcal{U}}  \sum_{x: P_{X|Y}(x|y) > 0 } P_{U|X}(u|x) P_{X|Y}(x|y) \\
& =     \max_{u \in \mathcal{U}}  \sum_{x: P_{X|Y}(x|y) > 0 } P_{U|X}(u|x) P_{X}(x)\frac{P_{Y|X}(y|x)}{P_Y(y)} \\
& \leq  \max_{x': P_{X|Y}(x'|y) > 0 } \frac{P_{Y|X}(y|x')}{P_Y(y)} \max_{u \in \mathcal{U}} \sum_{x: P_{X|Y}(x|y) > 0 } P_{UX}(u,x) \\
& \leq \max_{x': P_{X|Y}(x'|y) > 0 } \frac{P_{Y|X}(y|x')}{P_Y(y)} \max_{u \in \mathcal{U}} P_{U}(u),
\end{align*}
as desired.
\end{proof}

\subsection{Approximate Guessing}
Consider the case in which the adversary only needs the guess to be within a certain distance of the true function value, according to a given distance metric. As such, the random variable $U$, over which we are optimizing, now lives in a given metric space $\mathcal{U}$ and is no longer restricted to be discrete. We call this modified measure maximal locational leakage. The term ``locational'' is motivated by the scenario in which the variable of interest $U$ is a geographical location, such as a person's home address (potentially revealed by GPS traces~\cite{HomeGPS}) or a person's physical location (potentially revealed by cellular tracking data~\cite{WickerLocationPrivacy}). 
\vspace{1mm}
\begin{Definition}[Maximal Locational Leakage]\label{defleakcont}
Given a joint distribution $P_{XY}$ on finite alphabets $\mathcal{X}$ and $\mathcal{Y}$, and  a metric space $\mathcal{U}$ (with its associated Borel $\sigma$-field), the maximal locational leakage from $X$ to $Y$ is defined as
\begin{equation} \label{eqdefleakcont}
\mathcal{L}_{\mathcal{U}}(X \!\! \to \!\! Y)  = \sup_{\substack{U: U-X-Y \\ \exists u: \Pr(U \in B(u)) > 0}}  \log \frac{\sup_{\hat{u}(\cdot)} \Pr(U \in B(\hat u(Y))) }{\sup_{\hat{u}} \Pr(U \in B(\hat u)) },
\end{equation}
where $B(u)$ is the closed unit ball centered at $u \in \mathcal{U}$. 
\end{Definition}
\vspace{1mm}
\begin{Theorem} \label{lemmaconteqdisc}
For any joint distribution $P_{XY}$ on finite alphabets $\mathcal{X}$ and $\mathcal{Y}$, and any metric space $\mathcal{U}$,
\begin{equation*}
\mathcal{L}_{\mathcal{U}}(X \!\! \to \!\! Y) \leq\mathcal{L}(X \!\! \to \!\! Y), 
\end{equation*}
with equality if $\mathcal{U}$ has a countably infinite subset $S$ such that no pair of its elements can be contained in a single unit ball.
\end{Theorem}
\vspace{1mm}
\begin{proof} Assume, without loss of generality, that $X$ has full support.  Now consider any $U$ and $\hat{u}(Y)$ in the maximization of \eqref{eqdefleakcont}:
\begin{align*}
 \Pr(U \in B(\hat u(Y))
& \leq \sum_{y \in \mathcal{Y}} \sup_{u \in \mathcal{U}} P(U \in B(u), Y=y) \\
& = \sum_{y \in \mathcal{Y}} \sup_{u \in \mathcal{U}} \sum_{x \in \mathcal{X}} P(U \in B(u), X=x, Y=y) \\
& = \sum_{y \in \mathcal{Y}}  \sup_{u \in \mathcal{U}}  \sum_{x \in \mathcal{X}}  P(U  \in  B(u))  P(X=x|U  \in  B(u)) P_{Y|X}(y|x) \\
& \leq \sum_{y \in \mathcal{Y}} \sup_{u \in \mathcal{U}} P(U \in B(u)) \sup_{x \in \mathcal{X}} P_{Y|X}(y|x) \\
& = \left[ \sum_{y \in \mathcal{Y}} \sup_{x \in \mathcal{X}} P_{Y|X}(y|x) \right] \sup_{u \in \mathcal{U}} P(U \in B(u)).
\end{align*}
Therefore,
\begin{align*}
\mathcal{L}_{\mathcal{U}}(X \!\! \to \!\! Y)\leq \log \sum_{y \in \mathcal{Y}} \sup_{x \in \mathcal{X}} P_{Y|X}(y|x) = \mathcal{L}(X \!\! \to \!\! Y). 
\end{align*}
If $\mathcal{U}$ satisfies the given condition (e.g., $\mathcal{U}$ is unbounded), then exact guessing of discrete functions can be simulated by choosing $S$ to be the support of $U$. Hence $\mathcal{L}_{\mathcal{U}}(X \!\! \to \!\! Y) \geq \mathcal{L}(X \!\! \to \!\! Y)$, which implies the equality. 
\end{proof}

\subsection{Multiple Guesses}

The definition of maximal leakage (Definition~\ref{defmaxleakage}) allowed the adversary a single guess. However, an adversary might be able to make several guesses in some practical scenarios.
For example, if the adversary is trying to guess a password $U$ of some system, they can typically try several passwords before they are locked out. Similarly, if they are trying to guess a secret key to decrypt an encrypted message, they can make several attempts. 
 We modify the definition to allow for $k$ guesses, for any integer $k$, as follows. 
\vspace{1mm}
\begin{Definition}[$k$-Maximal Leakage]\label{defkleakage}
Given a joint distribution $P_{XY}$ on finite alphabets $\mathcal{X}$ and $\mathcal{Y}$, and a positive integer $k$, the $k$-maximal leakage from $X$ to $Y$ is defined as
\begin{equation*}
\kml{X}{Y} = \sup_{U - X - Y - (\hat{U}_i)_{i=1}^k} \log 
\frac{ \Pr\left(\bigvee_{i=1}^k U=\hat{U}_i    \right) }{ \max_{\substack{S \subseteq \mathcal{U}\\ |S|\leq k }} P_U(S)},
\end{equation*}
where $U$ takes values in a finite, but arbitrary, alphabet.
\end{Definition}
\vspace{1mm}
It turns out that $k$-maximal leakage and maximal leakage are equivalent.
\vspace{1mm}
\begin{Theorem} \label{lemmakeq1}
For any joint distribution $P_{XY}$ on finite alphabets $\mathcal{X}$ and $\mathcal{Y}$, and any $k \in \mathbb{N}$, 
\begin{align*}
\kml{X}{Y} = \ml{X}{Y}.
\end{align*}
\end{Theorem}
The proof is given in Appendix~\ref{app:kleakage}.

\subsection{General Gains}
We now consider the case in which different realizations of $U$ might have different significance for the adversary. For example, an adversary monitoring the timing of packet transmissions over a given network~\cite{ParvAnonymous} might seek to deduce source-destination pairs. However, they might be more interested in detecting communication between specific pairs, corresponding to (say) suspicious persons, governmental agencies, etc. Hence, there is more value to the detection of the existence of a link, rather than its absence. 
This mirrors the asymmetric cost of false alarm and missed detection in hypothesis testing. To account for this, we use a gain function $g: \mathcal{U} \times \hat{\mathcal{U}} \rightarrow [0, \infty)$ and maximize over gain functions as follows.
\vspace{1mm}
\begin{Definition}[Maximal Gain Leakage] \label{defmaxgains}
Given a joint distribution $P_{XY}$ on finite alphabets $\mathcal{X}$ and $\mathcal{Y}$, the \emph{maximal gain leakage} is defined as
\begin{align*}
\mlg{X}{Y} = \sup_{ \substack{U: U-X-Y \\ \mathcal{\hat{U}}, g: \mathcal{U} \times \hat{\mathcal{U}} \rightarrow [0, \infty): \\ \sup_{\hat{u}} \E[g(U,\hat{u})] > 0 }} \log \frac{\sup_{\hat{u}(\cdot)} \E[g(U,\hat{u}(Y))]}{\sup_{\hat{u}} \E[g(U,\hat{u})]},
\end{align*}
where $\mathcal{U}$ is a finite, but arbitrary, alphabet.
\end{Definition}
\vspace{1mm}
Similarly to previous variations, maximal gain leakage turns out to be equivalent to maximal leakage.
\begin{Theorem} \label{thmmaxgains}
For any joint distribution $P_{XY}$ on finite alphabets $\mathcal{X}$ and $\mathcal{Y}$, \[ \ml{X}{Y} = \mlg{X}{Y}. \]
\end{Theorem}

\begin{Remark}
For a similar result in which $U = X$ but one takes the supremum over
all $X$ distributions, see Alvim 
\emph{et al.}\cite{AlvimAdditiveMultiplicative}.
\end{Remark}

\begin{IEEEproof}
It follows straightforwardly from the definitions that $\mlg{X}{Y} \geq \ml{X}{Y}$. For the reverse direction, consider any $U$ satisfying $U-X-Y$, any (non-empty) set $\hat{\mathcal{U}}$ and function $g: \mathcal{U} \times \hat{\mathcal{U}} \rightarrow [0, \infty)$. Then
\begin{align*}
\sup_{\hat{u}(\cdot)} \E[g(U,\hat{u}(Y))] & = \sum_{y \in \mathcal{Y}} \sup_{\hat{u} \in \hat{\mathcal{U}}} \sum_{u \in \mathcal{U}} g(u,\hat{u}) P_{UY} (u,y)   \\
& = \sum_{y \in \mathcal{Y}} \sup_{\hat{u} \in \hat{\mathcal{U}}} \sum_{u \in \mathcal{U}}
\sum_{x \in \mathrm{supp}(X) } g(u, \hat{u}) P_X(x) P_{U|X}(u|x) P_{Y|X}(y|x) \\
& \leq \sum_{y \in \mathcal{Y}} \Big(\max_{x' \in \mathrm{supp}(X) } P_{Y|X}(y|x')\Big) \sup_{\hat{u} \in \hat{\mathcal{U}}} \sum_{u \in \mathcal{U}}
\sum_{x \in \mathrm{supp}(X) } g(u, \hat{u}) P_X(x) P_{U|X}(u|x) \\
& = \sum_{y \in \mathcal{Y}} \Big(\max_{x' \in \mathrm{supp}(X) } P_{Y|X}(y|x')\Big) \sup_{\hat{u}} \E[g(U,\hat{u})],
\end{align*}
as desired.
\end{IEEEproof}

\subsection{Conditional Maximal Leakage}

One of the main challenges in information leakage problems comes from the fact that the adversary can acquire information from multiple sources. This prompted researchers in database security to make very conservative assumptions about the knowledge of the adversary: differential privacy is introduced in a setup in which the adversary knows all the entries of the database except one~\cite{DworkCalibrating}. This also raises interest in the behavior of mechanisms under \emph{composition}~\cite{KairouzComposition}. That is, if the adversary receives multiple independent observations released by a given secure mechanism, how do the security guarantees degrade? 	
In order to answer these questions, we propose a conditional form of maximal leakage, which is defined analogously to Definition~\ref{defmaxleakage}.
\vspace{1mm}
\begin{Definition}[Conditional Maximal Leakage] \label{defcondmaxleakage}
Given a joint distribution $P_{XYZ}$ on alphabets $\mathcal{X}$, $\mathcal{Y}$ and $\mathcal{Z}$, the \emph{conditional maximal leakage} from $X$ to $Y$ given $Z$ is defined as
\begin{equation} 
\cml{X}{Y}{Z} = \sup_{\substack{U: U-X-Y|Z}} \log \frac{\Pr(U=\hat{U}(Y,Z))}{\Pr(U=\tilde{U}(Z))},
\end{equation}
where $U$ takes values in a finite, but arbitrary, alphabet, and $\hat{U}(Y,Z)$ and $\tilde{U}(Z)$ are the optimal (i.e., MAP) estimators of $U$ given $(Y,Z)$ and $Z$, respectively.
\end{Definition}
\vspace{1mm}
\begin{Remark}
The Markov chain $U-X-Y|Z$ is equivalent to $U-(X,Z)-Y$. The above definition is hence conservative, in that it allows the channel from $X$ to $U$ to depend on $Z$. One could instead consider $U$s satisfying the Markov chain $U-X-(Y,Z)$.
The quantity so modified appears to be considerably more difficult to analyze.
\end{Remark}
\vspace{1mm}
\begin{Theorem} \label{thmcondmaxleakage}
Given a joint distribution $P_{XYZ}$ on finite alphabets $\mathcal{X}$, $\mathcal{Y}$ and $\mathcal{Z}$, the conditional maximal leakage from $X$ to $Y$ given $Z$ is given by
\begin{equation} \label{eqcondmaxleakage}
\cml{X}{Y}{Z} = \log \left( \max_{z: P_Z(z)>0} \sum_y \max_{x: P_{X|Z}(x|z) > 0} P_{Y|XZ} (y|x,z) \right).
\end{equation}
\end{Theorem}
\vspace{1mm}
In other terms, $\cml{X}{Y}{Z} = \max_{z \in \mathrm{supp}(Z)} \mathcal{L}(X \!\! \to \!\! Y | Z=z)$, where $\mathcal{L}(X \!\! \to \!\! Y | Z=z)$ is interpreted as the unconditional maximal leakage evaluated with respect to the joint distribution $P_{XY|Z=z}$. 
The following corollary summarizes important properties of conditional maximal leakage.
\vspace{1mm}
\begin{Corollary} \label{corrpropcond}
Given a joint distribution $P_{XYZ}$ on finite alphabets $\mathcal{X}$, $\mathcal{Y}$ and $\mathcal{Z}$,
\begin{enumerate}
\item (\emph{Data Processing Inequality}) If the Markov chain $X-Y-V|Z$ holds for a discrete random variable $V$, then $\cml{X}{V}{Z} \leq \min\{\cml{X}{Y}{Z},\cml{Y}{V}{Z}  \}$.
\item $\cml{X}{Y}{Z} \leq \min\{ \log |\mathcal{X}|, \log |\mathcal{Y}|\}$.
\item $\cml{X}{Y}{Z} = 0$ iff $X-Z-Y$ holds.
\item (\emph{Additivity}) 
If $\{(X_i,Y_i,Z_i)\}_{i=1}^{n}$ are mutually independent, then
\begin{equation*}
\cml{X_1^n}{Y_1^n}{Z_1^n} = \sum_{i=1}^n \cml{X_i}{Y_i}{Z_i}.
\end{equation*}
\item $\cml{X}{Y}{Z} \geq I(X ; Y|Z)$.
\item $\cml{X}{Y}{Z}$ is not symmetric in $X$ and $Y$.
\item If $Z-X-Y$ holds, then \begin{align*}
\cml{X}{Y}{Z}  \leq \ml{X}{Y},
\end{align*}
with equality if for some $z \in \mathrm{supp}(Z)$, $\mathrm{supp}(P_{X|Z=z}) = \mathrm{supp}(P_X)$.
\item $\ml{X}{(Y,Z)} \leq \ml{X}{Z} + \cml{X}{Y}{Z}.$
\end{enumerate}
\end{Corollary}
\vspace{1mm}
Similarly to maximal leakage, properties 1)-4) can be seen as axiomatic for a conditional leakage metric.  
Property 5) is analogous to the relationship between maximal leakage and mutual information. Property 7) is interesting in that it exhibits a behavior similar to mutual information. Indeed, if $Z-X-Y$ holds, then $I(X;Y|Z) \leq I(X;Y)$. Property 8) can be viewed as a one-sided chain rule.
A simple consequence of properties 7) and 8) is the following composition lemma.
\vspace{1mm}
\begin{Lemma}[Composition Lemma] 
Given a joint distribution $P_{XYZ}$ on finite alphabets $\mathcal{X}$, $\mathcal{Y}$ and $\mathcal{Z}$, if $Z-X-Y$ holds, then 
\[ \ml{X}{(Y,Z)} \leq \ml{X}{Z}+\ml{X}{Y}. \]
\end{Lemma}
\vspace{1mm}
Hence, if an adversary has access to side information $Z$ and this is not known to the system designer (which is often the case in practice), then minimizing $\ml{X}{Y}$ (irrespective of $Z$) is still a reasonable objective. 

The proofs of Theorem~\ref{thmcondmaxleakage} and Corollary~\ref{corrpropcond} are given in Appendices~\ref{app:conditional} and~\ref{app:CorrConditionalProof}, respectively.

\subsection{General Alphabets} \label{sec:general}
Finally, we generalize Theorem~\ref{thmmainthm} to allow for a large class of random variables and stochastic processes. We use the general formula to study a simple model of the SSH side-channel and analyze the performance of commonly used mechanisms. Our analysis suggests that memoryless schemes generally do not perform well under maximal leakage.

Before stating the theorem for general alphabets, we introduce the following notation. For a given probability distribution $P_X$, and a measurable function $f: \mathcal{X} \rightarrow \mathbb{R}$, the \emph{essential supremum} of $f$ with respect to $P_X$ is defined as:
\begin{align}
\label{eqdefessup}
\essup_{P_X} f(X) = \inf \{ \alpha: P_X( \{x: f(x) > \alpha \}) =0 \}.
\end{align}
Equivalently, 
\begin{align}
\label{eqdefessup2}
\essup_{P_X} f(X) = \sup \{ \beta: {P_X}( \{x: f(x) > \beta \}) > 0 \}.
\end{align}
\begin{Theorem} \label{thmgeneralformula}
Let $(\mathcal{X} \times \mathcal{Y}, \sigma_{XY}, P_{XY})$ be a probability space with associated probability spaces $(\mathcal{X}, \sigma_X, P_{X})$ and $(\mathcal{Y}, \sigma_Y, P_{Y})$, where $\sigma_{XY}$ is the product sigma-algebra.
 \begin{enumerate}
\item If $P_{XY} \ll P_X \times P_Y$ and $\sigma_X$ is generated by a countable set, then
\begin{align} \label{eq:generalformula}
\ml{X}{Y} =\log  \int_\mathcal{Y} \essup_{P_X}  f(X,y)   P_Y(dy),
\end{align}
where $f(x,y) = \frac{dP_{XY}}{d(P_X \times P_Y)} (x,y) .$
\item If absolute continuity fails, then $\ml{X}{Y} = +\infty$.
\end{enumerate}
\end{Theorem}
\vspace{1mm}

In the discrete case, $\ml{X}{Y}$ depends on $P_{XY}$ only through $P_{Y|X}$ and the support of $P_X$. Although it is not immediately clear from~\eqref{eq:generalformula}, this holds true in the general case in the following sense. Define an equivalence relation on the set of probability measures on a given measurable space as follows: 
\begin{align} \label{eq:equivalence}
P \equiv Q \text{ if } P \ll Q \text{ and }Q \ll P.
\end{align}
Then $\ml{X}{Y}$ depends on $P_{XY}$ only through $P_{Y|X}$ and the equivalence class of $P_{X}$. We formalize this observation in the following lemma. 
\vspace{1mm}
\begin{Lemma} \label{lemmageneralformula}
Let $(\mathcal{X}, \sigma_X, P_{X_1})$ be a probability space, and let $(\mathcal{Y}, \sigma_Y)$ and $(\mathcal{X} \times \mathcal{Y}, \sigma_{XY})$ be measurable spaces, where $\sigma_{XY}$ is the product sigma-algebra. Fix a 
kernel $\mu$ from $\mathcal{X}$ to $\mathcal{Y}$, that is, a function
function $\mu: \mathcal{X} \times \sigma_Y \rightarrow [0,\infty)$ that satisfies:
\begin{enumerate}
\item For every $B \in \sigma_Y$, $\mu(\cdot,B)$ is $\sigma_X$-measurable.
\item For every $x \in \mathcal{X}$, $\mu(x,\cdot)$ is a probability measure on $(\mathcal{Y}, \sigma_Y)$.
\end{enumerate}
Let $P_{X_1 Y_1}$ and $P_{Y_1}$ be the probability measures induced by $P_{X_1}$ and $\mu(\cdot,\cdot)$ on $(\mathcal{X} \times \mathcal{Y}, \sigma_{XY})$ and $(\mathcal{Y}, \sigma_Y)$, respectively. 
 If $P_{X_1 Y_1} \ll P_{X_1} \times P_{Y_1}$, then $\mu(x,\cdot) \ll P_{Y_1}$. If, in addition, $\sigma_X$ is generated by a countable set,
 \begin{align} \label{eq:generalformula2}
 \ml{X_1}{Y_1} =\log \int_{\mathcal{Y}} \essup_{P_{X_1}} \left( \frac{d \mu(X,\cdot)}{dP_{Y_1}} \right) P_{Y_1}(dy)= \log \int_{\mathcal{Y}} \essup_{P_{X}} \left( \frac{d \mu(X,\cdot)}{dQ_{Y}} \right) Q_{Y}(dy),
\end{align} where $P_{X}$ is an arbitrary representative of the equivalence class (cf.~\eqref{eq:equivalence}) of $P_{X_1}$ and $Q_Y$ is any measure satisfying
$P_{Y_1} \ll Q_Y$. 

Consequently, if $P_{X_2}$ is a probability measure on $(\mathcal{X}, \sigma_X)$ satisfying $P_{X_2} \equiv P_{X_1}$, then $ P_{X_2 Y_2} \ll P_{X_2} \times P_{Y_2}$, $P_{Y_1} \equiv P_{Y_2}$, and
\begin{align*}
\ml{X_2}{Y_2} = \ml{X_1}{Y_1},
\end{align*}
where $P_{X_2 Y_2}$ and $P_{Y_2}$ are the induced probability measures on $(\mathcal{X} \times \mathcal{Y}, \sigma_{XY})$ and $(\mathcal{Y}, \sigma_Y)$, respectively. 
\end{Lemma}
\vspace{1mm}
The proofs of Theorem~\ref{thmgeneralformula} and Lemma~\ref{lemmageneralformula} are given in Appendices~\ref{app:generalformula} and~\ref{app:lemmageneral}, respectively. We now discuss implications and examples of the theorem.
\vspace{1mm}
\begin{Corollary} \label{corrgeneralformula}
	Let $(\mathcal{X} \times \mathcal{Y}, \sigma_{XY}, P_{XY})$ be a probability space with associated probability spaces $(\mathcal{X}, \sigma_X, P_{X})$ and $(\mathcal{Y}, \sigma_Y, P_{Y})$. Assume $P_{XY} \ll P_X \times P_Y$ and $\sigma_X$ is generated by a countable set. Then
	\begin{enumerate}
		\item $\ml{X}{Y} = 0$ iff $X$ and $Y$ are independent.
		\item (\emph{Additivity}) 
		If $\{(X_i,Y_i)\}_{i=1}^{n}$ are mutually independent, then
		\begin{equation*}
			\ml{X_1^n}{Y_1^n} = \sum_{i=1}^n \ml{X_i}{Y_i}.
		\end{equation*}
		\item $\ml{X}{Y} \geq I(X ; Y)$.
	\end{enumerate}
\end{Corollary}
\vspace{1mm}
\begin{proof}
For 3) it suffices to consider the case in which $P_{XY} \ll P_X \times P_Y$.
Let $f(\cdot,\cdot)$ denote the derivative of $P_{XY}$ with respect to
$P_X \times P_Y$ and consider the following.
\begin{align*}
	I(X;Y) = \E [ \log f(X,Y) ] \stackrel{\text{(a)}} \leq
	\log \E [ f(X,Y)] & = \log  \int_\mathcal{Y} \int_\mathcal{X}  f^2(x,y) P_X(dx)  P_Y(dy) \\
	& \leq \log \int_\mathcal{Y} \left(\essup_{P_X} f(X,y) \right) \int_\mathcal{X} f(x,y) P_X (dx) P_Y (dy) \\
	& \stackrel{\text{(b)}} = \ml{X}{Y},
\end{align*}
where (a) follows from Jensen's inequality, and (b) follows from the fact $\int_\mathcal{X} f(x,y) P_X (dx) = 1 \ \ \text{$Y$--a.s.}$ because
\begin{equation*}
\int_{\mathcal{Y}} \left[ \int_{\mathcal{X}} f(x,y) P_X(dx) \right]
     1(y \in B) P_Y(dy) = P(Y \in B)
\end{equation*}
for all $B$. 1) follows from the definition and 3). 2) follows from the fact that if $(X_1,Y_1)$ is independent of $(X_2,Y_2)$, then $f(x_1,x_2,y_1,y_2)=f(x_1,y_1)f(x_2,y_2)$.
\end{proof}

Recall that the data processing inequality also holds by Lemma~\ref{lemmadefprop}. Hence, the general formula of maximal leakage retains the axiomatic properties we required in (R3). It also covers all combinations of discrete, countable, or continuous random variables $X$ and $Y$.
\vspace{1mm} 
\begin{Corollary} \label{corrLXYcontinuous}
	If $X$ and $Y$ are jointly continuous real random variables,  
	\begin{align} \label{LXYcontinuous}
		\ml{X}{Y} = \log \int_{\mathbb{R}} \sup_{x: f_X(x) > 0} f_{Y|X}(y|x) dy,
	\end{align}
	where $f_X$ and $f_{Y|X}(\cdot|\cdot)$ are the marginal pdf of $X$ and the conditional pdf of $Y$ given $X$, respectively.
\end{Corollary}
\vspace{1mm}
\begin{Example}
	If $X$ and $Y$ are jointly Gaussian, then
	\begin{align*}
		\ml{X}{Y}= \begin{cases}
			0, & \text{if }X \text{ and } Y \text{ are independent,} \\
			+\infty, & \text{otherwise}.
		\end{cases}
	\end{align*}
\end{Example}
\vspace{1mm}
\begin{Example}
	Suppose $X$ is real and its pdf satisfies $f_X(x) > 0$ for all $x \in \mathbb{R}$. Let $Y = X +Z$, where $Z$ is a continuous real random variable independent of $X$. Let $z_0 = \argmax f_Z(z)$. Then
	\begin{align*}
		\ml{X}{Y}= \log \int_{\mathbb{R}} \sup_{x \in \mathbb{R}} f_{Y|X}(y|x) dy = \log \int_{\mathbb{R}} \sup_{x} f_{Z}(y-x) dy =  \log \int_{\mathbb{R}} f_Z (z_0) dy = + \infty.
	\end{align*}
\end{Example}
\vspace{1mm}
 
The above examples suggest that ``adding independent noise'' is not necessarily secure in the maximal leakage sense. The following example considers a simple model of the SSH side-channel and further illustrates this point.
\vspace{1mm}
\begin{Example} Consider the SSH side-channel and suppose we wish to perturb the packet timings before they are sent over the network so that we decrease information leakage. We represent the process of incoming packets as a Poisson process of a given rate, $\lambda$. More formally, 
	fix $T \in \mathbb{R}_+$ and let $\Omega_T$ be the set of all counting functions on $[0,T]$, i.e., $\omega \in \Omega_T$ is an integer-valued, nondecreasing, right-continuous function on $[0,T]$ and $\omega(0)=0$. Let $\{\mathcal{F}_t\}_{t=0}^T$ be the filtration over $\Omega_T$ generated by the mapping $\omega \mapsto \omega_t$. Let $X_0^T$ be a Poisson process of rate $\lambda$, representing the incoming packets. Let $Y_0^T$ be a point process on $(\Omega_T, \mathcal{F}_T)$ representing the outgoing packets.
	\paragraph{Memoryless scheme} Suppose we hold each packet for an independent random amount of time before releasing it into the network. More specifically, let $Y_0^T$ be the output of an initially-empty exponential-server queue with rate $\mu > \lambda$ and input $X_0^T$. Then
	\begin{align} \label{eq:poissonqueue}
		\frac{1}{T} \ml{X_0^T}{Y_0^T} = \mu, 
	\end{align}
	and, as $T \rightarrow \infty$, the average waiting time for a packet (between arrival and transmission) tends to $\frac{1}{\mu-\lambda}$. 
	Note that the system is unstable if $\mu < \lambda$. 
	\begin{proof}
	 Let $P_0$ be the probability measure on $(\Omega_T, \mathcal{F}_T)$ under which the output is distributed as a Poisson process of rate one. It is known~\cite{AaronReliabilityTiming}\cite[Ch. VI, Theorem T3]{BremaudPoint} that for $(x, y) \in \Omega_T \times \Omega_T$,
	\begin{align*}
		\frac{d P_{XY}}{d P_{X} \times P_0} (x,y) = \exp \left[ \int_0^T \log ( \mu \mathbb{I}(x_t > y_{t-}))dy_t + \int_0^T (1- \mu \mathbb{I} (x_t > y_t) )dt  \right] =: L(x,y). 
	\end{align*}
	Now note that,
	\begin{align*}
		\frac{d P_{XY}}{d P_{X} \times P_Y} dP_Y = \frac{d P_{XY}}{d P_{X} \times P_Y} \frac{d P_{X} \times P_0 }{d P_{X} \times P_0} \frac{dP_Y}{dP_0} dP_0 = L dP_0,
	\end{align*}
	where the equalities follow from~\cite[Ex. 32.6, p. 426]{Billingsley95} and the fact that $ \frac{d P_{X} \times P_0}{d P_{X} \times P_Y} = \frac{d P_0}{dP_Y}$. Then
	\begin{align*}
		\ml{X_0^T}{Y_0^T} = \log \int_{ \Omega_T} \essup_{P_X} L(X,y) P_0(dy).
	\end{align*}
	It is easy to verify that $\essup_{P_X} L(X,y) = \exp(y_T \log \mu +T)$. By noting that $y_T$ is distributed as $\mathrm{Poi}(T)$ under $P_0$, we get
	\begin{align*}
		\frac{1}{T} \ml{X_0^T}{Y_0^T} =  \frac{1}{T} \log \int_{ \Omega_T} \exp \left[ y_T \log \mu + T\right] P_0 (dy) = \frac{1}{T} \log \exp [T ( \mu -1) + T]  = \mu. 
	\end{align*}
	The computation of the average waiting time is standard 
(e.g.,~\cite[(3.26)]{kleinrock}).
	\end{proof}
	\paragraph{Accumulate-and-dump} Fix $\tau \in \mathbb{R}_+$ and $m \in \mathbb{N}$. Assume (for simplicity) that $\tau$ divides $T$, and consider the following scheme. The packets are accumulated, then released (``dumped'') only at integer multiples of $\tau$. If in a given interval more than $m$ packets are received, only the first $m$ are sent and the remaining ones are dropped. Then
	\begin{align} \label{eq:accanddump}
	\frac{1}{T} \ml{X_0^T}{Y_0^T} = \frac{1}{\tau} \log(m+1),
	\end{align}
	and the average waiting time for a packet is $\tau/2$ assuming it is
     not dropped. Moreover, 
	\begin{align} \label{eq:accanddumperror}
	P_e \leq e^{\lambda \tau (\nu - (1+\nu)\log(1+\nu))},
\end{align}	 
	 where $P_e$ is the probability that the number of packets exceeds $m$ in a given interval of length $\tau$, and $\nu = (m+1)/(\lambda \tau)-1$. 
	Hence, choosing $m$ to be  $(1+\hat{\nu}) \lambda \tau-1$, for some $\hat{\nu} >0$, yields 
	\[ \frac{1}{T} \ml{X_0^T}{Y_0^T} = \frac{1}{\tau}  \log((1+\hat{\nu})\lambda \tau),\]
	and a probability of dropping a packet that is exponentially small in $\hat{\nu}$. Note that, as opposed to the memoryless scheme above, accumulate-and-dump can make the leakage arbitrarily small. For a more direct comparison, suppose we wish the average waiting time to be no more than $1/\lambda$. Hence, for the memoryless scheme we choose $\mu=2\lambda$, which leads to a leakage of $2\lambda$. For the accumulate-and-dump scheme, choose $\tau=2/\lambda$ and $\hat{\nu} = e^3/2-1 (\approx 9)$. Then one can readily verify that the leakage is $3 \lambda/2$ and $P_e$ is on the order of $10^{-12}$.
	\begin{proof}
	Since the number of arrivals in a Poisson process are identically distributed and independent for non-overlapping intervals of the same length, 
	\begin{align*}
	\frac{1}{T}  \ml{X_0^T}{Y_0^T} = \frac{1}{T} \frac{T}{\tau} \ml{X_0^\tau}{Y_0^\tau} = \frac{1}{\tau} \ml{X_0^\tau}{Y_\tau} = \frac{1}{\tau} \log(m+1),
	\end{align*} where the last equality follows from the fact that $Y_\tau$ is a deterministic function of $X_{0}^{\tau}$ that takes values in $\{0,1,\dots,m\}$. To compute the average waiting time, it is enough to consider the waiting time for the packets that arrive in the first interval $[0, \tau]$. Conditioned on $X_{\tau} = N$, the arrival times are distributed as the ordered statistics of $N$ independent uniform random variables over $[0,\tau]$~\cite[Ch.~4, Theorem 4A]{Parzen99}. Therefore, the conditional average waiting time is $\tau/2$. Hence, the average waiting time is $\tau/2$. Finally, the upper bound on $P_e$ is an application of the Chernoff bound to Poisson random variables.
	\end{proof}
	\paragraph{Inject dummy packets} Song \emph{et al.}~\cite{SSHTiming} suggest using dummy packets to keep the rate of transmission fixed. That is, they use accumulate-and-dump with an extra parameter $m_b \in \mathbb{N}$. If in a given interval $N < m_b$ packets are received, we inject $m_b-N$ dummy packets. Then
	\begin{align} \label{eq:dummypacket}
	\frac{1}{T}  \ml{X_0^T}{Y_0^T}  = \frac{1}{\tau} \log(m-m_b+1).
	\end{align}
	\begin{Remark}
	In~\cite{SSHTiming}, the authors do not suggest an upper bound on the number of packets that can be released in a given interval. They implicitly assume that there exists an $m$ for which the number of arrivals in a given interval is at most $m$ almost surely. This is not true for the Poisson process, but for any process that satisfies this property, the leakage of accumulate-and-dump + inject-dummy-packets is upper-bounded by the right hand side of~\eqref{eq:dummypacket}.
	\end{Remark}
	\vspace{1mm}
	The choice of $m$ and $m_b$ provides a trade-off between the overhead of injecting dummy packets and the probability of dropping a packet. Song \emph{et al.}~\cite{SSHTiming} also point out an important drawback of the memoryless scheme. If the adversary observes several independent instances of the output for the same input (e.g., they eavesdrop several times on the same user while they are inputting their password),
	they can diminish the effect of randomization by considering the average (over the different observations) of the inter-arrival times between successive packets.
\end{Example}
Similar observations hold for the optimal mechanism for the Shannon
cipher system, to which we turn next.

\section{Shannon Cipher System} \label{sec:cipher}

The main goal in quantifying information leakage is to enable the design of mechanisms to mitigate it. As an application, we study a (traditional) secrecy setup known as
the Shannon cipher system~\cite{ShannonSecrecy}. 
 The setup consists of a transmitter and a legitimate receiver that are linked by a public noiseless channel and share a common key, and an eavesdropper who has access to the public channel and is aware of the source statistics and the used encryption schemes. The encryption schemes must allow the legitimate receiver to perfectly reconstruct the source sequence. Shannon~\cite{ShannonSecrecy} showed that perfect secrecy (i.e., making the source $X^n$ and the public message $M$ independent) requires a key rate as high as the message rate, which is typically not feasible in practice.  Hence several works~\cite{ShannonSecrecy,yamamoto1997shannoncipher,cuff2014henchman,
cuff2013secrecycausal,MerhavArikanShannonCipher,
MySecrecySystem} studied the optimal \emph{partial} secrecy achievable for a given key rate $r$, and used different measures to assess secrecy guarantees.
\begin{figure}[htp]
\centering
\includegraphics[scale=0.8]{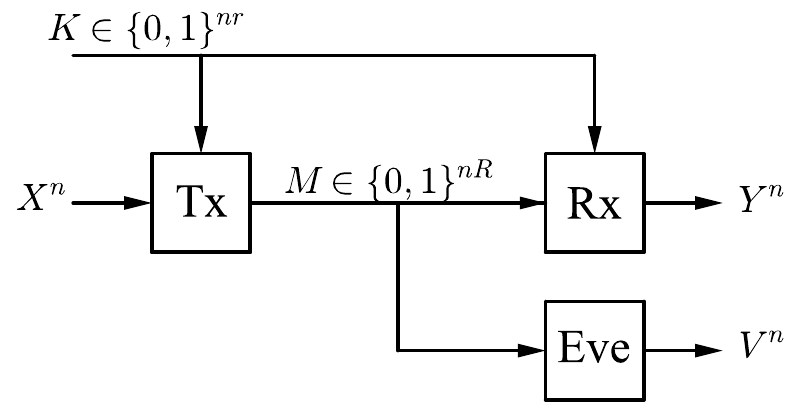}
\caption{The Shannon cipher system with lossy communication: the transmitter and the legitimate receiver have access to a common key $K$, which consists of $nr$ purely random bits, where $r$ is called the key rate. Using both the public message $M$ and the secret key $K$, the legitimate receiver generates a reconstruction $Y^n$ that should satisfy a given distortion constraint. The eavesdropper has access to the public message $M$ only.
}\label{figureShannonCipher2}
\end{figure}

In this section, we use maximal leakage to assess the performance of any feasible encryption scheme. Similarly to previous works, we are concerned with the dependence between the source and the public message (i.e., $\ml{X^n}{M}$ in our case, as opposed to the dependence between the secret key and the public message).
Moreover, we allow for lossy communication by introducing a distortion function $d$ at the legitimate receiver, as shown in Figure~\ref{figureShannonCipher2}. For a given distortion level $D$, we require that the probability of violating the distortion constraint decays as $2^{-n\alpha}$, for a given $\alpha >0$. Then, for a given $D$ and $\alpha$, we study the asymptotic behavior of the normalized maximal leakage.

For a discrete memoryless source (DMS), we derive the optimal (i.e., minimal) limit of the normalized maximal leakage. The scheme we propose for the primary user (i.e., the transmitter--legitimate receiver pair) operates on a type-by-type basis. With each type, we associate a good rate-distortion code. The codebooks are then divided into bins, and the key is used to randomize, within a bin, the choice of codeword associated with a particular source sequence. However, types with low enough probability are discarded, i.e., a dummy message is associated with all the source sequences belonging to such types. 

We also derive the optimal limit when the requirement of a decaying probability of violating the distortion constraint is replaced with an expected distortion constraint. 
In this scenario, one might expect that memoryless schemes are sufficient for optimality. We evaluate this claim by considering the case in which there is no common key and the rate of the channel is high (cf.~Figure~\ref{figurenokey}).
\begin{figure}[htp]
\centering
\includegraphics[scale=0.8]{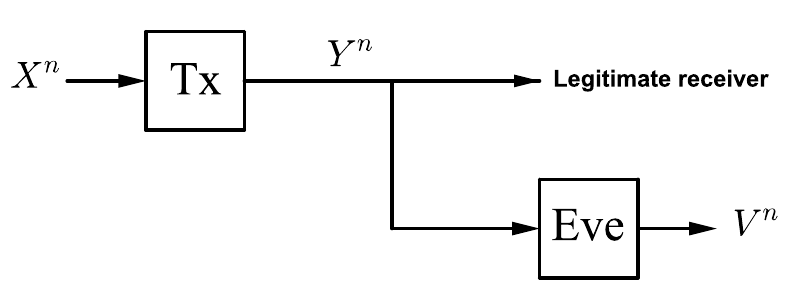}
\caption{Information blurring system.} \label{figurenokey}
\end{figure} 
This setup was dubbed the ``information blurring system'' in~\cite{MySecrecySystem}, and it represents a stylized model of side-channels. For instance, in the SSH setup,  $X^n$ could represent the timings of the incoming packets, $Y^n$ could represent the perturbed timings of the outgoing packets, and the distortion function could represent required quality (e.g., delay) constraints imposed on the system. We show that, even in this setup, memoryless schemes are strictly suboptimal in general. This strengthens our earlier observations in Section~\ref{sec:general} and suggests that commonly used memoryless schemes are generally outperformed by quantization-based schemes.

\subsection{Problem Setup and Statement of Result} \label{SecSCSResult}
Let $\mathcal{X}$ and $\mathcal{Y}$ be the alphabets associated with the transmitter and the legitimate receiver, respectively. The transmitter and the legitimate receiver are connected through a noiseless channel of rate $R$, and share common randomness $K_n \in \mathcal{K}_n = \{0,1\}^{nr}$, where $K_n$ is uniformly distributed over $\mathcal{K}_n$, and $r > 0$ is the rate of the key. The transmitter observes an $n$-length message $X^n=(X_1,X_2,\cdots,X_n)$, independent of $K_n$, and wishes to communicate it to the receiver. Let $f$ and $h$ be, respectively, the transmitter's encoding and the receiver's decoding functions. The transmitter then sends a message $M_n = f(X^n,K_n)$, $M_n \in \mathcal{M}_n =\{0,1\}^{nR}$, and the receiver generates a reconstruction $Y^n = h(M_n,K_n)$. We allow the functions $f$ and $h$ to be randomized (beyond the private randomness in $K_n$). 
For a given distortion function $d: \mathcal{X} \times \mathcal{Y} \rightarrow \mathbb{R}_+$, distortion level $D$, and excess distortion probability $\alpha$, we require that $\Pr(d(X^n,Y^n) > D) \leq 2^{-n\alpha}$, where $d(X^n,Y^n) = \frac{1}{n} \sum_{i=1}^n d(X_i,Y_i)$. 

An eavesdropper intercepts the message $M$. We assume they know the source statistics as well as the encoding and decoding functions, but do not have access to the key $K_n$. 

The primary user aims to minimize the maximal leakage to the eavesdropper $\ml{X^n}{M_n}$. We characterize the asymptotically-optimal normalized maximal leakage under the following assumptions\footnote{Note that it is necessary to have $R \geq \max_{Q: D(Q||P) \leq \alpha } R(Q,D)$ for the primary user's problem to be feasible.}:
\begin{enumerate}
\item[(A1)] The alphabets $\mathcal{X}$ and $\mathcal{Y}$ are finite.
\item[(A2)] The source is memoryless and has full support.
\item[(A3)] The distortion function $d$ is bounded, i.e., there exists $D_{\max}$ such that, for all $x \in \mathcal{X}$ and $y \in \mathcal{Y}$, $d(x,y) \leq D_{\max}$. Moreover, $D \geq D_{\min}$, where $D_{\min} = \max_{x \in \mathcal{X}} \min_{y \in \mathcal{Y}} d(x,y)$.
\item[(A4)] $R > \max_{Q: D(Q||P) \leq \alpha } R(Q,D)$, where $R(Q,D)$ is the rate-distortion function for source distribution $Q$.
\end{enumerate}
We denote the optimal limit by $L(P,D,\overrightarrow{R},\alpha)$, where $P$ is the source distribution, and $\overrightarrow{R}=(R,r)$:
\begin{align*}
L(P,D,\overrightarrow{R},\alpha) = \lim_{n \rightarrow \infty} \min_{\{f_n \in \mathcal{F}_n \}} \frac{1}{n} \ml{X^n}{f(X^n,K_n)},
\end{align*}
where $\{\mathcal{F}_n\}$ is the set of feasible schemes, i.e., $\mathcal{F}_n = \{ f_n: \mathcal{X}^n \times \{0,1\}^{nr} \rightarrow \{0,1\}^{nR} ~\big|~ \text{there exists } g: \{0,1\}^{nR} \times \{0,1\}^{nr} \rightarrow \mathcal{Y}^n \text{ satisfying } \Pr\left( d\Big(X^n, 
g\big( f \left(X^n,K_n\right), K_n \big) \Big) \right) \leq 2^{-n \alpha}   \}$. 

It will be more notationally convenient in this section to give the answers in bits rather than nats. Hence we will use the logarithm to the base 2 when computing maximal leakage. To avoid confusion, we will explicitly mention the unit we are using.

The main result of this section is the characterization of the optimal limit as follows:
\vspace{1mm}
\begin{Theorem} \label{Thmmain}
Under assumptions (A1)-(A4), for any DMS P and distortion function $d$ with associated distortion level $D \geq D_{\min}$ and distortion excess probability $\alpha > 0$:
\begin{align}
\label{eqmainthm}
L(P,D,\overrightarrow{R},\alpha) =  \max_{Q: D(Q||P) \leq \alpha} [R(Q,D) -r ]^+ \quad (bits),
\end{align}
where $[a]^+ = \max\{0,a\}$.
\end{Theorem}
\vspace{1mm}
Note that the case $\alpha=\infty$ (i.e., when the distortion constraint is imposed almost surely) is included in the theorem. Moreover, in that case, the theorem holds even if the source is not memoryless, as long as the support of $X^n$ is $\mathcal{X}^n$. This follows from the fact that $\ml{X^n}{M_n}$ and the constraint, when imposed almost surely, depend on the distribution of $X^n$ only through its support. Therefore, solving for any specific distribution on that support is equivalent to solving for all distributions on the same support.

Before proving the theorem (in Sections~\ref{SecSCSAch} and~\ref{SecSCSConv} for achievability and converse, respectively), we discuss a variation using an expected distortion constraint and its implication on the performance of memoryless schemes.

\subsection{Memoryless Schemes with Expected Distortion } \label{SecSCSExpected}
Instead of requiring a decaying probability of violating the distortion constraint, we could require that the distortion constraint holds only in expectation---as is common in many works in the literature. In that case, we modify assumption (A4) to be:
\begin{enumerate}
\item[(A4')] $R > R(P,D)$.
\end{enumerate}
\vspace{1mm}
\begin{Theorem} \label{Thmvary}
Under assumptions (A1)-(A3) and (A4'), for any DMS P and distortion function $d$ with associated distortion level $D \geq D_{\min}$:
\begin{align}
\label{eqthmvary}
L(P,D,\overrightarrow{R}) =  [R(P,D) -r ]^+  \quad (bits).
\end{align}
\end{Theorem}
\vspace{1mm}
\begin{proof}
The achievability argument follows by a similar manner as the one given in subsection~\ref{SecSCSAch}. However, instead of encoding on a type-by-type basis, we simply use a good rate-distortion code that satisfies the expected distortion requirement and divide it into bins of size $2^{nr}$. One could also derive it from Theorem~\ref{Thmmain} as follows:
\begin{align*}
L(P,D,\overrightarrow{R}) \leq \lim_{D' \rightarrow D^+} \lim_{\alpha \rightarrow 0} L(P,D',\overrightarrow{R},\alpha) = [R(P,D)-r]^+.
\end{align*}
As for the lower bound, we use the fact that $I_\infty(X;Y) \geq I(X;Y)$~\cite{Verdualpha}. This problem, with mutual information replacing maximal leakage, has already been solved by Schieler and Cuff~\cite{cuff2013secrecycausal}. More specifically, Corollary 5 of~\cite{cuff2013secrecycausal} yields  that the optimal normalized mutual information is indeed given by $[R(P,D) - r]^+$.
\end{proof}

With the expected distortion constraint, one might venture that the optimal limit is achievable with a memoryless scheme, in which the encoder passes the source through i.i.d copies of an optimal conditional distribution $P_{Y|X}$. Counter to this common intuition, and counter to the case in which leakage is measured via mutual information, we show that this is generally not the case when the objective is maximal leakage.

To that end, consider the case in which $r=0$ and $R = \log |\mathcal{Y} |$ (cf. Figure~\ref{figurenokey}). By Theorem~\ref{Thmvary}, $L(P,D) = R(P,D)$. Now define
\begin{align} \label{eqHammingMemoryless0}
L^{\text{mem}}_n(P,D)  =  \min_{P_{Y|X}: \E[d(X,Y)] \leq D} \frac{1}{n} \ml{X^n}{Y^n},
\end{align}
where $P_{X^n Y^n} = \prod_{i=1}^n P_{X_i}P_{Y_i|X_i}$. By the additive property of maximal leakage, it follows straightforwardly that $L^{\text{mem}}_n(P,D) $ does not depend on $n$, so we will drop the $n$ subscript. 
\vspace{1mm}
\begin{Lemma} \label{lemmaMemWeak}
$L^{\text{mem}}(P,D) = R(P,D)$ if and only if there exists $P_{Y|X}$ that achieves the rate-distortion function and satisfies 1) $P_{XY}(x,y)P_{XY}(x',y) > 0 \Rightarrow P_{Y|X}(y|x) = P_{Y|X}(y|x')$, and 2) $\sum_{x: P_{XY}(x,y) >0} P_X(x) = \sum_{x': P_{XY}(x',y')>0} P_X(x')$ for all $y,y' \in \mathrm{supp}(Y)$.
\end{Lemma}
\vspace{1mm}
\begin{proof}
The proof follows straightforwardly from Lemma~\ref{lemmacompareMI}.
\end{proof}
The above conditions imply that for some conditional achieving the rate-distortion function $P^\star_{Y|X}$, $L^{\text{mem}}(P,D)= \ml{X}{Y} =I(X;Y)=R(P,D) $. If $X$ has full support, then $\ml{X}{Y}= I(X;Y)  \Rightarrow \ml{X}{Y} = C(P^\star_{Y|X})$. Hence, $R(P,D) = C(P^\star_{Y|X})$. The latter equality is not a sufficient condition, however. Hence, memoryless schemes are strictly suboptimal, except  in very special cases.

We next strengthen this observation by relaxing the constraint in~\eqref{eqHammingMemoryless0} by allowing the choice of the conditional distribution $P_{Y_i|X_i}$ to depend on the index $i$. So define
\begin{align} \label{eqHammingMemoryless}
L^{\text{mem,i}}_n(P,D)  = & \min_{P_{Y^n|X^n}: P_{Y^n|X^n} = \prod_{i=1}^n P^{(i)}_{Y_i|X_i}} \frac{1}{n} \ml{X^n}{Y^n} \\ 
\text{subject to } & \E[d(X^n,Y^n)]  \leq D. \notag
\end{align}
This is still not sufficient to achieve optimality in general, as the following lemma shows.
\vspace{1mm}
\begin{Lemma} \label{LemmaMemoryless}
Suppose $X^n$ is i.i.d $\sim \mathrm{Ber}(p)$, $p \in (0,1/2]$, $d$ is the Hamming distortion, and $D \in [0,p]$. Then
\begin{align*}
L_n^{\text{mem,i}}(P,D) \geq (1-D/p) \quad (bits). 
\end{align*}
\end{Lemma}
\vspace{1mm}
On the other hand, by Theorem~\ref{Thmvary}, $L(P,D)=R(P,D) = H(p) - H(D)$.
Since $H(p) - H(D) < 1-D/p$ in general (where the inequality can be checked using convexity), memoryless schemes are strictly suboptimal.
\begin{proof}
For any $P_{Y^n|X^n}$ in the minimization, let $D_i = \E[d(X_i,Y_i)]=\Pr(X_i \neq Y_i)$. Without loss of generality, we can assume $D_i \leq p$. Then
\begin{align*}
\ml{X^n}{Y^n}  = \sum_{i=1}^n \ml{X_i}{Y_i} \geq \sum_{i=1}^n \min_{ \substack{ P_{Y_i | X_i}: \\ \Pr(X_i \neq Y_i) \leq D_i }} \ml{X_i}{Y_i}.
\end{align*}
We show in Appendix~\ref{app:maxleakhamming} that 
\begin{align} \label{eqHammingresult}
\min_{ \substack{ P_{Y_i | X_i}: \\ \Pr(X_i \neq Y_i) \leq D_i }} \ml{X_i}{Y_i} = \log_2 (2-D_i/p) \quad (bits).
\end{align}
Thus,
\begin{align*}
\ml{X^n}{Y^n} 
& \geq \sum_{i=1}^n \log_2(2-D_i/p) \\
& = \sum_{i=1}^n \log_2 \Big(2 - (D_i/p)(1)-(1-D_i/p)(0) \Big) \\
& \stackrel{\text{(a)}} \geq \sum_{i=1}^n(D_i/p) \log_2(1) +(1-D_i/p) \log_2(2) \\
& = \sum_{i=1}^n (1-D_i/p)  \\
& \stackrel{\text{(b)}} \geq n(1-D/p) ,
\end{align*}
where (a) follows from the fact that $\log_2(2-x)$ is concave in $x$, and (b) follows from the constraint in~\eqref{eqHammingMemoryless}. 
\end{proof}

\subsection{Notation}
In the following, $\mathcal{Z}$ is an arbitrary discrete set, and $Z$ is a random variable over $\mathcal{Z}$.
\begin{enumerate} 
\item[-] For a sequence $z^n \in \mathcal{Z}^n$, $Q_{z^n}$ is the empirical PMF of $z^n$, also referred to as its type.
\item[-] $\mathcal{Q}_\mathcal{Z}^n$ is the set of types in $\mathcal{Z}^n$, i.e., the set of rational PMF's with denominator $n$.
\item[-] For $Q_Z \in \mathcal{Q}_\mathcal{Z}^n$, the type class of $Q_Z$ is $T_{Q_Z} \triangleq \{z^n \in \mathcal{Z}^n: Q_{z^n}=Q_Z \}$.
\item[-] $\E_Q[\cdot]$, $H_Q(\cdot)$, and $I_Q(\cdot;\cdot)$ denote respectively expectation, entropy, and mutual information taken with respect to distribution $Q$.
\item[-] $\exp_{2} \{.\}$ denotes $2^{(\cdot)}$.
\end{enumerate}

\subsection{Achievability Proof of Theorem~\ref{Thmmain}} \label{SecSCSAch}
We will slightly abuse notation and shorten $L(P,D,\overrightarrow{R},\alpha)$ to $L$ in the following. We now show that the right-hand side of \eqref{eqmainthm} upper-bounds $L$.

Consider any $\epsilon > 0$ and let $n$ be large enough such that we can construct a rate-distortion code $\mathcal{C}_{Q_X}^n$, for each type $Q_X \in \mathcal{Q}_{\mathcal{X}}^n$, satisfying the following: each sequence $x^n \in T_{Q_X}$ is covered and $| \mathcal{C}_{Q_X}^n | \leq 2^{n(R(Q_X,D)+ \epsilon)}$. Such construction is guaranteed by the type covering lemma (Lemma 9.1 in~\cite{korner}). We divide the codebook $ \mathcal{C}_{Q_X}^n $ into  $ \left\lceil \left| \mathcal{C}_{Q_X}^n \right| /2^{nr} \right\rceil $ bins, each of size $2^{nr}$, except for possibly the last one. We denote by $ \mathcal{C}_{Q_X}^n (i,\cdot)$ the $i$th partition of the codebook, and by $ \mathcal{C}_{Q_X}^n (i,j)$ the $j$th codeword in the $i$th partition. For each $x^n \in T_{Q_X}$, let $i_{x^n}$ and $j_{x^n}$ denote, respectively, the index of the partition containing the codeword associated with $x^n$ and the index of the codeword within the partition (Note that if more than one codeword can be associated with $x^n$, we fix any one of them arbitrarily). Finally, let $m(Q_X,i,j)$ be a message consisting of the following:
\begin{itemize}
\item $\left\lceil \log_2 | \mathcal{Q}_{\mathcal{X}}^n | \right\rceil$ bits to describe the type $Q_X$. 
\item $ \left\lceil \log_2 \left\lceil \left| \mathcal{C}_{Q_X}^n \right| \big/2^{nr} \right\rceil \right\rceil$ bits to describe the index $i$, where $1 \leq i \leq \left\lceil \left| \mathcal{C}_{Q_X}^n \right| \big/2^{nr} \right\rceil $.
\item $ \left\lceil \log_2 \left| \mathcal{C}_{Q_X}^n(i,\cdot) \right| \right\rceil $ bits to describe the index $j$, where ${ 0   \leq   j   \leq   \exp_2  {\left\lceil \log_2 \left| \mathcal{C}_{Q_X}^n(i,\cdot) \right| \right\rceil}   -   1 }$.
\end{itemize}

Now, for any $\delta \in \mathbb{R}$, let $\mathcal{Q}(\alpha,\delta) = \{Q_X : D(Q_X||P) \leq \alpha + \delta \}$, $\mathcal{Q}_n(\alpha,\delta) = \{Q_X \in \mathcal{Q}_{\mathcal{X}}^n : D(Q_X||P) \leq \alpha + \delta \}$, and consider the following lemma.
\begin{Lemma} \label{lemmacontinuity}
\begin{align*}
\lim_{\delta \rightarrow 0} \max_{Q_X \in \mathcal{Q}(\alpha,\delta) } R(Q_X,D) = \max_{Q_X \in \mathcal{Q}(\alpha,0) } R(Q_X,D).
\end{align*}
\end{Lemma}
\begin{proof}
This follows directly from the convexity of $D(Q||P)$, and Propositions 12 and 13 in~\cite{MySecrecySystem}. 
\end{proof}
Now let $\delta > 0$ be such that $  \max_{Q_X \in \mathcal{Q}(\alpha,\delta)} R(Q_X,D) < R$ (Such $\delta$ exists by Lemma~\ref{lemmacontinuity} and  (A4)).  Finally, for each sequence $x^n$, let $s(x^n) = \left\lceil \log \left| \mathcal{C}_{Q_X}^n(i_{x^n},\cdot) \right| \right\rceil$, and let $K_{s(x^n)}$ be the first $s(x^n)$ bits of $K_n$. The transmitter encodes as follows. Given $x^n$, if $Q_{x^n} \in \mathcal{Q}_n (\alpha,\delta)$, then
\begin{equation} \label{eqencf}
f(x^n,K_n) = m \left(Q_{x^n}, i_{x^n}, j_{x^n} \oplus K_{s(x^n)} \right), 
\end{equation}
where the XOR-operation is performed bitwise. Note that, in this case, the legitimate receiver can retrieve the type of the transmitted sequence and the index of the bin from the first two parts of the message, and the index of the sequence within the bin using the last part of the message and the key $K_n$, so that $h(M_n,K_n) = \mathcal{C}_{\mathcal{Q}_{x^n}}^n (i_{x^n},j_{x^n})$. Now, consider an $m_0 \in \mathcal{M}_n$ that has not been used by the previous encoding (Assumption (A4) and the choice of $\delta$ ensures the existence of such $m_0$). Then, for all $x^n$ such that $Q_{x^n} \notin \mathcal{Q}_n (\alpha,\delta)$,
\begin{equation}
\label{eqencf2}
f(x^n,K_n)=m_0.
\end{equation}
\begin{Remark}
To verify that the suggested scheme satisfies the excess distortion probability constraint, consider the following:
\begin{align*}
 \Pr(d(X^n,Y^n) > D) & \leq \sum_{Q_X \notin \mathcal{Q}_n(\alpha,\delta)} P(Q) \leq \sum_{Q_X \notin \mathcal{Q}_n(\alpha,\delta)}  2^{-nD(Q_X||P)}  \leq (n+1)^{|\mathcal{X}|} 2^{-n(\alpha+\delta)}  < 2^{-n\alpha},
\end{align*}
where the last inequality holds for large enough $n$.
\end{Remark}

Effectively, we are leaking the first two parts of the message $Q_{X^n}$ and $i_{X^n}$, and \emph{hiding} completely the last part $j_{X^n}$. Since there are only polynomially many types, the first part does not affect the normalized leakage. The second part, however, consists roughly of $R(Q,D) - r$ bits, whenever $R(Q,D) > r$; otherwise, i.e., when $R(Q,D) \leq r$, there is only one bin and there is no information to be leaked.

For a more rigorous analysis, let $P_f$ be the induced joint probability distribution of $(X^n,M_n)$. Then, for $x^n$ satisfying $Q_{x^n} \in \mathcal{Q}_n (\alpha,\delta)$, we get from~\eqref{eqencf}:
\begin{equation*}
P_f   \left( m(Q_{x^n},i_{x^n},j) \big| x^n \right)  =  2^{-s(x^n)}, ~~ 0 \leq j \leq 2^{s(x^n)} - 1.
\end{equation*}
Let $S(x^n)=2^{s(x^n)}$. Note that we can equivalently denote $S(x^n)$ by $S(Q_{x^n}, i_{x^n})$, since the dependence on the sequence is only through the type and the index of the bin. Therefore, we get
\begin{align}
 \exp_2\{\ml{X^n}{M_n}\} 
 & = \sum_{m \in \mathcal{M}_n} \max_{x^n \in \mathcal{X}^n} P_f(m|x^n) \notag \\
& = \max_{x^n \in \mathcal{X}^n} P_f(m_0|x^n) + \sum_{ \substack{Q_X \in \\ \mathcal{Q}_n(\alpha,\delta)}}     
\sum_{i=1}^{\left\lceil | \mathcal{C}_{Q_X}^n | /2^{nr} \right\rceil}   \sum_{j=0}^{ S(Q_X,i)-1 }     \max_{x^n \in \mathcal{X}^n} P_f(m(Q_X,i,j)|x^n) \notag \\
& = 1 + \sum_{ \substack{Q_X \in \\ \mathcal{Q}_n(\alpha,\delta)}}     
\sum_{i=1}^{\left\lceil | \mathcal{C}_{Q_X}^n | /2^{nr} \right\rceil}   \sum_{j=0}^{ S(Q_X,i)-1 } S(Q_X,i)^{-1} \notag \\
& \leq  1 + \sum_{ \substack{Q_X \in  \mathcal{Q}_n(\alpha,\delta)}} (2^{n(R(Q_X,D)+\epsilon-r)} +1 ) \notag \\
& \leq  1 + 2    \sum_{ \substack{Q_X \in \mathcal{Q}_n(\alpha,\delta)}} 2^{ n \max \{R(Q_X,D)+\epsilon-r,0\}} \notag \\
& \leq 4(n+1)^{|\mathcal{X}|} \exp_2 \{ n      \max_{Q_X \in \mathcal{Q}_n(\alpha,\delta)}      [R(Q_X,D)+\epsilon-r]^+   \}. \label{eqachmainchain}
\end{align}
Taking the limit as $n$ tends to infinity, and noting that $\epsilon$ and $\delta$ were arbitrary, we get that
\begin{align*}
L \leq \max_{Q: D(Q||P) \leq \alpha} [R(Q,D) -r ]^+,
\end{align*}
where the inequality follows from Lemma~\ref{lemmacontinuity} and the following lemma, the simple of proof of which is omitted.
\begin{Lemma} \label{lemmalimit}
\begin{align*}
\lim_{n \rightarrow \infty} \max_{\substack{Q \in \mathcal{Q}_{\mathcal{X}}^n:  D(Q||P) \leq \alpha}}      R(Q,D) =    \max_{\substack{Q: D(Q||P) \leq \alpha}}     R(Q,D).
\end{align*}
\end{Lemma}

\subsection{Converse Proof of Theorem~\ref{Thmmain}} \label{SecSCSConv}
We now show that $ L $ is lower-bounded by the right-hand side of~\eqref{eqmainthm}. To that end, consider any valid encoding function $f$. To lower-bound $\ml{X^n}{M_n}$, we consider a specific $P_{U|X^n}$. In particular, we consider the ``shattering'' $P_{U|X^n}$ given in~\eqref{eqshattering}. Recall 
\begin{align*}
P_{U|X^n}((i_u,j_u)|x^n)  = \begin{cases}
\frac{p^{\star}}{P(x^n)}, &      i_u=x^n, ~1 \leq j_u \leq \lfloor k(x^n) \rfloor , \\
1 - \frac{(\lceil k(x^n) \rceil -1)p^{\star}}{P(x^n)}, &     i_u=x^n,~ j_u=\lceil k(x^n) \rceil, \\
0, &      i_u \neq x^n, ~ 1 \leq j_u \leq \lceil k(i_u) \rceil.
\end{cases}
\end{align*} 

Therefore, $\max_{u \in \mathcal{U}} P_U(u) = p^\star$. We will also consider a sub-optimal guessing function for $U$.  The scheme is as follows: the eavesdropper first tries to guess the key $K_n$ by choosing an element uniformly at random from $\{0,1\}^{nr}$. We denote this guess by $\tilde{K}_n$. Then, proceeding by assuming that the key guess was correct, they try to guess the sequence $x^n$ using a guessing function given by Lemma~\ref{lemmaguessing} below. We denote this stage by $g_1$. Finally, again proceeding by assuming that the source sequence guess was correct, the eavesdropper attempts to guess $U$ by using the MAP rule. We denote this stage by $g_2$, and we get for each $x^n \in \mathcal{X}^n$,
\begin{align}
g_2(x^n) & = (x^n, 1), 
\text{ and } \Pr(g_2(x^n)  = U^n|x^n)  = p^\star/P(x^n). \label{eqmap}
\end{align}
\vspace{1mm}
\begin{Lemma} \label{lemmaguessing}
There exists a function $g_1: {\mathcal{Y}^n \rightarrow \mathcal{X}^n}$ such that, for all $(x^n,y^n) $ satisfying $d(x^n,y^n) \leq D,$ 
$\Pr \left(x^n = g(y^n) \right) \geq c_n 2^{-n(H_{Q_{x^n}}(X)-R(Q_{x^n},D))}$, where $c_n=(n+1)^{ -|\mathcal{X}||\mathcal{Y}|(|\mathcal{X}|+1)}$.
\end{Lemma}
\vspace{1mm}
\begin{proof}
This is an application of Lemma 5 in~\cite{MySecrecySystem}. In particular, we set in Lemma 5 $\mathcal{V}$ to be $\mathcal{X}$, $d_e$ to be the Hamming distortion function, and $D_e$ to be zero. Then, $I_{P^\star_n(Q_{x^n y^n})}(X;V|Y)$ (as defined in~\cite{MySecrecySystem}) satisfies:
\begin{align*}
I_{P^\star_n(Q_{x^n y^n})}(X;V|Y)  = H_{Q_{x^ny^n}} (X|Y) 
& = H_{Q_{x^n}}(X) - H_{Q_{x^n}}(X) + H_{Q_{x^ny^n}} (X|Y) \\
& \leq H_{Q_{x^n}}(X)-R(Q_{x^n},D). 
\end{align*}  
\end{proof}
To analyze the above scheme, fix $\epsilon >0$, and let $P_f$ denote the induced joint probability on $(X^n,K_n,M_n)$. Furthermore, without loss of generality, we can assume that the decoding function $h$ is a deterministic function of $M_n$ and $K_n$. Finally, define
\begin{align}
\mathcal{M}_D(x^n,k) & = \{ m \in \mathcal{M}_n: d(x^n,h(m,k)) \leq D\}, ~~ x^n \in \mathcal{X}^n, k \in \mathcal{K}_n, \label{eqdefMD} \\
\text{and } \mathcal{A} & = \{ (x^n,y^n)   \in    \mathcal{X}^n    \times    \mathcal{Y}^n    :     d(x^n,y^n)    >     D\}. 
\end{align} 
Letting $g$ be the concatenation of the two stages, we get 
\begin{align}
& \Pr(U = g(M) ) \notag \\
& = \sum_{x^n \in \mathcal{X}^n} \sum_{u \in \mathcal{U}} \sum_{k \in \mathcal{K}_n} \sum_{m \in \mathcal{M}_n}       P(x^n) P_{U|X^n}(u|x^n) P_{K_n}(k) 
 P_f (m|x^n,k) P(u = g(m)|x^n,m,k ) \notag \\
& \geq \sum_{x^n \in \mathcal{X}^n} \sum_{u \in \mathcal{U}} \sum_{k \in \mathcal{K}_n} \sum_{m \in \mathcal{M}_D(x^n,k)}  P(x^n) P_{U|X^n}(u|x^n)    P_{K_n}(k) P_f (m|x^n,k) P(u = g(m)|x^n,m,k ) \notag \\
& \geq \sum_{x^n \in \mathcal{X}^n} \sum_{u \in \mathcal{U}} \sum_{k \in \mathcal{K}_n} \sum_{m \in \mathcal{M}_D(x^n,k)} P(x^n) P_{U|X^n}(u|x^n)    P_{K_n}(k) P_f (m|x^n,k) P(\tilde{K}_n=k) . \notag \\
& \qquad \qquad \qquad \qquad \qquad P(g_1(h(m,k))=x^n) P(g_2(x^n) = u | x^n) \notag \\
& \stackrel{(a)} \geq c_n   \sum_{x^n \in \mathcal{X}^n}  \sum_{k \in \mathcal{K}_n} \sum_{m \in \mathcal{M}_D(x^n,k)}  P(x^n) P_{K_n}(k) P_f (m|x^n,k)    2^{-nr} 2^{-n(H_{Q_{x^n}}(X)-R(Q_{x^n},D))} p^\star/P(x^n) \notag \\
&  = c_n p^\star 2^{-nr} \sum_{Q_X \in \mathcal{Q}_{\mathcal{X}}^n}  \sum_{x^n \in T_{Q_X}} \sum_{k \in \mathcal{K}_n}  \sum_{m \in \mathcal{M}_D(x^n,k)} P(x^n) P_{K_n}(k) P_f (m|x^n,k) 2^{-n(H_{Q_X}(X)-R(Q_{X},D))}/P(x^n) \notag \\
&  = c_n p^\star 2^{-nr} \sum_{Q_X \in \mathcal{Q}_{\mathcal{X}}^n}  \sum_{x^n \in T_{Q_X}} \sum_{k \in \mathcal{K}_n}  \sum_{m \in \mathcal{M}_D(x^n,k)}  P(x^n) P_{K_n}(k) P_f (m|x^n,k) 2^{n(R(Q_X,D)+D(Q_X||P))} \notag \\
& = c_n p^\star 2^{-nr} \sum_{Q_X \in \mathcal{Q}_{\mathcal{X}}^n} 2^{n(R(Q_X,D)+D(Q_X||P))}  P_f( \mathcal{A}^c \cap T_{Q_X}), \label{eqconvmainchain1}
\end{align}
where (a) follows from Lemma~\ref{lemmaguessing}, \eqref{eqmap}, and~\eqref{eqdefMD}.
Now, note that for any $Q$, 
\begin{align*}
 P_f(\mathcal{A}^c|T_Q) = 1 - P_f(\mathcal{A}|T_Q) 
& \geq 1 - \min\{1,P_f(\mathcal{A})/P(T_Q)\} \\
& \geq 1   -   \min\{1, 2^{-n(\alpha - D(Q||P) - \frac{|\mathcal{X}|}{n} \log(n+1))}\} \\
& =   \max\{0, 1 - 2^{-n(\alpha - D(Q||P) - \frac{|\mathcal{X}|}{n} \log(n+1))}   \}.
\end{align*}
Then, continuing~\eqref{eqconvmainchain1}, we get
\begin{align}
  \Pr(U = g(M) ) 
& \geq c_n p^\star 2^{-nr} \sum_{Q_X \in \mathcal{Q}_{\mathcal{X}}^n} 2^{n(R(Q_X,D)+D(Q_X||P))} P(T_{Q_X}) \max\{0, 1-2^{-n(\alpha - D(Q_X ||P) - \frac{|\mathcal{X}|}{n} \log(n+1))} \} \notag \\
& \stackrel{(a)} \geq c'_n p^\star 2^{-nr} \sum_{Q_X \in \mathcal{Q}_n(\alpha,-\epsilon)} 2^{nR(Q_X,D)} (1-2^{-n(\alpha - D(Q_X ||P) - \frac{|\mathcal{X}|}{n} \log(n+1))}) \notag \\
& \stackrel{(b)} \geq c'_n p^\star 2^{-nr} \sum_{Q_X \in \mathcal{Q}_n(\alpha,-\epsilon)} 2^{nR(Q_X,D)} (1/2) \notag \\
& \geq (c'_n p^\star /2) \max_{Q_X \in \mathcal{Q}_n(\alpha,-\epsilon)} \exp_2\{n(R(Q_X,D)-r)\},
\end{align}
where (a) and (b) hold for large enough $n$, and $c'_n={(n+1)^{-|\mathcal{X}|}c_n}$. Finally taking the ratio of $\Pr(U=g(M))$ and $\max_u P_U(u)$, and taking the limit as $n$ tends to infinity, and noting that $\epsilon$ is arbitrary, we get
\begin{align*}
L \geq \max_{Q: D(Q||P) \leq \alpha} R(Q,D) -r,
\end{align*}
where the inequality follows from Lemmas~\ref{lemmacontinuity} and~\ref{lemmalimit}. 
Since $L$ is positive by definition,
\begin{align*}
L \geq [\max_{Q: D(Q||P) \leq \alpha} R(Q,D) -r ]^+  = \max_{Q: D(Q||P) \leq \alpha} [R(Q,D) -r ]^+.
\end{align*}

\section{Learning Maximal Leakage from Data} \label{sec:learning}

In the previous two sections, we analyzed leakage-mitigating schemes for a simple model of the SSH side-channel and derived the (asymptotically) optimal scheme for the Shannon cipher system. In general, computing the maximal leakage induced by a given scheme might become intractable for complicated schemes.  Consider, for instance, an on-chip network with several processes sharing its memory. Suppose one of the processes is malicious and another process is decrypting a message using a secret key. As we mentioned in the introduction, a side-channel exists between these two processes because the memory access patterns of the latter affect the memory access delays of the former. This side-channel, however, is determined by the operation of the memory controller which could depend on many variables, as well as the behavior of other processes on the chip which might be difficult to model.

For such complicated schemes, one might simulate the system and attempt to estimate maximal leakage from data traces. This section investigates the complexity of this task, i.e., the number of samples needed to estimate $\ml{X}{Y}$, which we equivalently denote by $\mathcal{L}(P_X;P_{Y|X})$. 
To this end, an estimator is defined as a randomized function $f: (\mathcal{X} \times \mathcal{Y})^\star \rightarrow \mathbb{R}$, which maps a sequence of samples drawn from a joint distribution to an estimate of its maximal leakage.  Given a desired level of accuracy $\delta$ and a probability of error $\epsilon$, the sample complexity of an estimator $f$ is defined as:
\begin{align}
\label{eqsamplecomplexityf}
S_{\delta, \epsilon} \big( |\mathcal{X}|, |\mathcal{Y}| \big) [f] = \min \{n: P_{XY} \left( \left| \mathcal{L}(P_X;P_{Y|X}) - f(X^n,Y^n) \right| > \delta \right) < \epsilon, \text{ for all } P_{XY} \in \mathcal{P}_{\mathcal{X} \times \mathcal{Y}} \}, 
\end{align}
where $P_{XY} \in \mathcal{P}_{\mathcal{X} \times \mathcal{Y}}$ is the set of all probability distributions on $\mathcal{X} \times \mathcal{Y}$, and $(X^n,Y^n)$ are drawn independently from $P_{XY}$. Then the sample complexity of maximal leakage is defined as:
\begin{align}
\label{eqsamplecomplexity}
S_{\delta, \epsilon} \big( |\mathcal{X}|, |\mathcal{Y}| \big) = \inf_f S_{\delta,\epsilon} \big( |\mathcal{X}|, |\mathcal{Y}|,\theta \big)[f] .
\end{align}

We show that $S_{\delta, \epsilon} ( |\mathcal{X}|, |\mathcal{Y}| )$ turns out to be infinity for interesting values of the parameters. Hence, the design of secure systems should take amenability to analysis into consideration. That is, it is preferable to design, for instance, a memory controller that we can study analytically, rather than one that follows complicated ad-hoc rules that are (only) believed to be secure. \\

The impossibility result is mainly due to the discontinuity of maximal leakage in the support of $X$. More precisely, let $\theta$ be a lower bound on the minimum strictly positive probability of an element in $\mathcal{X}$, and define
\begin{align}
\label{eqPtheta}
\mathcal{P}^\theta_{\mathcal{X} \times \mathcal{Y}} & = \{ P_{XY} \in \mathcal{P}_{\mathcal{X} \times \mathcal{Y}}: 
\min_{x \in \mathcal{X}:  P_X(x) >0}  P_X(x) \geq \theta \}, \\
\label{eqsamplecomplexitytheta}
S_{\delta, \epsilon} \big( |\mathcal{X}|, |\mathcal{Y}|,\theta \big) & = \! \inf_f \min \{n \! : P_{XY} \left( \left| \mathcal{L}(P_X;P_{Y|X}) \! - \! f(X^n,Y^n) \right| \! > \delta \right) \! < \epsilon, \! \text{ for all } P_{XY} \! \in \! \mathcal{P}^\theta_{\mathcal{X} \times \mathcal{Y}}  \}.
\end{align}
Then the following lower bound holds. 
\vspace{1mm}
\begin{Theorem} \label{thmlowerbound}
For $\epsilon = 0.1$ and $c_0 < 1/2$ there exists $c$ such that for all
$\theta$, all $|\mathcal{X}|$, all sufficiently large $|\mathcal{Y}|$,
and all $1/\mathcal{|Y|} < \delta < c_0$, we have
\begin{equation}
\label{eq:lowerJA}
S_{\delta,\epsilon}(|\mathcal{X}|,|\mathcal{Y}|, \theta)
   \ge c \frac{\theta |\mathcal{Y}|}{\log |\mathcal{Y}|} \log^2 \frac{1}{\delta}.
\end{equation}
\end{Theorem}
\vspace{1mm}
If $\theta \rightarrow 0$, the bound diverges to infinity, which justifies our earlier claim that $S_{\delta, \epsilon} ( |\mathcal{X}|, |\mathcal{Y}| )$ is $ + \infty$.   
Nevertheless, if a lower bound $\theta$ is known, then the following upper bound holds.
\vspace{1mm}
\begin{Theorem} \label{thmupperbound}
For all $\theta \in (0,1)$, finite alphabets $\mathcal{X}$ and $\mathcal{Y}$, $\delta >0$, and $\epsilon \in (0,1)$,
\begin{align}
\label{equpperbound}
S_{\delta, \epsilon} \big( |\mathcal{X}|, |\mathcal{Y}|,\theta \big) \leq  \frac{ 8 \big(\log (5/\epsilon) + |\mathcal{Y}|  \log |\mathcal{X}|\big) }{ \theta \big( (2-e^{-\delta}) \log (2-e^{-\delta}) +e^{-\delta}-1 \big)  }.
\end{align}
\end{Theorem}
\vspace{1mm}
For small $\delta$, the denominator behaves as $\delta^2$. If $\theta$ is of the order of $1/|\mathcal{X}|$, we get $S(\theta,\delta, \epsilon) \leq O \left( |\mathcal{X}|( |\mathcal{Y}|  \log |\mathcal{X}| + \log(1/\epsilon) )/\delta^2 \right)$.
\vspace{1mm}
\begin{Remark}
In terms of the dependence on the alphabets and $\theta$, the upper and lower bounds are within sub-polynomial factors of each other.
\end{Remark}
\vspace{1mm}

We prove the achievability result in Section~\ref{seclearningupper} and the converse result in Section~\ref{seclearninglower}. Both proofs use the standard technique of Poisson sampling, so we now clarify the connection between Poisson and fixed-length sampling.

\subsection{Poisson Sampling}

With Poisson sampling, for a given $n$, we first generate $N \sim \Poi(n)$  and then generate $(X^N,Y^N)$ from $P_{XY}$. So we define the Poisson sample complexity as follows:
\begin{align} \label{eq:PoissonSampling}
\widetilde{S}_{\delta, \epsilon} \big( |\mathcal{X}|, |\mathcal{Y}|,\theta \big)  = \inf_f  \min  \{n:  N \sim \Poi (n),  \Pr \left( \left| \mathcal{L}(P_X;P_{Y|X}) - f(X^N,Y^N) \right| > \delta \right) < \epsilon, \text{ for all } P_{XY} \in \mathcal{P}^\theta_{\mathcal{X} \times \mathcal{Y}}  \}.
\end{align}
The following lemma will be useful for our analysis. It is a simple application of the Chernoff bound, hence its proof is omitted.
\vspace{1mm}
\begin{Lemma} \label{lemmapoissondeviation}
Consider $\delta \in (0,1)$, $\lambda >0$, and let $N \sim \Poi(\lambda)$.
\begin{align} \label{eqlemmapoissonrighttail}
\Pr(N \geq (1+\delta) \lambda ) & \leq \exp\{ \lambda( \delta - (1+\delta) \log (1+\delta))\}, \\
\intertext{and } \label{eqlemmapoissonlefttail}
\Pr(N \leq (1-\delta) \lambda) & \leq \exp \{ \lambda ( -\delta - (1-\delta) \log (1-\delta))\}.
\end{align} 
\end{Lemma}
\vspace{1mm}
\begin{Remark}
It is a simple exercise to check that the exponents are negative for all $\delta \in (0,1)$. 
\end{Remark}
\vspace{1mm}
We now show that fixed-length sampling and Poisson sampling are equivalent, up to constant factors.
\begin{Lemma} \label{lemmaPoissonFixedequiv}
Fix $\epsilon \in (0,1)$, $\delta > 0$, and $\theta \in (0,1)$. Suppose there exists $f$ such that, given $n_1 \geq \log(5/\epsilon)/\log(4/e)$ and $N \sim \Poi(n_1)$, 
\[ \Pr \left( \left| \mathcal{L}(P_X;P_{Y|X}) - f(X^N,Y^N) \right| > \delta \right) < \frac{4\epsilon}{5}. \]
Then
\begin{align}
S_{\delta, \epsilon} \big( |\mathcal{X}|, |\mathcal{Y}|,\theta \big) \leq  2 n_1.
\end{align}
On the other hand, if there exists $n_2 \geq  \log(1/\epsilon) / \log(e/2)$ such that, for all estimators $f$,
\[ \Pr \left( \left| \mathcal{L}(P_X;P_{Y|X}) - f(X^N,Y^N) \right| > \delta \right) > 2 \epsilon, \]
where $N \sim \Poi(n_2)$, then
\begin{align}
S_{\delta, \epsilon} \big( |\mathcal{X}|, |\mathcal{Y}|,\theta \big) \geq \frac{n_2}{2}.
\end{align}
\end{Lemma}
\vspace{1mm}
\begin{proof}
 Consider an optimal fixed-length estimator with $2n_1$ samples. Then, a $\Poi(n_1)$ estimator can outperform it only if $N > 2n_1$. However, by Lemma~\ref{lemmapoissondeviation},
\begin{align*} 
\Pr ( \Poi(n_1) > 2 n_1 ) \leq e^{-n_1 (2 \log 2 -1)} \leq \epsilon/5.
\end{align*} 

Conversely, consider an optimal fixed-length estimator $n_2/2$ samples. Then, it can outperform a $\Poi(n_2)$ estimator only if $N < n_2/2$. However, by Lemma~\ref{lemmapoissondeviation},
\begin{align*}
\Pr( \Poi(n_2) < n_2/2 ) \leq e^{-n_2 (1+\log(1/2))} \leq \epsilon.
\end{align*}
\end{proof}

\subsection{Proof of Theorem~\ref{thmupperbound}} \label{seclearningupper}

Let
\begin{align}
\label{eqMXY}
M(P_{X};P_{Y|X}) := \exp\{\mathcal{L} (P_{X};P_{Y|X})\} = \sum_{y \in  \mathcal{Y}} \max_{ \substack{ x \in \mathcal{X}: \\ P_X(x) > 0}} P_{Y|X}(y|x).
\end{align} 
It is straightforward to verify that a $(1-e^{-\delta})$-multiplicative estimator for $M(P_{X};P_{Y|X})$ translates to a $\delta$-additive estimator for $\mathcal{L} (P_{X};P_{Y|X})$, where a $\hat{\delta}$-multiplicative estimator means that $| M - \hat{M}| \leq \hat{\delta} M$. Therefore, in the remainder, we will analyze multiplicative estimators of $M$.
Now,
consider $n \in \mathbb{N}$ and let $N \sim \Poi(n)$. 
 Let $(X_1,Y_1), (X_2,Y_2), \dots, (X_N,Y_N)$ be $N$ independent samples drawn from a distribution $P_{XY}$. For each $x \in \mathcal{X}$ and $y \in \mathcal{Y}$, let $N_x$ denote the number of times $x$ appears, $N_y$ the number of times $y$ appears, and $N_{x,y}$ the number of times $(x,y)$ appears in the sequence. Then, $N_x \sim \Poi(nP_X(x))$, $N_y \sim \Poi(nP_Y(y))$, and $N_{x,y} \sim \Poi(nP_{XY}(x,y))$. Now, let $\theta' = \theta/4$. The estimator works as follows:
\begin{enumerate}
\item For each $x \in \mathcal{X}$ with $N_x >0$, generate a random variable $\widetilde{N}_x \sim \Poi(n\theta')$. If $N_x = 0$, set $\widetilde{N}_x =0$.
\item For each $x \in \mathcal{X}$ with $N_x > 0$, keep only the first $\widetilde{N}_x$ samples containing $x$ and disregard the rest. 
\begin{enumerate}
\item If there are not enough samples for some $x$ (i.e., $\widetilde{N}_x > N_x$), then let $\hat{M}=1$.
\item Otherwise, let 
\begin{align}
\label{eqestimator}
\hat{M} = \sum_{y \in \mathcal{Y}} \max_{x \in \mathcal{X}} \frac{\widetilde{N}_{x,y}}{n\theta'},
\end{align}
where $\widetilde{N}_{x,y}$ is the number of times $(x,y)$ appears in the truncated sequence.
\end{enumerate}
\end{enumerate}
To analyze the above estimator, we first consider a slightly modified setting. In particular, suppose the estimator has access to an infinite sequence $(X_1,Y_1), (X_2,Y_2), \dots$ Then, $N_x = +\infty$ with probability $1$ for each $x \in \mathrm{supp}(X)$. 
In this case, for each $(x,y)$ with $P_X(x) >0$, $\widetilde{N}_{x,y} \sim \Poi(n \theta' P_{Y|X} (y|x))$. 
For each $y \in \mathcal{Y}$, let $x(y) \in \argmax_{x: P_X(x) >0} P_{Y|X}(y|x)$.  Let $\hat{\delta} = 1-e^{-\delta}$, and consider the following:
\begin{align} 
\Pr \left( \hat{M} - M \leq -\hat{\delta} M \right) & = \Pr \left( \sum_{y \in \mathcal{Y}} \max_{x \in \mathcal{X}} \widetilde{N}_{x,y}/n\theta' \leq (1-\hat{\delta})M \right)   \notag  \\
& =  \Pr \left( \sum_{y \in \mathcal{Y}} \max_{x \in \mathcal{X}} \widetilde{N}_{x,y} \leq (1-\hat{\delta})M n \theta' \right)  \notag \\
& \leq \Pr \left( \sum_{y \in \mathcal{Y}}  \widetilde{N}_{x(y),y} \leq (1-\hat{\delta})M n \theta' \right)  \notag \\
& \stackrel{(a)} = \Pr \left(  \Poi \left(n \theta' M \right) \leq (1-\hat{\delta})M n \theta' \right)  \notag \\
& \stackrel{(b)} \leq  \exp \left\lbrace M n \theta' \left( -\hat{\delta} - (1-\hat{\delta}) \log (1-\hat{\delta}) \right) \right\rbrace  \notag \\
& \stackrel{(c)} \leq \exp \left\lbrace n \theta' \left( -\hat{\delta} - (1-\hat{\delta}) \log (1-\hat{\delta}) \right) \right\rbrace , \label{eqlefttailbound}
\end{align}
where (a) follows from the fact that $\widetilde{N}_{x,y}$'s are independent $\Poi\left(n \theta' P_{Y|X}(y|x(y))\right)$, (b) follows from Lemma~\ref{lemmapoissondeviation}, and (c) follows from the fact that $M \geq 1$.  
Now consider the probability that $\hat{M}$ exceeds $M$ by a factor of at least $\hat{\delta} M$:
\begin{align}
\Pr \left( \hat{M} - M \geq \hat{\delta} M \right) & = \Pr \left( \sum_{y \in \mathcal{Y}} \max_{x \in \mathcal{X}} \widetilde{N}_{x,y} \geq (1+\hat{\delta})Mn\theta' \right)   \notag \\
& = \Pr \left( \bigcup_{ (x_1,\dots,x_{|\mathcal{Y}|}) \in \mathcal{X}^{|\mathcal{Y}|} } \left( \sum_{y \in \mathcal{Y}} \widetilde{N}_{x_y,y} \geq (1+\hat{\delta})Mn\theta' \right) \right) \notag  \\
& = \Pr \left( \bigcup_{ (x_1,\dots,x_{|\mathcal{Y}|}) \in \mathcal{X}^{|\mathcal{Y}|} } \left( \Poi \left(n \theta' \sum_{y \in \mathcal{Y}} P_{Y|X}(y|x_y) \right) \geq (1+\hat{\delta})Mn\theta' \right) \right)  \notag \\
& \stackrel{(a)} \leq |\mathcal{X}|^{|\mathcal{Y}|} \Pr \left( \Poi(n \theta'M) \geq (1+\hat{\delta}) M n \theta' \right)  \notag \\
& \stackrel{(b)} \leq |\mathcal{X}|^{|\mathcal{Y}|} \exp \left\lbrace M n \theta'\left( \hat{\delta} - (1+\hat{\delta}) \log (1+\hat{\delta}) \right) \right\rbrace \notag \\
& \stackrel{(c)} \leq |\mathcal{X}|^{|\mathcal{Y}|} \exp \left\lbrace n \theta'\left( \hat{\delta} - (1+\hat{\delta}) \log (1+\hat{\delta}) \right) \right\rbrace, \label{eqrighttailbound}
\end{align}
where (a) follows from Lemma~\ref{lemmatwopoissons} below and the fact that for any $(x_1,\dots,x_{|\mathcal{Y}|}) \in \mathcal{X}^{|\mathcal{Y}|}$, $\sum_{y \in \mathcal{Y}} P_{Y|X}(y|x_y) \leq M$, (b) follows from Lemma~\ref{lemmapoissondeviation}, and (c) follows from the fact that $M \geq 1$. 
\begin{Lemma} \label{lemmatwopoissons}
Consider $\lambda_1 > 0$, $\lambda_2 >0$ such that $\lambda_1 \geq \lambda_2$, and let $N_1 \sim \Poi(\lambda_1)$, and $N_2 \sim \Poi(\lambda_2)$. Then, for all $k$,
\begin{align*}
\Pr(N_1 \geq k) \geq \Pr(N_2 \geq k).
\end{align*}
\end{Lemma}
\vspace{1mm}
The proof follows from a simple coupling argument and is omitted. Now let 
\begin{align} \label{eqnstarupper}
n^\star = \frac{\log (5/\epsilon) + |\mathcal{Y}|  \log |\mathcal{X}|}{ \theta' \left( (1+\hat{\delta}) \log (1+\hat{\delta}) - \hat{\delta} \right)  } .
\end{align} 
For such a choice, we get by~\eqref{eqlefttailbound} and~\eqref{eqrighttailbound},
\begin{align} \label{eqleftrightbound}
\Pr \left( | \hat{M} - M | \geq \hat{\delta} M \right) \leq 2\epsilon/5.
\end{align}
\begin{Remark}
For all $\hat{\delta} \in (0,1)$, $\hat{\delta} + (1-\hat{\delta}) \log ( 1- \hat{\delta}) \geq (1+\hat{\delta}) \log (1+\hat{\delta}) - \hat{\delta}$.
\end{Remark}
Note that the Poisson estimator behaves identically to the infinite-sequence estimator unless there exists $x \in \mathrm{supp}(X)$ for which $N_x = 0$ or $ \widetilde{N}_x > N_x$. Therefore, we need to compute the probability of that event.
\begin{align}
\Pr \left(\text{there exists }x \in \mathrm{supp}(X): \widetilde{N}_x > N_x \right) & \leq \sum_{x \in \mathrm{supp}(X)} \Pr(\widetilde{N}_x > N_x) \notag \\
& \leq \sum_{x  \in \mathrm{supp}(X) } \Pr\left( \Poi(n^\star \theta') \geq \Poi \left(n^\star P_X(x) \right) \right)  \notag \\
& \stackrel{(a)} \leq \sum_{x  \in \mathrm{supp}(X) } \exp\left\lbrace-\left(\sqrt{n^\star P_X(x)} - \sqrt{n^\star\theta'}\right)^2 \right\rbrace  \notag \\
&  \stackrel{(b)} \leq \sum_{x  \in \mathrm{supp}(X) } \exp\left\lbrace-\left(\sqrt{4n^\star\theta'} - \sqrt{n^\star\theta'}\right)^2 \right\rbrace  \notag \\
& \leq |\mathcal{X}| e^{-n^\star \theta'} \notag \\
& \stackrel{(c)} \leq \epsilon/5, \label{eqPoissoninfinite}
\end{align}
where (a) follows from the Chernoff bound, (b) follows from the fact that for all $x \in \mathrm{supp}(X)$, $P_X(x) \geq \theta = 4\theta'$, and (c) follows from the fact that $(1+\hat{\delta}) \log (1+\hat{\delta}) - \hat{\delta} < 2 \log 2 - 1 < 1$ for $\hat{\delta} \in (0,1)$. Similarly,
\begin{align}
\Pr \left( \text{there exists } x \in \mathrm{supp}(X): N_x = 0 \right) \leq \!\!\!\!\! \sum_{x \in \mathrm{supp}(X)} \!\!\! \Pr( N_x =0) = \!\!\!\!\! \sum_{x \in \mathrm{supp}(X)} \!\!\!\!\! e^{-n^\star P_X(x) } & \leq |\mathcal{X}| e^{-n^\star \theta'} \notag \\
& \leq \epsilon/5. \label{eqPoissonNxzero}
\end{align}
It follows from equations~\eqref{eqleftrightbound},~\eqref{eqPoissoninfinite},~\eqref{eqPoissonNxzero}, and Lemma~\ref{lemmaPoissonFixedequiv} that
\begin{align*} 
S_{\delta, \epsilon} ( |\mathcal{X}|, |\mathcal{Y}|,\theta)  \leq 2 \frac{\log (5/\epsilon) + |\mathcal{Y}|  \log |\mathcal{X}|}{ \theta' \left( (1+\hat{\delta}) \log (1+\hat{\delta}) - \hat{\delta} \right)  } .
\end{align*}
\begin{Remark}
One can readily verify that $n^\star \geq \log(5/\epsilon)/\log(4/\epsilon)$.
\end{Remark}
Plugging in  $\hat{\delta}= 1-e^{-\delta}$ and $\theta'= \theta/4$ yields Theorem~\ref{thmupperbound}.

\begin{Remark}
The proof shows that the risk of overestimating leakage is what controls
the sample complexity of the estimator. If one is merely interested in
ensuring that the estimator does not underestimate the true leakage, 
which is often the case in practice, then from (\ref{eqlefttailbound}) and 
(\ref{eqPoissoninfinite}) 
the sample complexity is
$$
\frac{ 8 \big(\log (5/\epsilon) +  \log |\mathcal{X}|\big) }{ \theta \big( (2-e^{-\delta}) \log (2-e^{-\delta}) +e^{-\delta}-1 \big)  }.
$$
\end{Remark}

\subsection{Proof of Theorem~\ref{thmlowerbound}} \label{seclearninglower}

Let $|\mathcal{Y}|=k$. We will derive a lower-bound on complexity by considering a subproblem, i.e., we will restrict our attention to a subset of $\mathcal{P}_{\mathcal{X} \times \mathcal{Y}}^\theta$ (cf.~\eqref{eqPtheta}). In particular, consider $P_{XY} \in \mathcal{P}_{\mathcal{X} \times \mathcal{Y}}^\theta$ that satisfy $P_X(x_1)=\theta \in (0,1)$ and $P_{Y|X}$ that have the following form:
\begin{align}
\label{eqPY|X}
P_{Y|X} = \begin{bmatrix}
p_1 & p_2 & \cdots & p_{k} \\
1/k & 1/k & \cdots & 1/k \\
\vdots & \vdots & & \vdots \\
1/k & 1/k & \cdots & 1/k
\end{bmatrix},
\end{align}
where  $p_Y=(p_1,p_2,\cdots,p_k)$ is some distribution over $\mathcal{Y}$. 
Now, for any distribution $p_Y$ over $\mathcal{Y}$, define
\begin{align}
\label{eqdefh}
h(p_Y) =\log \left( \sum_{y \in \mathcal{Y}} \max \left\lbrace \frac{1}{k}, ~ p_y \right\rbrace  \right).
\end{align}
Therefore, 
\begin{align}
\label{eqMLP}
\mathcal{L} (P_{X};P_{Y|X}) = h(P_{Y|X}(\cdot|x_1)).
\end{align}
Hence, estimating maximal leakage for this subproblem is the same as estimating a property of $P_{Y|X}(\cdot|x_1)$. 
Let
\begin{align}
\widetilde{S}_{\delta, \epsilon}^h (|\mathcal{Y}|) & = \inf_f \min \{n: N \sim \Poi (n),  \Pr \left( \left| h(P_Y) - f(Y^N) \right| > \delta \right) < \epsilon, \text{ for all } P_{Y} \in \mathcal{P}_{ \mathcal{Y}}  \},
\end{align}
where $\mathcal{P}_{\mathcal{Y}}$ is the set of all probability distributions on $\mathcal{Y}$, and $Y^n$ is drawn independently according to $P_Y$.
Since sampling $\Poi(n)$ from $P_{XY}$ gives $\Poi(n\theta)$ samples from $P_{Y|X}(\cdot|x_1)$, we get
\begin{align}
\label{eqcomplexitylowerboundequi}
\widetilde{S}_{\delta, \epsilon} ( |\mathcal{X}|, |\mathcal{Y}|,\theta)) \geq \widetilde{S}^h_{\delta, \epsilon} ( |\mathcal{Y}|) /\theta.
\end{align}

It remains to show that
\begin{align}
\widetilde{S}_{\delta, \epsilon}^h (k) 
   \ge c \cdot \frac{k}{\log k} \log^2 \frac{1}{\delta}.
\end{align}
We shall show this by relating the problem of estimating $h(p_Y)$
to the problem of estimating the support size of $p_Y$. Consider
any distribution $p_Y$ with the property that 
\begin{equation}
\label{eq:minprobconstraint}
p_Y(y) \ge 1/k \ \text{for all $y$ such that} \ p_Y(y) > 0.
\end{equation}
Then we have
\begin{align}
\nonumber
e^{h(p_Y)} & = \sum_{y \in \mathcal{Y}} \max\left\{\frac{1}{k},p_Y(y)\right\} \\
\nonumber
        & = \sum_{y: p_Y(y) = 0} \frac{1}{k} + 
              \sum_{y: p_Y(y) \ge 1/k}
                    p_Y(y) \\
\nonumber
        & = \frac{k - |\mathrm{supp}(p_Y)|}{k} + 1 \\
\label{eq:supportsizeconversion}
        & = 2 - \frac{|\mathrm{supp}(p_Y)|}{k}.
\end{align}
Choose $\alpha$ such that
$$
\alpha \delta \le \log\left(1 + \frac{\delta}{10}\right)
$$
for all $0 < \delta < 1/2$
and let $f(\cdot)$ be an estimator such that
$$
\Pr(|h(p_Y) - f(Y^{\overline{N}})| > \alpha \delta) < \epsilon
$$
where $\bar{N}$ is Poisson with mean $\theta n$ and $Y^{\overline{N}}$ is
i.i.d.\ $P_Y$.
Then we have
$$
\Pr\left(|h(p_Y) - f(Y^N)| > \log(1 + \frac{\delta}{10}\right) < \epsilon
$$
which, since
$$
\log\left(1 + \frac{\delta}{10}\right) = 
\min\left[ 
\log\left(1 + \frac{\delta}{10}\right),
-\log\left(1 - \frac{\delta}{10}\right)\right]
$$
implies that
\begin{equation}
\Pr\left(|1 - e^{h(p_Y) - f(Y^{\overline{N}})}| > \frac{\delta}{10}\right) 
            < \epsilon.
\end{equation}
Since $h(\cdot) \le 2$, we may assume that $f(\cdot) \le 2$, in which
case the previous inequality implies
\begin{equation}
\Pr(|1 - e^{h(p_Y) - f(Y^N)}| > \delta e^{-f(Y^N)}) < \epsilon,
\end{equation}
which, by defining
$$
\tilde{f}(Y^{\overline{N}}) = \left(2 - e^{f(Y^N)}\right) k,
$$
substituting (\ref{eq:supportsizeconversion}), 
and rearranging, gives
\begin{equation}
\Pr\left(\left| |\mathrm{supp}(p_Y)| - \tilde{f}(Y^{\overline{N}}) \right| > \delta k\right) < \epsilon.
\end{equation}
Thus $\tilde{f}(\cdot)$ estimates the support of $p_Y$ with accuracy
$\delta k$ with probability at least $\epsilon$ for all
$p_Y$ satisfying~(\ref{eq:minprobconstraint}). It then follows from, e.g., 
Wu and Yang~\cite[Theorem~2]{WuYang} (where the role of $\epsilon$ and $\delta$
are reversed) that there exists a constant $c$ such that for all
$k$ and all $\delta$ such that 
$\frac{1}{k} < \delta < c_0$
$$
\theta n \ge c \frac{k}{\log k} \log^2 \frac{1}{\epsilon}.
$$

\section{Guessing Framework to Interpret Leakage Metrics} \label{sec:commonframe}

Finally, we use the guessing framework to provide operational definitions for commonly used information leakage metrics. The new operational definitions clarify in which cases each metric should be used.

In particular, we show that Shannon capacity is suitable for \emph{covert} channel analysis rather than side-channel analysis. Local differential privacy emerges when the system designer is extremely risk averse (i.e., the probability that $U$ is revealed should be small for every possible realization $y$, regardless of the probability of the latter event). The analysis will naturally lead us to define an information metric that is intermediate between maximal leakage and local differential privacy, which we call maximal realizable leakage (cf.~Section~\ref{secmaxrealizable}).

On the other hand, maximal correlation captures the multiplicative decrease, upon observing $Y$, of the variance of functions of $X$. Hence it is more suitable for \emph{estimation} problems, rather than guessing problems. This also naturally leads to a \emph{cost}-based notion of leakage, which considers reductions in costs and is investigated in Section~\ref{seccostleakage}.

\subsection{Shannon Capacity} \label{subsec:compareCapacity}

Shannon justifies the choice of mutual information by arguing
that ``From the point of view of the cryptanalyst [i.e., the adversary],
a secrecy system is almost identical with a noisy communication system''~\cite{ShannonSecrecy}.
This argument is not persuasive, however, because a noisy communication system
(the rate of which is governed by mutual information) relies on coding, of which
there is generally none in the side-channel setting. One could argue that Shannon is simply taking
a ``pessimistic'' view by upper-bounding leakage by assuming
that the transmitter is a cooperative participant and thus willing to code.
This reasoning is erroneous, however; Shannon capacity is generally
lower than maximal leakage.
The reason is that Shanon capacity is concerned  with (the size of) message sets that can be \emph{reliably} reconstructed at the receiver, whereas leakage does not impose any reliability constraint. This inspires the following definition.
\vspace{1mm}
\begin{Definition}[Recoverable Leakage] \label{defcapacitymaxleak}
Given $\epsilon >0$ and a joint distribution $P_{XY}$ on finite alphabets $\mathcal{X}$ and $\mathcal{Y}$, the \emph{recoverable leakage} from $X$ to $Y$ is defined as 
\begin{align}
\label{eqMLcap}
\mathcal{L}_\epsilon^C (X \!\! \to \!\! Y) = \sup_{\substack{(U,X): U-X-Y \\ \Pr(U = \hat{U}(Y)) \geq 1-\epsilon}} \log \frac{\Pr(U = \hat{U}(Y))}{\max_u P_U(u)}, 
\end{align}
where the support of $U$ is finite but of arbitrary size, and $\hat{U}(Y)$ is the MAP estimator.
\end{Definition}
\vspace{1mm}
\begin{Remark} $\mathcal{L}_\epsilon^C (X \!\! \to \!\! Y)$ depends on $P_{XY}$ only through $P_{Y|X}$. \end{Remark}
\vspace{1mm}
\begin{Theorem} \label{thmcapacitymaxleak}
For any joint distribution $P_{XY}$ on finite alphabets $\mathcal{X}$ and $\mathcal{Y}$,
\begin{align}
\label{eqMLcapthm}
\lim_{\epsilon \rightarrow 0} \lim_{n \rightarrow \infty} \frac{1}{n} \mathcal{L}_\epsilon^C (X^n \!\! \to \!\! Y^n) = C(P_{Y|X}),
\end{align}
where $(X^n,Y^n)$ is distributed i.i.d according to $P_{XY}$, and $C(P_{Y|X})$ is the capacity of the channel $P_{Y|X}$.
\end{Theorem}
\vspace{1mm}
Shannon capacity is a suitable metric for covert channel analysis, in which there are two adversaries attempting to (covertly) communicate through $P_{Y|X}$. That is, they are indeed concerned with sending and reconstructing messages reliably. 
To compare with maximal leakage, suppose $X$ has full support. Then,
\begin{align*}
\ml{X}{Y}&  \stackrel{\text{(a)}} = \lim_{\epsilon \rightarrow 0} \lim_{n \rightarrow \infty} \frac{1}{n} \ml{X^n}{Y^n} \\
& \stackrel{\text{(b)}} = \lim_{\epsilon \rightarrow 0} \lim_{n \rightarrow \infty}  \frac{1}{n} \sup_{\substack{(U,X^n): \\ U-X^n-Y^n  }} \log \frac{\Pr(U = \hat{U}(Y^n))}{\max_u P_U(u)} \\
& \geq \lim_{\epsilon \rightarrow 0} \lim_{n \rightarrow \infty} \frac{1}{n} \sup_{\substack{(U,X^n): \\ U-X^n-Y^n \\ \Pr( U =\hat{U}(Y^n)) \geq 1-\epsilon }}\log \frac{\Pr(U = \hat{U}(Y^n))}{\max_u P_U(u)} \\
& = C(P_{Y|X}),
\end{align*}
where (a) follows from the additivity of maximal leakage, and (b) follows from the 
fact that $\ml{X}{Y}$ depends on $P_X$ only through its support. One can readily verify that the inequality can be strict.
\vspace{1mm}
\begin{Example} \label{example} 
Consider $X \sim \mathrm{Ber}(p)$, $p \in (0,1/2)$. If $Y$ is the output of a $\mathrm{BEC}(\epsilon)$ ($\epsilon \in (0,1)$) with input $X$, then $\ml{X}{Y}  =  \log(2-\epsilon) > (1-\epsilon) \log 2 = C(P_{Y|X})$.
\end{Example} 
\vspace{1mm}
\begin{proof}
To show that the left-hand side upper-bounds the right-hand side, consider
\begin{align}
\mathcal{L}_\epsilon^C (X^n \!\! \to \!\! Y^n) & = \!\! \sup_{\substack{(U,X^n): U-X^n-Y^n \\ \Pr(U = \hat{U}(Y^n)) \geq 1-\epsilon}} \!\! \log \frac{\Pr(U = \hat{U}(Y))}{\max_u P_U(u)} \notag \\
& \geq \!\! \sup_{\substack{(U,X^n): U-X^n-Y^n \\ \Pr(U = \hat{U}(Y^n)) \geq 1-\epsilon \\ U \sim \text{uniform} }}  \!\! \log |\mathcal{U}| \! + \! \log (1-\epsilon).
\end{align}
Note that the right-hand side of the above equation is exactly the channel coding setup: $U$ is the uniform message, $P_{X^n|U}$ is the (stochastic) encoding map, $P_{Y|X}$ is the memoryless channel,  and $\epsilon$ is the allowed average probability of decoding error. Therefore, for any $\delta > 0$, any $U$ with $|\mathcal{U}| < 2^{n(C-\delta)}$ is feasible for large enough $n$, yielding the lower bound. For the reverse direction, consider the following.
\begin{align*}
\mathcal{L}_\epsilon^C (X^n \!\! \to \!\! Y^n) & = \!\!  \sup_{\substack{(U,X^n): U-X^n-Y^n \\ \Pr(U = \hat{U}(Y^n)) \geq 1-\epsilon}}  \!\! \log \frac{\Pr(U = \hat{U}(Y))}{\max_u P_U(u)} \\
& \leq  \!\! \sup_{\substack{(U,X^n): U-X^n-Y^n \\ \Pr(U = \hat{U}(Y^n)) \geq 1-\epsilon}}  \!\!  \log \frac{1}{2^{-H_\infty(P_U)}} \\
& = \sup_{\substack{(U,X^n): U-X^n-Y^n \\ \Pr(U = \hat{U}(Y^n)) \geq 1-\epsilon}} H_\infty(P_U) \\
& \stackrel{\text{(a)}} \leq \sup_{\substack{(U,X^n): U-X^n-Y^n \\ \Pr(U = \hat{U}(Y^n)) \geq 1-\epsilon}} \frac{I(U;\hat{U})+1}{\Pr(U=\hat{U})} \\
& \leq \sup_{\substack{(U,X^n): U-X^n-Y^n \\ \Pr(U = \hat{U}(Y^n)) \geq 1-\epsilon}}  \frac{nC(P_{Y|X})+1}{1-\epsilon},
\end{align*}
where (a) follows from~\cite[Theorem 5]{VerduFanoGeneralized}. Taking the limit as $n \rightarrow \infty$ and $\epsilon \rightarrow 0$ yields the upper bound. 
\end{proof}

\subsection{Maximal Realizable Leakage} \label{secmaxrealizable}

We now consider a variation of the definition of maximal leakage, which captures a different scenario of interest. It will be also useful for interpreting local differential privacy in the guessing framework (cf.~Section~\ref{subsec:LDP}).
In particular, maximal leakage considers the \emph{average} guessing performance of the adversary, $\Pr( U = \hat{U}(Y))$, for each $U$ satisfying $U-X-Y$ since our threat model ``tolerates''
realizations $y$ of $Y$ that lead to a high probability of correct guessing if the corresponding probabilities $P_Y(y)$'s are very small. For scenarios in which such small probability events are still unacceptable, we need to consider the \emph{maximum} instead of the average performance. This is the case, for example, when $U$ represents an individual's medical data or when we do not expect leakage to be concentrated around $Y$ (e.g., $Y$ is a public database as opposed to a running stochastic process). 
This leads to the following definition.
\vspace{1mm}
\begin{Definition}[Maximal Realizable Leakage] 
\label{defmaxleakagerealizable}
Given a joint distribution $P_{XY}$ on finite alphabets $\mathcal{X}$ and $\mathcal{Y}$, the \emph{maximal realizable leakage} from $X$ to $Y$ is defined as
\begin{equation} \label{eqdefmaxleakagerealizable}
\mlr{X}{Y}  =  \sup_{U:U - X - Y} \log \frac{ \max_{ \substack{y \in \mathrm{supp}(Y)}} \max_{u \in \mathcal{U}} P_{U|Y}(u|y) }{\max_{u \in \mathcal{U}} P_U(u)},
\end{equation}
where the support $U$ is finite but of arbitrary size.
\end{Definition}
\vspace{1mm}
\begin{Theorem} \label{thmmaxleakrealizable}
For any  joint distribution $P_{XY}$ on finite alphabets $\mathcal{X}$ and $\mathcal{Y}$, the maximal realizable leakage from $X$ to $Y$ is given by the Renyi Divergence of order infinity, $D_{\infty} (P_{XY} || P_X \times P_Y)$. That is,
\begin{align}
\label{eqthmrealizable}
\mlr{X}{Y}  =   \max_{ \substack{ (x,y) \in \mathcal{X} \times \mathcal{Y} \\ P_{XY}(x,y) > 0}} \log \frac{P_{Y|X}(y|x)}{P_Y(y)} = D_{\infty} (P_{XY} || P_X \times P_Y).
\end{align}
\end{Theorem}
\vspace{1mm}
In contrast, $\ml{X}{Y} = I_{\infty} (X;Y) = \inf_{Q_Y} D_\infty (P_{XY} || P_X \times Q_Y)$. Consequently, $\mlr{X}{Y}$ depends on $P_X$ (as opposed to $\ml{X}{Y}$ which only depends on the support) and is symmetric in $X$ and $Y$.  Note that it is equal to the maximum \emph{information rate}, which is the random variable the expectation of which is mutual information.  It follows straightforwardly from the definitions that $\mlr{X}{Y} \geq \ml{X}{Y}$. Moreover, $\mlr{X}{Y}$ cannot be bounded in terms of $|\mathcal{X}|$ and $|\mathcal{Y}|$: consider the $\mathrm{BEC}$ example (Example~\ref{example}) where $X \sim \mathrm{Ber}(p)$ ($p \in (0,1/2)$) and $Y$ is the output of a $\mathrm{BEC}(\epsilon)$ ($\epsilon \in (0,1)$) with input $X$, then $\mlr{X}{Y} = \log(1/p) \xrightarrow{p \rightarrow 0} \infty$.

Furthermore, $\mlr{X}{Y}$ exhibits desirable properties of a leakage metric: it satisfies the data processing inequality, it is zero if and only $X$ and $Y$ are independent, and it is additive over independent pairs $\{(X_i,Y_i)\}$. These properties are known for Renyi divergence of order $\infty$~\cite{RenyiKLDiv}.
\vspace{1mm}
\begin{Remark}
The fact that using the max in~\eqref{eqdefmaxleakagerealizable} and the average in~\eqref{eq:maxleakagedef} both lead to quantities with desirable properties suggests that we could also consider weighted averages, i.e., replace the numerator by $\left( \sum_y P_Y(y) \max_{u} P_{U|Y}^\alpha (u|y) \right)^{1/\alpha}$, for some $\alpha > 0$. See also~\cite{LiaoSibson}.
\end{Remark}
\vspace{1mm}
\begin{proof}
That $\mlr{X}{Y} \leq D_{\infty} (P_{XY} || P_X \times P_Y)$ follows directly from Proposition~\ref{prop:guessboundY}. 
For the reverse direction, we again consider the shattering $P_{U|X}$ (cf.~equation~\eqref{eqshattering}).
It is a simple exercise to check that this choice yields the desired lower bound. 
\end{proof}

\subsection{Local Differential Privacy} \label{subsec:LDP}

Differential privacy~\cite{DiffPrivacySurvey} is a widely adopted metric in the database security literature. Roughly speaking, it requires that, for any two \emph{neighboring} databases, the probabilities of any given output do not differ significantly.  
Local differential privacy~\cite{LocalDiffPrivacy} adapts that notion to the setting of a given conditional distribution $P_{Y|X}$. It is defined as:
\begin{align} \label{eqdefLDPorig}
L^{dp}(X \!\! \to \!\!  Y) = \max_{ \substack{y \in \mathcal{Y}, \\ x,x' \in \mathcal{X}}} \log \frac{P_{Y|X}(y|x)}{P_{Y|X}(y|x')}.
\end{align}
Local differential privacy is known to be pessimistic~\cite{LocalDiffPrivacy}. It is indeed very strict: for the $\mathrm{BEC}$ example (Example~\ref{example}) where $X \sim \mathrm{Ber}(p)$ ($p \in (0,1/2)$) and $Y$ is the output of a $\mathrm{BEC}(\epsilon)$ ($\epsilon \in (0,1)$) with input $X$, $L^{dp} ( X \!\! \to \!\! Y) = \infty$.  Interestingly, we also noted in the previous section that $\lim_{p \rightarrow 0} \mlr{X}{Y} = \infty$.

So what operational problem is local differential privacy solving? Similarly to maximal realizable leakage, local differential privacy is concerned with worst-case analysis over the \emph{realizations} of $Y$. Moreover, being a function of $P_{Y|X}$, it is robust against the worst-case distribution $P_X$. Hence, differential privacy is suitable for database security problems in which we do not tolerate low risk, and we do not make any assumptions about the distribution generating the data. 
This yields the following definition.
\vspace{1mm}
\begin{Definition}
Given a conditional distribution $P_{Y|X}$ from $\mathcal{X}$ to $\mathcal{Y}$, where  $\mathcal{X}$ and $\mathcal{Y}$ are finite alphabets, let
\begin{align}
\mathcal{L}^{dp} ( X\!\! \to \!\! Y) & =  \sup_{P_X}  \sup_{U:U - X - Y} \!\! \log \frac{ \max_{ \substack{y }} \! \max_{u } \! P_{U|Y}(u|y) }{\max_{u } P_U(u)} \notag \\
& = \sup_{P_X} \mlr{X}{Y} .\label{eqdefLDP}
\end{align}
\end{Definition}
\vspace{1mm}
\begin{Theorem} \label{thmLDP}
For any conditional distribution $P_{Y|X}$ from $\mathcal{X}$ to $\mathcal{Y}$, where  $\mathcal{X}$ and $\mathcal{Y}$ are finite alphabets,
\begin{align}
\label{eqthmLDP}
\mathcal{L}^{dp} ( X\!\! \to \!\! Y) = L^{dp} (X \!\! \rightarrow \!\! Y). 
\end{align}
\end{Theorem}
\vspace{1mm}
Clearly, $\mathcal{L}^{dp} ( X\!\! \to \!\! Y) \geq \mlr{X}{Y} \geq {\ml{X}{Y}}$. Theorems~\ref{thmmaxleakrealizable} and~\ref{thmLDP} imply that $\mathcal{L}^{dp} ( X\!\! \to \!\! Y)= {\mathcal{L}^r ( X\!\! \to \!\! Y)}$ if and only if $X$ and $Y$ are independent. Thus,  $\mathcal{L}^{dp} ( X\!\! \to \!\! Y)= {\ml{X}{Y}}$ if and only if $X$ and $Y$ are independent. Moreover, an interesting implication of~\eqref{eqdefLDP} is that one could incorporate information about the marginal $P_X$ by restricting the optimization set of the $\sup$.
\begin{proof}
By Theorem~\ref{thmmaxleakrealizable}, we can rewrite~\eqref{eqdefLDP} as
\begin{align*}
\mathcal{L}^{dp} ( X\!\! \to \!\! Y) =  \sup_{P_X} \max_{ \substack{ (x,y) \in \mathcal{X} \times \mathcal{Y} \\ P_{XY}(x,y) > 0}} \log \frac{P_{Y|X}(y|x)}{P_Y(y)}.
\end{align*}
The upper bound thus follows from the fact that $P_Y(y) \geq \min_{x} P_{Y|X}(y|x)$. For the lower bound, consider the following. Let $y^\star$ be an element achieving the max in~\eqref{eqdefLDPorig}. Let $x_0 \in \argmin_x P_{Y|X}(y^\star|x)$ and $x_1 \in \argmax_x P_{Y|X}(y^\star|x)$. Finally, for a given $\alpha > 0$, let  $P_X(x_0) = 1- \alpha$ and $P_X(x_1) = \alpha$. Then
\begin{align*}
 \mathcal{L}^{dp} ( X\!\! \to \!\! Y) 
& \geq \log \frac{\max_x P_{Y|X}(y^\star|x) }{ P_Y(y^\star) } \\
& = \log \frac{P_{Y|X}(y^\star|x_1)}{(1-\alpha) P_{Y|X}(y^\star|x_0) + \alpha P_{Y|X}(y^\star|x_1) } \\
& \xrightarrow{\alpha \rightarrow 0} \log \frac{P_{Y|X}(y^\star|x_1)}{P_{Y|X}(y^\star|x_0)} = L^{dp} (X \!\! \rightarrow \!\! Y). 
\end{align*} 
\end{proof}
It is worth noting that Dwork \emph{et al.}~\cite{DworkCalibrating} provide an operational definition closely related to the above definition. In particular, a simple modification of their result yields that
\[ L^{dp} ( X\!\! \to \!\! Y)= \sup_{P_X} \sup_{ \substack{f: \mathcal{X} \rightarrow \{0,1\} \\ y \in \mathcal{Y} }} \left| \log \left( \frac{\Pr( f(X) = 1 | Y=y)}{\Pr(f(X)=1)}  \right) \right|. \]
Alternatively, Kairouz \emph{et al.}~\cite{KairouzComposition} give an operational definition of $(\epsilon,\delta)$-differential privacy in the framework of hypothesis testing. They show that it determines the trade-off between the probabilities of false alarm and missed detection.

\subsection{Maximal Correlation}

Given a joint distribution $P_{XY}$, the Hirschfeld-Gebelein-R{\'e}nyi maximal correlation~\cite{HMaxCorr,GMaxCorr,RenyiMaxCorr}  $\rho_m (X;Y)$ is defined as
\begin{align}
\rho_m (X;Y) = \sup_{\substack{f,g: \\ \E[f]=\E[g]=0 \\ \E[f^2]=\E[g^2]=1}} \E[f(X)g(Y)].
\end{align}
Calmon \emph{et al.} showed that, when the alphabet $\mathcal{X}$ is finite, 
\[ \sup_{N ,f: \mathcal{X} \rightarrow [N]} \left( \sup_{\hat{f}(\cdot)} \Pr( f(X) = \hat{f}(Y)) - \max_{k \in [N]} P_f(k) \right) \leq \rho_m(X;Y), \]
where the supremum is over deterministic functions~\cite[Theorem 9]{CalmonPIC}. However, as we saw in the introduction, the left-hand side can be zero even when $X$ and $Y$ are dependent.
We posit that maximal correlation is more precisely capturing the change in variance. That is, we define \emph{variance leakage} as follows.
\vspace{1mm}
\begin{Definition}[Variance Leakage] \label{defmaxleakagevariance}
Given a joint distribution $P_{XY}$ on alphabets $\mathcal{X}$ and $\mathcal{Y}$, the \emph{variance leakage} from $X$ to $Y$ is defined as
\begin{equation} \label{eqmaxleakagevariance}
\mathcal{L}^v(X \!\! \to \!\! Y) = \sup_{\substack{U: U-X-Y \\ \var(U)>0}} \log \frac{\var(U)}{\E[ (U-\E[U|Y])^2 ] }.
\end{equation}
\end{Definition}
\vspace{1mm}
\begin{Lemma} \label{thmvariance}
For any joint distribution $P_{XY}$ on  alphabets $\mathcal{X}$ and $\mathcal{Y}$, the variance leakage from $X$ to $Y$ is given by
\begin{equation}
\mathcal{L}^v(X \!\! \to \!\! Y) = - \log  (1- \rho_m^2 (X;Y)),
\end{equation}
where $\rho_m (X;Y)$ is the maximal correlation.
\end{Lemma}
\vspace{1mm}
As such, maximal correlation is capturing the multiplicative decrease in variance. If the $U$ of interest is discrete, which is often the case in practice (e.g., $U$ is a password, a social security number, etc.), the probability of correct guessing is arguably the more relevant quantity.
This holds true even in the case of location privacy, in which (as we saw earlier) typical functions of interest, such as work/home addresses or political affiliations, are discrete.
\begin{proof} The proof is a simple rewriting of R{\'e}nyi's equivalent characterization, taking into account randomized functions of $X$. We include it here for completeness.
Without loss of generality, we can restrict the optimization in~\eqref{eqmaxleakagevariance} to $U$'s that satisfy $\E[U]=0$, and $\E[U^2]=1$. So, we rewrite
\begin{align} 
\mathcal{L}^v(X \!\! \to \!\! Y) & =   \!\! \!\!  \sup_{\substack{U: U-X-Y \\ \E[U]=0,~\E[U^2]=1}} \!\!  \!\! \log \frac{1}{\E[U^2] -\E[\E[U|Y]^2] } \notag \\
& = \!\! \!\!  \sup_{\substack{U: U-X-Y \\ \E[U]=0,~\E[U^2]=1}}  \!\! \!\! - \log  \left(1 -\E \left[\E[U|Y]^2\right] \right)  \! . \label{eqmaxleakvarequivalent}
\end{align}
Also, we can rewrite maximal correlation using Renyi's equivalent characterization~\cite{RenyiMaxCorr}:
\begin{align} \label{eqmaxcorrequiv}
\rho_m (X;Y) =  \!\! \sup_{\substack{f: \E[f(X)]  =0  \\ ~\E[f^2(X)]  =1 }}  \!\! \sqrt{\E \left[\E[f(X)|Y]^2\right]}.
\end{align}
Now note that
\begin{align}
 \rho^2_{m}(X;Y) & = \sup_{\substack{f: ~\E[f]=0,~ \E[f^2]=1}} \E \left[ \E[f(X)|Y]^2 \right]  \notag \\
& \stackrel{\text{(a)}} \leq \sup_{\substack{U: U-X-Y \\ \E[U]=0,~\E[U^2]=1}} \E \left[ \E[U|Y]^2 \right] \notag \\
& \leq \sup_{\substack{U: U-X-Y \\ \E[U]=0,~\E[U^2]=1}} \sup_{\substack{h: ~\E[h(U)]=0, \\ \E[h^2(U)]=1}} \E \left[ \E[h(U)|Y]^2 \right] \notag \\
&  = \sup_{\substack{U: U-X-Y \\ \E[U]=0,~\E[U^2]=1}} \rho^2_{m} (U;Y) \notag \\
& \stackrel{\text{(b)}} \leq \rho^2_{m}(X;Y) \label{eqMaxCorrDPI},
\end{align}
where (b) follows from the fact that maximal correlation obeys the data processing inequality, which can be shown using standard properties of conditional
expectation.
Therefore (a) is in fact an equality. Plugging it in~\eqref{eqmaxleakvarequivalent} yields our desired result. 
\end{proof}

Definition~\ref{defmaxleakagevariance}, with the restriction that
$U = X$, has also been recently investigated by Asoodeh \emph{et al.}~\cite{EstimationEfficiency}. Note that it can be rewritten as
\begin{align} \label{eq:varleakcost}
\mathcal{L}^v(X \!\! \to \!\! Y) = \sup_{\substack{U: U-X-Y \\ \var(U)>0}} \log 
\frac{\inf_{u} \E[(U-u)^2]}{\inf_{u(\cdot)} \E[(U-u(Y))^2]}. 
\end{align}
Hence, $\mathcal{L}^v(X \!\! \to \!\! Y)$ measures the reduction in cost incurred by the adversary, where cost is measured by the mean squared error. In the next section, we consider a natural extension in which we do not assume the cost function is known a priori.

\subsection{Maximal Cost Leakage} \label{seccostleakage}
In this section, we introduce a leakage metric that is dual to maximal leakage. Whereas maximal leakage  considers the maximum gain that the adversary achieves, we could alternatively consider the maximum reduction in \emph{cost} they incur. 
\begin{Definition}[Maximal Cost Leakage]\label{defcostleakage}
Given a joint distribution $P_{XY}$ on alphabets $\mathcal{X}$ and $\mathcal{Y}$, the \emph{maximal cost leakage} from $X$ to $Y$ is defined as
\begin{equation} 
\mcl{X}{Y}= \sup_{\substack{U: U-X-Y \\ \hat{\mathcal{U}}, ~ d: \hat{\mathcal{U}} \times \mathcal{U} \rightarrow \mathbb{R}_+}} \log \frac{\inf_{\hat{u} \in \hat{\mathcal{U}}} \E[d(U,\hat{u})]  }{  \inf_{\hat{u}(\cdot) } \E[d(U,\hat{u}(Y))] },
\end{equation}
where $U$ takes value in a finite (but arbitrary) alphabet, and $\frac{0}{0}=1$ by convention.
\end{Definition}
\vspace{1mm}
It is important to note that the gain-based approach is more operationally meaningful (for side-channel analysis) than the cost-based approach. To illustrate this, suppose $d$ is the Hamming distortion and consider 
\begin{align} \label{eq:error}
\sup_{ \substack{U:U-X-Y}} \log \frac{1-\max_{u \in \mathcal{U} }P_U(u)}{1-\sup_{\hat{u}(\cdot)} \Pr(U = \hat{u}(Y))}.
\end{align}
Recall that, in the definition of maximal leakage, we considered the ratio of the guessing probabilities (as opposed to the difference) because we are typically interested in functions that are hard to guess (e.g., passwords). 
However, the quantity in~\eqref{eq:error} is much more sensitive to changes for functions that are easy to guess: suppose for some $U$, $\max_{u} P_U(u) = 1-10^{-9}$ and $\sup_{\hat{u}(\cdot)} \Pr(U = \hat{u}(Y))=1$, then the ratio in~\eqref{eq:error} is $\infty$. On the other hand, if $ \max_{u} P_U(u) = 10^{-9}$ and $\sup_{\hat{u}(\cdot)} \Pr(U = \hat{u}(Y))=10^{-3}$, the ratio is only $\approx 1.001$ despite the significant change. More generally, it is more intuitive to associate a gain to the adversary if they compromise the system, rather than a cost if they fail to do so\footnote{The cost and gain approaches may be equivalent if we are interested in the difference between the incurred cost (or achieved gain) when $Y$ is observed versus when no observations are made. This is not the case, however, if we are considering the ratio instead of the difference.} (an adversary does not ``lose'' if the system is not compromised). Even in the rate-distortion-based approach to information leakage, the more robust metric is the probability that the adversary incurs a small distortion~\cite{MySecrecySystem} rather than (say) the expected value of the distortion. That is, the probability-metric falls under the gain approach (similar to maximal locational leakage (cf.~Definition~\ref{defleakcont}) or maximal gain leakage (cf.~Definition~\ref{defmaxgains})), albeit the gain is defined indirectly through a distortion function.

Nevertheless, maximal cost leakage admits a simple form for discrete $X$ and $Y$, given in the following theorem.
\begin{Theorem} \label{thmcostleakage}
For any joint distribution $P_{XY}$ on finite alphabets $\mathcal{X}$ and $\mathcal{Y}$, the maximal cost leakage from $X$ to $Y$ is given by
\begin{equation}
\mcl{X}{Y} = - \log   \sum_{y \in \mathcal{Y}} \min_{ \substack{x \in \mathcal{X}: \\ P_X(x) > 0}} P_{Y|X}(y|x)  .
\end{equation}
\end{Theorem}
\vspace{1mm}
It is worth noting that $\mcl{X}{Y}$, similarly to maximal leakage, depends on $P_{XY}$ only through $P_{Y|X}$ and the support of $P_X$. Moreover, a relation analogous to~\eqref{eq:SibsonMI} holds for $\mcl{X}{Y}$:
\begin{align} \label{eq:MCLequiv}
\mcl{X}{Y} = \inf_{Q_Y} D_{\infty} (P_X \times Q_Y || P_{XY}).
\end{align}
The proofs for Theorem~\ref{thmcostleakage} and the above relation are given in Appendices~\ref{app:costleakthm} and~\ref{app:Cost2}, respectively. 
The following corollary, the proof of which is given in Appendix~\ref{app:costleakcorr}, summarizes useful properties of $\mcl{X}{Y}$.
\vspace{1mm}
\begin{Corollary} \label{corrcostleakage}
For any joint distribution $P_{XY}$ on finite alphabets $\mathcal{X}$ and $\mathcal{Y}$,
\begin{enumerate}
\item (Data Processing Inequality) If the Markov chain $X-Y-Z$ holds for a discrete random variable $Z$, then $\mcl{X}{Z} \leq \min\{ \mcl{X}{Y}, \mcl{Y}{Z} \}.$
\item $\mcl{X}{Y} = 0$ iff $X$ and $Y$ are independent.
\item (\emph{Additivity}) 
If $\{(X_i,Y_i)\}_{i=1}^{\ell}$ are mutually independent, then
\begin{equation*}
\mcl{X_1^\ell}{Y_1^\ell}= \sum_{i=1}^\ell \mcl{X_i}{Y_i}.
\end{equation*}
\item For any non-trivial deterministic law $P_{Y|X}$ (i.e., $ | \{y: P_Y(y) >0 \}| > 1$),  $\mathcal{L}^c(X \!\! \to \!\! Y)= +\infty $. 
\item $\mcl{X}{Y}$ is not symmetric in $X$ and $Y$.
\item $\mcl{X}{Y} \leq L^{dp} ( X \!\! \to \!\! Y)$.
\item $\mcl{X}{Y}$ is convex in $P_{Y|X}$ for fixed $P_X$.
\end{enumerate}
\end{Corollary}
\vspace{1mm}
Thus maximal cost leakage satisfies axiomatic properties of a leakage measure. However, it cannot be bounded in terms of $|\mathcal{X}|$ and $|\mathcal{Y}|$. Indeed, even if $X$ is a single bit, $X \sim \mathrm{Ber}(p)$ for $p \in (0,1)$, $\mcl{X}{X} = +\infty$. 
We evaluate $\mcl{X}{Y}$ for some other examples. 
\begin{Example}
If $X \sim \text{Ber}(q)$, $0 < q <1$, and $Y$ is the output of a BSC with input $X$ and parameter $p$, $0 \leq p \leq 1/2$, then $\mcl{X}{Y} = - \log (2p) $. 
\end{Example}
\vspace{1mm}
\begin{Example} \label{exampleasy2}
If $X \sim \text{Ber}(q)$, $0 < q <1$, and $Y$ is the output of a BEC with input $X$ and parameter $\epsilon$, $0 \leq \epsilon < 1$, then $\mcl{X}{Y} = -\log(\epsilon)$, and $\mathcal{L}^c(Y \!\! \to \!\! X)= +\infty $. 
\end{Example}
\vspace{1mm}
\begin{Remark}
One can note that in each of the examples above, $\mcl{X}{Y} \geq \ml{X}{Y}$. This is always true when $|\mathcal{X}|=|\mathcal{Y}|=2$, but it is not necessarily true in general. As a counter example, say $X$ has full support and $P_{Y|X} = \begin{bmatrix}
0.2 & 0.5 & 0.3 \\ 0.3 & 0.4 & 0.3 \\ 0.2 & 0.4 & 0.4
\end{bmatrix}$. Then $\exp\{  \ml{X}{Y} \} = 1.2$ and $\exp\{  \mcl{X}{Y} \} = 1/0.9=1.\bar{1}$. 
\end{Remark}
\vspace{1mm}

\subsubsection{Comparison with Maximal Correlation} Definition~\ref{defcostleakage} restricted $U$ to be discrete, but the proof of the upper bound in Theorem~\ref{thmcostleakage} does not need this assumption. That is, if we take the supremum over all real-valued $U$'s, the theorem still holds. Comparing with~\eqref{eq:varleakcost}, we get $\mcl{X}{Y} \geq \mathcal{L}^{v} (X \!\! \to \!\! Y)$. We can rewrite this inequality as follows.
\vspace{1mm}
\begin{Corollary} \label{corr:costleakmaxcorr}
For any joint distribution $P_{XY}$ on finite alphabets $\mathcal{X}$ and $\mathcal{Y}$,
\begin{equation}
\label{eq:maxcorrmaxleakge}
\rho_{m} (X;Y) \leq \sqrt{ 1- e^{-\mathcal{L}^c( X \to Y)} }. 
\end{equation}
Consequently, for a fixed conditional distribution $P_{Y|X}$,
\[ \sup_{P_X} s^{\star} (X;Y) \leq  1- \sum_{y \in \mathcal{Y}} \min_{ \substack{x \in \mathcal{X}}} P_{Y|X}(y|x) , \]
where $s^{\star} (X;Y) := \sup_{U: U-X-Y} \frac{I(U;Y)}{I(U;X)}$ is the strong data processing coefficient.
\end{Corollary}
\vspace{1mm}
Note that inequality (\ref{eq:maxcorrmaxleakge}) is tight in the extremal cases, i.e., if $X$ and $Y$ are independent, if $Y$ is a deterministic function of $X$, or if $X$ is a deterministic function of $Y$ (it can be readily verified in this case that $\sum_y \min_{x} P_{Y|X}(y|x)=0$, unless $X$ or $Y$ is determinstic). The second inequality follows from the fact that $ \sup_{P_X} s^{\star} (X;Y) =  \sup_{P_X} \rho_m^2 (X;Y)$~\cite[Theorem 8]{Ahlswede:Hypercontraction}.

\subsubsection{Maximal Realizable Cost}
Similarly to the modification of maximal leakage to maximal realizable leakage, we could consider the minimum cost incurred at the adversary, instead of the average cost. We show next that this yields the maximum of the \emph{negative} of the information rate. Maximizing it over the input distribution also yields local differential privacy.
\vspace{1mm}
\begin{Definition}[Maximal Realizable Cost] \label{defcostrealizable}
Given a joint distribution $P_{XY}$ on alphabets $\mathcal{X}$ and $\mathcal{Y}$, the \emph{maximal realizable cost} from $X$ to $Y$ is defined as
\begin{equation} 
\mathcal{L}^{rc} (X \!\! \to \!\! Y) = \sup_{\substack{U: U-X-Y \\ \hat{\mathcal{U}}, ~ d: \hat{\mathcal{U}} \times \mathcal{U} \rightarrow \mathbb{R}_+}} \log \frac{\inf_{\hat{u} \in \hat{\mathcal{U}}} \E[d(U,\hat{u})]  }{ \min_{y \in \mathrm{supp}(Y)}  \inf_{ \hat{u} \in \hat{\mathcal{U}}}  \E[d(U,\hat{u}) | Y=y] },
\end{equation}
where $\hat{\mathcal{U}}$ is a finite alphabet, and $\frac{0}{0}=1$ by convention.
\end{Definition}
\vspace{1mm}
\begin{Theorem} \label{thmcostrealizable}
For any joint distribution $P_{XY}$ on finite alphabets $\mathcal{X}$ and $\mathcal{Y}$, the maximal realizable cost from $X$ to $Y$ is given by the R{\'e}nyi divergence of order infinity, $D_\infty (P_{X} \times P_Y || P_{XY})$.  That is,
\begin{equation}
\mathcal{L}^{rc}(X \!\! \to \!\! Y) =  \log\max_{\substack{ x,y: \\ P_{X}(x)P_Y(y) >0}  } \frac{P_{Y}(y)}{P_{Y|X}(y|x)} = D_\infty (P_{X} \times P_Y || P_{XY}) .
\end{equation}
\end{Theorem}
Similarly to maximal realizable leakage, $\mathcal{L}^{rc}(X \!\! \to \!\! Y) $ depends on $P_X$ not only through its support. They are also analogous in that the former is equal to $D_\infty (P_{XY} || P_X \times P_Y)$ and the latter is equal to $D_\infty (P_{X} \times P_Y || P_{XY})$. 
\vspace{1mm}
\begin{Corollary} \label{corrLDPmaxcostrealized}
For any conditional distribution $P_{Y|X}$ from $\mathcal{X}$ to $\mathcal{Y}$, where  $\mathcal{X}$ and $\mathcal{Y}$ are finite alphabets,
\begin{align}
\label{eqLDPmaxcostrealized}
\max_{P_X} \mathcal{L}^{rc}(X \!\! \to \!\! Y) =L^{dp} ( X \!\! \to \!\! Y).
\end{align}
\end{Corollary}
A consequence of Theorem~\ref{thmLDP} and Corollary~\ref{corrLDPmaxcostrealized} is that local differential privacy is concerned with both worst-case reductions in costs incurred and worst-case increases in gains achieved at the adversary.
The proofs of Theorem~\ref{thmcostrealizable} and Corollary~\ref{corrLDPmaxcostrealized} are given in Appendices~\ref{app:costrealizable} and~\ref{app:corrcostrealized}, respectively.

\section{Discussion} 

It is worth noting that Sibson's mutual information of infinite
order~(\ref{eq:SibsonMI}) has appeared  
in the data compression literature as the Shtarkov sum~\cite{Shtarkov}, which evaluates the worst-case regret.
More recently, it has also been used as a complexity measure in the study of communication complexity~\cite{SibsonInfComp}.

If $X$ is binary and not deterministic and $Y = X$, then the maximal
leakage from $X$ to $Y$ is one bit. Thus if $X$ represents, say, whether
Alice has a stigmatized disease, and Alice
reveals this information to Bob, maximal leakage
would declare that only one bit has been leaked to him. Maximal leakage 
would likewise declare that one bit has been leaked if Alice
revealed the first bit of her phone number or whether she was born
on an even- or odd-numbered day. Thus maximal leakage fails to capture
the gravity of revealing highly-confidential quantities if those quantities
can only take a few possible values. The reason is simply that maximal
leakage measures the extent to which randomized functions of $X$ that
are difficult to guess \emph{a priori} become easy to guess after
observing $Y$. Any binary-valued function can be guessed 
\emph{a priori} with probability at least $1/2$. Therefore the
increase in the guessing probability upon observing $Y$ cannot
be large. According to maximal leakage, revealing whether Alice
has a particular disease is not a concern because
Bob already has a reasonably high probability of guessing correctly even 
without any information from Alice. Thus maximal leakage is an
appropriate metric when the goal is to prevent Bob from guessing
quantities, such as passwords or keys, that are \emph{a priori}
hard to guess. Other metrics, such as differential 
privacy~(\ref{eqdefLDPorig})
are more appropriate in the above scenario in which revealing a single 
bit represents a significant breach.

Following the publication of an early version of this work, maximal leakage was used as a privacy metric in the context of hypothesis testing~\cite{HypTestML}, and in a more general setup of privacy-utility trade-offs~\cite{Wang2017Privacy}.  Variations on the definition of maximal leakage that yield Sibson mutual information of finite orders have also been considered~\cite{LiaoSibson}.

\label{sec:conc}

\section*{Acknowledgment}

The authors thank Emre Telatar for pointing out the connection between
the Sibson mutual information of infinite order and the Shtarkov sum
and for suggesting that no scalar multiple of $I(X;Y)$ can upper-bound
$\ml{X}{Y}$ in Lemma~\ref{lemmacompareMI}. This research was supported
by the National Science Foundation under grant CCF-1704443.

\begin{appendices}

\section{Proofs for Section~\ref{sec:main}}

\subsection{Proof of Lemma~\ref{lemmadefprop}} \label{app:lemmadefproof}
\begin{enumerate}
\item Consider any discrete $U$ satisfying $U-X-Y$ and define $G(U;Y) = \sup_{\hat{U}:U-Y-\hat{U}} \log \frac{\Pr{\left(U=\hat{U}\right)}}{\max_{u \in \mathcal{U}} P_U(u)}$. Clearly if $U-X-Y-Z$ holds, then $G(U;Z) \leq G(U;Y)$. 
So if $X-Y-Z$ holds,
\begin{align*}
\ml{X}{Z} = \sup_{U: U-X-Z} G(U;Z) = \sup_{U: U-X-Y-Z} G(U;Z) \leq \sup_{U: U-X-Y-Z} G(U;Y) = \ml{X}{Y},
\end{align*}
Similarly,
\begin{align*}
\ml{X}{Z} = \sup_{U: U-X-Z} G(U;Z) = \sup_{U: U-X-Y-Z} G(U;Z) \leq \sup_{U: U-Y-Z} G(U;Z) = \ml{Y}{Z}.
\end{align*}
\item If $Y$ is discrete, then for any discrete $U$
\begin{align*}
\sup_{\hat{U}:U-X-Y-\hat{U}} \Pr{\left(U=\hat{U}\right)}  = \sum_{y \in \mathrm{supp}(Y)} \max_{u \in \mathcal{U}} P_{UY}(u,y) \leq   \sum_{y \in \mathrm{supp}(Y)} \max_{u \in \mathcal{U}} P_{U}(u) = |\mathrm{supp}(Y)|  \max_{u \in \mathcal{U}} P_{U}(u).
\end{align*} Hence for any $U$ satisfying $U-X-Y$, $G(U;Y) \leq \log |\mathrm{supp}(Y)|$ and subsequently $\ml{X}{Y}  \leq \log |\mathrm{supp}(Y)|$.
\item If $X$ is discrete, then $\ml{X}{Y} \leq \ml{X}{X} \leq \log |\mathrm{supp}(X)|$, where the first inequality follows from 1) and the second from 2). 
\item If $X$ and $Y$ are independent, then any $U$ satisfying $U-X-Y$ is independent from $Y$. Hence $G(U;Y)=0$ for all $U$, which implies $\ml{X}{Y} = 0$. The non-negativity is obvious.
\end{enumerate}

\subsection{Proof of Corollary~\ref{corrprop}} \label{app:corrprop}

\begin{enumerate}
\item If $\ml{X}{Y}=0$, then $\sum_{y \in \mathcal{Y}} \max_{x \in \mathrm{supp}(X)} P_{Y|X}(y|x) = 1$. Hence, $\sum_{y \in \mathcal{Y}} \max_{x \in \mathrm{supp}(X)} P_{Y|X}(y|x) = \sum_{y \in \mathcal{Y}} P_Y(y)$. Since $ \max_{x \in \mathrm{supp}(X)} P_{Y|X} (y|x) \geq P_Y(y)$ for every $y \in \mathcal{Y}$, it follows that  $\max_{x \in \mathrm{supp}(X)} P_{Y|X} (y|x)= P_Y(y)$ for all $y \in \mathcal{Y}$. Therefore, $X$ and $Y$ are independent. The reverse direction follows from Lemma~\ref{lemmadefprop}.
\item The additivity property is known for $I_{\infty}(X;Y)$~\cite{SibsonInfo,Verdualpha}.
\item Since 
\begin{align*}
\sum_{y \in \mathcal{Y}} \max_{x \in \mathcal{X}} P_{Y|X}(y|x) \leq \sum_{y \in \mathcal{Y}} \sum_{x \in \mathcal{X}} P_{Y|X}(y|x) = |\mathcal{X}|,
\end{align*} 
equality holds if and only if for all $y \in \mathcal{Y}$, $ \max_{x \in \mathcal{X}} P_{Y|X}(y|x) = \sum_{x \in \mathcal{X}} P_{Y|X}(y|x)$.  This condition holds if and only if for all $y \in \mathcal{Y}$, there exists a unique $x_y \text{ such that } P_{Y|X}(y|x_y) > 0$. Finally, the latter condition holds if and only if for all $y \in \mathcal{Y}$, there exists $x_y$ such that $P_{X|Y}(x_y|y) =  1$.
\item The equality is straightforward to verify.
\item The asymmetry is illustrated in Example~\ref{exampleasy1}.
\item Convexity in $P_{Y|X}$ follows from the fact that for each $y \in \mathcal{Y}$, $\max_x P_{Y|X}(y|x)$ is convex in $P_{Y|X}$. 
\item Concavity in $P_X$ follows from the fact that for any $\lambda \in (0,1)$ and any two distributions $P_1$ and $P_2$ on $\mathcal{X}$, ${\mathrm{supp}( \lambda P_1 +(1-\lambda) P_2) }= \mathrm{supp}(P_1) \cup \mathrm{supp}(P_2)$.
\end{enumerate}

\subsection{Proof of Lemma~\ref{lemmacompareMI}}  \label{app:lemmaMIMLproof}
Consider the following chain of inequalities.
\begin{align*}
I(X;Y) & = \sum_{x \in \mathcal{X}, y \in \mathcal{Y}} P_{XY}(x,y) \log \frac{P_{Y|X}(y|x)}{P_Y(y)}  
= \sum_{\substack{x \in \mathcal{X}, y \in \mathcal{Y}: \\ P_{XY}(x,y) >0}} P_{XY}(x,y) \log \frac{P_{Y|X}(y|x)}{P_Y(y)} \\
& \stackrel{\text{(a)}} \leq \log \sum_{\substack{x \in \mathcal{X}, y \in \mathcal{Y}: \\ P_{XY}(x,y) >0}} P_{XY}(x,y)\frac{P_{Y|X}(y|x)}{P_Y(y)} 
 = \log \sum_{\substack{x \in \mathcal{X}, y \in \mathcal{Y}: \\ P_{XY}(x,y) >0}} P_{X|Y}(x|y) P_{Y|X}(y|x) \\
& \stackrel{\text{(b)}} \leq \log \sum_{\substack{x \in \mathcal{X}, y \in \mathcal{Y}: \\ P_{XY}(x,y) >0}} P_{X|Y}(x|y) \max_{x' \in \mathcal{X}: P_X(x') > 0} P_{Y|X}(y|x') 
 = \log \sum_{\substack{ y \in \mathcal{Y}: \\ P_Y(y) > 0}} \max_{x \in \mathcal{X}: P_X(x) > 0} P_{Y|X}(y|x) \\
& \stackrel{\text{(c)}} = \log \sum_{\substack{ y \in \mathcal{Y}}} \max_{x \in \mathcal{X}: P_X(x) > 0} P_{Y|X}(y|x)  = \mathcal{L}(X \to Y),
\end{align*}
where (a) is Jensen's inequality, and (c) follows from the fact that $P_{Y}(y)=0$ implies that $\max_{x \in \mathcal{X}: P_X(x) > 0} P_{Y|X}(y|x) =0.$ Now, note that (b) is an equality if and only if condition 1) holds. Given condition 1), it can be seen that condition 2) is necessary and sufficient for (a) to become equality (by expanding $P_Y(y) = \sum_{x: P_{XY}(x,y) >0} P_X(x) P_{Y|X}(y|x)$).

\section{Proofs for Section~\ref{sec:variations}}

\subsection{Proof of Theorem~\ref{lemmakeq1}} \label{app:kleakage}

To show $\kml{X}{Y} \geq \ml{X}{Y}$, we consider an arbitrary $P_{U|X}$ and construct $P_{V|X}$ such that $\kml{X}{Y} [V] = {\ml{X}{Y}[U]}$. In particular, for a given $P_{U|X}$ and associated alphabet $\mathcal{U}$, let 
\begin{align*}
\mathcal{V} =\bigcup_{u \in \mathcal{U}} \{(u,1),(u,2),\ldots,(u,k)\}, 
\text{ and }
P_{V|X} (v|x)  = P_{V|X}((a_v,b_v)|x) = P_{U|X}(a_v|x)/k.
\end{align*}
Then the probability of correctly guessing $V$ with $k$ guesses after observing $Y$ is:
\begin{align}
 \sup_{X - Y - (\hat{V}_i)_{i=1}^k} \Pr(V = \hat{V}_1 \vee \dots \vee V =\hat{V}_k) 
& = \sum_{y \in \mathcal{Y}} \max_{\substack{ v_1,v_2,\ldots,v_k  \\ v_i \neq v_j, i\neq j}} \sum_{i=1}^k \sum_{x \in \mathcal{X}} P_X(x) P_{V|X}(v_i|x) P_{Y|X}(y|x) \notag \\
& = \sum_{y\in \mathcal{Y}} \sum_{i=1}^k \! \max_{v_i \neq v_1,\ldots,v_{i-1}}  \sum_{x \in \mathcal{X}} \! P_X(x) P_{V|X}(v_i|x) P_{Y|X}(y|x) \notag \\
& \stackrel{(a)} = \sum_{y \in \mathcal{Y}} \max_u \sum_{x \in \mathcal{X}} P_X(x)P_{U|X}(u|x)P_{Y|X}(y|x) , \label{Vnum}
\end{align}
where (a) follows by setting $v_i=(u^\star,i)$, where
\begin{align*}
u^\star  = \argmax_{u \in \mathcal{U}} \sum_{x \in \mathcal{X}} P_X(x) P_{U|X}(u|x) P_{Y|X}(y|x).
\end{align*}
Now note that \eqref{Vnum} is simply the probability of guessing $U$ correctly with a single guess after observing $Y$. A similar argument shows that, with no $Y$ observation, the probability of guessing $V$ correctly with $k$ guesses is equal to the probability of guessing $U$ correctly with a single guess, hence $\kml{X}{Y} [V] = \ml{X}{Y}[U]$, which establishes $\kml{X}{Y} \geq \ml{X}{Y}$.

It remains to show $\ml{X}{Y} \geq \kml{X}{Y}$. For any $P_{V|X}$, we construct $P_{U|X}$ such that $\ml{X}{Y}[U]=\kml{X}{Y}[V]$. So let $P_{V|X}$ be given, with associated alphabet $\mathcal{V}$, and let $\ell \triangleq |\mathcal{V}| \geq k$. Now, let 
\begin{align*}
\mathcal{U} = \{ S \subset \mathcal{V}: |S|=k\},
\text{ and } P_{U|X}(u|x) & =  c \sum_{v \in u} P_{V|X}(v|x),
\end{align*}
where $c={1}/{\binom{\ell-1}{k-1}}$. Then, observing $Y$, the probability of guessing $U$ correctly with a single guess is
\begin{align}
 \sup_{X-Y-\hat{U}} \Pr(U = \hat{U}) 
&  = \sum_{y \in \mathcal{Y}} \max_{u \in \mathcal{U}} \sum_{x \in \mathcal{X}} P_X(x)P_{U|X}(u|x)P_{Y|X}(y|x) \notag \\
& = \sum_{y \in \mathcal{Y}} \max_{u \in \mathcal{U}} \sum_{x \in \mathcal{X}} P_X(x) \sum_{v \in u} P_{V|X}(v|x) P_{Y|X}(y|x) c \notag \\
& = c \! \sum_{y \in \mathcal{Y}} \max_{\substack{v_1,v_2,\ldots,v_k \\ v_i \neq v_j, i \neq j}} \! \sum_{x \in \mathcal{X}} \! \sum_{i=1}^k \! P_X(x)P_{V|X}(v_i|x)P_{Y|X}(y|x),  \notag
\end{align}
which is the probability, normalized by $c$, of guessing $V$ correctly with $k$ guesses after observing $Y$. A similar argument shows that, with no $Y$ observation, the probability of guessing $U$ correctly with a single guess is equal to the probability, normalized by $c$, of guessing $V$ correctly with $k$ guesses, hence $\ml{X}{Y} [V] = \kml{X}{Y}[U]$, which establishes $\ml{X}{Y} \geq \kml{X}{Y}$.

\subsection{Proof of Theorem~\ref{thmcondmaxleakage}} \label{app:conditional}
Assume, without loss of generality, that $X$ and $Z$ have full marginal support. 
To show that the left-hand side is upper-bounded by the right-hand side, fix $P_{U|XZ}$ and consider the following.
\begin{align*}
\frac{\Pr(U=\hat{U}(Y,Z))}{\Pr(U=\tilde{U}(Z))} 
 = \frac{\sum_z p(z) \sum_y p(y|z) \max_u p(u|y,z)}{\sum_z p(z) \max_u p(u|z) } 
 \leq \max_z \frac{\sum_y p(y|z) \max_u p(u|y,z)}{\max_u p(u|z) }.
\end{align*} 
Then by noting that the ratio being maximized is $ \exp\{ \mathcal{L}( X \!\! \to \!\! Y | Z=z) \}$, we get
\begin{align*}
\sup_{\substack{U: U-(X,Z)-Y}}  \frac{\Pr(U=\hat{U}(Y,Z))}{\Pr(U=\tilde{U}(Z))} \leq \max_z \sum_y \max_{x: P_{X|Z}(x|z) > 0} P_{Y|XZ} (y|x,z).  
\end{align*}
To get the reverse inequality, let $\epsilon_n = 1/n$ for $n \in \mathbb{N}$, $z^\star \in \argmax \sum_y \max_{x: P_{X|Z}(x|z) > 0} P_{Y|XZ} (y|x,z)$, and $p^\star = \min_{x: p(x|z^\star)> 0} p(x|z^\star)$. Construct $P_{U|XZ}$ as follows. If $Z=z^\star$, then $P_{U|X,Z=z^\star}$ is the ``shattering'' conditional with respect to the distribution $P_{X|Z=z^\star}$ (cf.~equation~\eqref{eqshattering}). If $Z \neq z^\star$, then $U \sim \mathrm{Unif}([n])$, independent of $X$. Using Proposition~\ref{prop:shattering}, we get
\begin{align*}
\frac{\Pr(U=\hat{U}(Y,Z))}{\Pr(U=\tilde{U}(Z))} 
&  = \frac{\sum_{z \neq z^\star} p(z) \sum_y p(y|z) \max_u p(u|y,z) + p(z^\star) \sum_y p(y|z^\star) \max_u p(u|y,z^\star)}{\sum_{z \neq z^\star} p(z)  \max_u p(u|z) + p(z^\star)  \max_u p(u|z^\star) } \\
& =  \frac{\sum_{z \neq z^\star} p(z) \epsilon_n + p(z^\star) p^\star \sum_y \max_{x: P_{X|Z}(x|z^\star) > 0} P_{Y|XZ} (y|x,z^\star)}{\sum_{z \neq z^\star} p(z) \epsilon_n + p(z^\star) p^\star } \\
& =   \frac { (1-p(z^\star))\epsilon_n + p(z^\star) p^\star \sum_y \max_{x: P_{X|Z}(x|z^\star) > 0} P_{Y|XZ} (y|x,z^\star)}{(1-p(z^\star)) \epsilon_n + p(z^\star) p^\star }
\end{align*}
Letting $n \rightarrow \infty$ (i.e., $\epsilon_n \rightarrow 0$) yields our lower bound.

\subsection{Proof of Corollary~\ref{corrpropcond}} \label{app:CorrConditionalProof}
\begin{enumerate}
\item The data processing inequality follows directly from the definition as in the unconditional case.
\item The upper bound follows from Theorem~\ref{thmcondmaxleakage} and Lemma~\ref{lemmadefprop}.
\item[3-4)] The independence and additivity properties follow straightforwardly from the theorem.
\item[5)] $I(X;Y|Z) \leq \max_{z \in \mathrm{supp}(Z)} I(X;Y|Z=z) \leq \max_{z \in \mathrm{supp}(Z)} \cml{X}{Y}{Z=z} = \cml{X}{Y}{Z}$.
\item[6)] The asymmetry follows immediately from the unconditional case.
\item[7)] Let $Z-X-Y$ be a Markov chain. Then
\begin{align*}
\cml{X}{Y}{Z} & = \log \left( \max_{z \in \mathrm{supp}(Z)} \sum_y \max_{x: P_{X|Z}(x|z) > 0} P_{Y|XZ} (y|x,z) \right)\\
 & = \log \left( \max_{z \in \mathrm{supp}(Z)} \sum_y \max_{x: P_{X|Z}(x|z) > 0} P_{Y|X} (y|x) \right) \\
 &  \leq  \log \left(\sum_y \max_{x: P_{X}(x) > 0} P_{Y|X} (y|x) \right) \\
 & = \ml{X}{Y},
\end{align*}
where the inequality follows from the fact that $\mathrm{supp}(X) \supseteq \mathrm{supp}(X|Z=z)$ for any $z \in \mathrm{supp}(Z)$. Note that the inequality becomes an equality if for some $z \in \mathrm{supp}(Z)$, $\mathrm{supp}(X) = \mathrm{supp}(X|Z=z)$.
\item[8)]
\begin{align*}
\ml{X}{(Y,Z)} - \ml{X}{Z} & = \log   \frac{\sum_{z,y} \max_{x: P_{X}(x) > 0} P_{YZ|X} (y,z|x) }{\sum_z \max_{x: P_{X}(x) > 0} P_{Z|X} (z|x)}   \\
& \leq \log  \max_{z \in \mathrm{supp}(Z)} \frac{\sum_{y} \max_{x: P_{X}(x) > 0} P_{YZ|X} (y,z|x) }{ \max_{x: P_{X}(x) > 0} P_{Z|X} (z|x)}    \\
& = \log \max_{z \in \mathrm{supp}(Z)} \frac{\sum_{y} \max_{x: P_{X}(x) > 0} P_{Z|X} (z|x) P_{Y|XZ}(y|x,z) }{ \max_{x: P_{X}(x) > 0} P_{Z|X} (z|x)}  \\
& \stackrel{\text{(a)}} =   \log \max_{z \in \mathrm{supp}(Z)} \frac{\sum_{y} \max_{x: P_{X|Z}(x|z) > 0} P_{Z|X} (z|x) P_{Y|XZ}(y|x,z) }{ \max_{x: P_{X}(x) > 0} P_{Z|X} (z|x)} \\
& = \log \max_{z \in \mathrm{supp}(Z)} \sum_{y} \max_{x: P_{X|Z}(x|z) > 0}  P_{Y|XZ}(y|x,z) \frac{P_{Z|X} (z|x)}{\max_{x': P_{X}(x') > 0} P_{Z|X} (z|x')} \\
& \leq \log \max_{z \in \mathrm{supp}(Z)} \sum_{y} \max_{x: P_{X|Z}(x|z) > 0}  P_{Y|XZ}(y|x,z) \\
& = \cml{X}{Y}{Z},
\end{align*}
where (a) follows from the fact that $P_{X|Z}(x|z) =0$, $P_X(x) >0$ and $P_Z(z) > 0$ implies that $P_{Z|X}(z|x) = 0$, so that the maximum is achieved outside this set.
\end{enumerate}

\subsection{Proof of Theorem~\ref{thmgeneralformula}} \label{app:generalformula}
{\bf Proof of 1):} 
To show that the right-hand side upper-bounds the left-hand side, fix any $P_{U|X}$, and consider the following
\begin{align*}
\sup_{\hat{U}(Y)} \Pr (U =\hat{U}(Y)) & = \int_{\mathcal{Y}} \max_{u \in \mathcal{U}} \int_{\mathcal{X}} P_{U|X}(u|x) P_{XY}(dxdy) \\
& = \int_{\mathcal{Y}} \max_{u \in \mathcal{U}} \int_{\mathcal{X}} P_{U|X}(u|x) f(x,y) P_X(dx) P_Y(dy) \\
& \leq \int_{\mathcal{Y}} \max_{u \in \mathcal{U}} \int_{\mathcal{X}} P_{U|X}(u|x) (\essup_{P_X}  f(X,y)) P_X(dx) P_Y(dy) \\
& = \int_{\mathcal{Y}} (\essup_{P_X} f(X,y)) \left( \max_{u \in \mathcal{U}} \int_{\mathcal{X}} P_{U|X}(u|x) P_X(dx) \right) P_Y(dy) \\
& = (\max_{u \in \mathcal{U}} P_U(u)) \int_{\mathcal{Y}} (\essup_{P_X} f(X,y))  P_Y(dy). 
\end{align*}

To show the reverse direction, we will show it first for discrete $X$, and then extend the result by discretizing more general $X$'s. 
Suppose $X$ has a finite alphabet. In this case, $\sigma(X)$ is generated by a finite set, and $P_{XY} \ll P_X \times P_Y$ since $I(X;Y) \leq H(X) < \infty$. Without loss of generality, suppose $X$ has full support. Consider the ``shattering'' $P_{U|X}$. Recall: $p^\star = \min_{x \in \mathcal{X}} P_X(x)$. For each $x \in \mathcal{X}$, let $k(x) = P_X(x)/p^{\star}$, and let $ \mathcal{U} = \bigcup_{x \in \mathcal{X}} \{ (x,1), (x,2), \ldots, (x,\lceil k(x)\rceil) \}$.  For each $u=(i_u,j_u) \in \mathcal{U}$ and $x \in \mathcal{X}$, let $P_{U|X}(u|x)$ be:
\begin{align*}
 P_{U|X}((i_u,j_u)|x) 
& = \begin{cases}
\frac{p^{\star}}{P_X(x)}, &  i_u=x, ~~1 \leq j_u \leq \lfloor k(x) \rfloor , \\
1 - \frac{(\lceil k(x) \rceil -1)p^{\star}}{P_X(x)}, & i_u=x,~~ j_u=\lceil k(x) \rceil, \\
0, &  i_u \neq x, ~~1 \leq j_u \leq \lceil k(i_u) \rceil.
\end{cases}
\end{align*} 
Then
\begin{align*}
\sup_{\hat{U}(Y)} \Pr (U =\hat{U}(Y)) & = \int_{\mathcal{Y}} \max_{(i_u,j_u) \in \mathcal{U}} \sum_{x \in \mathcal{X}} P_{U|X}((i_u,j_u)|x) f(x,y) P_X(x) P_Y(dy) \\
& = \int_{\mathcal{Y}} \max_{(i_u,1) \in \mathcal{U}} p^\star f(i_u,y)  P_Y(dy) \\
& = p^\star \int_{\mathcal{Y}} \max_{x \in \mathcal{X}}  f(x,y)  P_Y(dy).
\end{align*}
The proof for the discrete case is completed by noticing that $p^\star = \max_u P_U(u)$. 

Now, consider the more general case. Let $\{A_n\}_{n=1}^{\infty}$ be a countable collection of sets generating $\sigma(X)$. We will prove the result by considering a series of discretizations of $X$, each of which is a refinement of the previous one. To that end,  let $\mathcal{S}_n$ be the finite partition generating $\sigma(\cup_{i=1}^n A_i)$. It can be readily verified that $\mathcal{S}_{n+1}$ is a refinement of $\mathcal{S}_n$. Let $N_n = | \mathcal{S}_n|$, $\mathcal{S}_n =\{S_{n,1},S_{n,2}, \cdots, S_{n,N_n}\}$, and define
\begin{align*}
U_n(X) = \sum_{i=1}^{N_n} i ~\mathbb{I}\{X \in S_{n,i}\},
\end{align*} where $\mathbb{I} \{.\}$ is the indicator function.
Then we get $\mathcal{L}( X \to Y ) \geq \mathcal{L} (U_n \to Y)$ since $U_n - X- Y$ is a Markov chain, and the data processing inequality holds by Lemma~\ref{lemmadefprop}. By the earlier result for finite $X$, we have
\begin{align*}
\mathcal{L} (U_n \to Y) = \log \int_{\mathcal{Y}} \sup_{u: P_{U_n}(u_n) >0} f_n(u_n,y) P_Y(dy),
\end{align*}
where $f_n(u_n,y) = \frac{dP_{U_nY}}{d(P_{U_n} \times P_Y)}.$ We next compute $f_n(u_n,y)$. Let $A \subseteq \mathcal{U}_n \times \mathcal{Y}$. Then
\begin{align*}
P_{U_n,Y} (A) & = \int_{\mathcal{Y}} \sum_{u_n} \mathbb{I} \{(u_n,y) \in A\} \int_{\mathcal{X}} P_{U_n|X} (u_n|x) f(x,y) P_X(dx) P_Y(dy)  \\
& = \int_{\mathcal{Y}} \sum_{u_n} \mathbb{I} \{(u_n,y) \in A\} \left( \int_{S_{n,u_n}}  f(x,y) P_X(dx) \right) P_Y(dy) \\
& = \int_{\mathcal{Y}} \sum_{u_n: P_{U_n}(u_n)>0 } \mathbb{I} \{(u_n,y) \in A\} \left( \frac{ \int_{S_{n,u_n}}  f(x,y) P_X(dx)} { \int_{S_{n,u_n}}  P_X(dx)} \right) P_{U_n}(u_n) P_Y(dy),
\end{align*}
so that 
\begin{align*}
f_n(u_n,y) = \frac{ \int_{S_{n,u_n}}  f(x,y) P_X(dx)} { \int_{S_{n,u_n}}  P_X(dx)}.
\end{align*}
Let $S_n(x)$ be the set in $\mathcal{S}_n$ containing $x$. Then we can view $f_n(u_n,y)$ as a function of $(x,y)$:
\begin{align*}
f_n(x,y) = \frac{ \int_{S_{n}(x)}  f(x,y) P_X(dx)} { \int_{S_{n}(x)}  P_X(dx)}.
\end{align*}
We can rewrite $f_n(x,y) = \E[f(X,y)   | X \in S_n(x)]$, so that
\begin{align}
f_n(X,y) =  \E[f(X,y) | \sigma(\mathcal{S}_n)]. ~~ P_X-\text{a.s.}
\end{align}
Since $\mathcal{S}_n$'s are refinements, $f_n(X,y)$ is a martingale process, and it follows by Levy's upward Theorem~\cite[Theorem 14.2]{ProbabilityWithMartingales} that \begin{align} \label{eqfnalmostsureconv1}
f_n(X,y) \stackrel{\text{a.s.}} \rightarrow \E\left[f(X,y) | \sigma\left(\cup_{i=1}^\infty \mathcal{S}_i\right)\right].
\end{align}
Then
\begin{align} \label{eqfnalmostsureconv2}
\E\left[f(X,y) | \sigma\left(\cup_{i=1}^\infty \mathcal{S}_i\right)\right] = \E\left[f(X,y) | \sigma\left(\cup_{i=1}^\infty A_i\right)\right] = \E\left[f(X,y) | \sigma\left(X \right)\right] \stackrel{\text{a.s.}} = f(X,y).
\end{align}
Moreover, 
\begin{align}
\mathcal{L}( X \to Y ) \geq \limsup_{n \rightarrow \infty} \mathcal{L} (U_n \to Y) & = \limsup_{n \rightarrow \infty} \log \int_{\mathcal{Y}} \sup_{u: P_{U_n}(u_n) >0} f_n(u_n,y) P_Y(dy) \\
& = \limsup_{n \rightarrow \infty} \log \int_{\mathcal{Y}} \sup_{x: P_{X}(S_n(x)) >0} f_n(x,y) P_Y(dy).
\end{align}
Since $\mathcal{S}_{n+1}$ is a refinement of $\mathcal{S}_n$, the integrand is nondecreasing with $n$. Therefore, by the monotone convergence theorem, 
\begin{align}
\mathcal{L}( X \to Y ) \geq  \log \int_{\mathcal{Y}} \lim_{n \rightarrow \infty} \sup_{x: P_{X}(S_n(x)) >0} f_n(x,y) P_Y(dy)  .
\end{align}
Then it remains to show that
\begin{align}
\label{eqlimsupsuplim}
\lim_{n \rightarrow \infty} \sup_{x: P_{X}(S_n(x)) >0} f_n(x,y) \geq \essup_{P_X} f(X,y)
\end{align}
for all $y$.
To that end, let $B = \{ \alpha : P_X( f(X,y) > \alpha ) > 0 \}.$ Consider $r \in B$  and let $E_r = \{x: f(x,y) > r \}$. Then $P_X(E_r) > 0$. Therefore, by~\eqref{eqfnalmostsureconv1} and~\eqref{eqfnalmostsureconv2}, $f_n(X,y)$ converges almost everywhere to $f(X,y)$ on $E_r$. By Egoroff's Theorem~\cite[Theorem 7.12]{bartle2014elements}, for every $\delta >0$, there exists $E'_\delta$ such that $P_X(E'_\delta)<\delta$ and $f_n$ converges uniformly to $f$ on $E_r \backslash E'_\delta$. Call the latter set $E_{r \backslash \delta}$.
So fix $\delta > 0$ small such that $P_X( E_{r \backslash \delta}) > 0$. For each $n$, let $\mathcal{S}_n (E_{r \backslash \delta})$ be a collection of sets in $\mathcal{S}_n$ satisfying: $\cup_{S \in \mathcal{S}_n (E_{r \backslash \delta})} \supseteq E_{r \backslash \delta}$ and $S \in \mathcal{S}_n (E_{r \backslash \delta}) \Rightarrow S \cap E_{r \backslash \delta} \neq \varnothing$. Then there must exist $S \in \mathcal{S}_n (E_{r \backslash \delta})$ satisfying $P(S) > 0$. Denote the latter set by $S_n ( E_{r \backslash \delta})$. Hence,
\begin{align}
\lim_{n \rightarrow \infty} \sup_{x: P_{X}(S_n(x)) >0} f_n(x,y) &
\stackrel{\text{(a)}} \geq  \lim_{n \rightarrow \infty} \sup_{x \in S_n ( E_{r \backslash \delta})} f_n(x,y) \\ 
& \stackrel{\text{(b)}} \geq \lim_{n \rightarrow \infty} \inf_{x \in E_{r \backslash \delta} } f_n (x,y) \\
& \stackrel{\text{(c)}} =  \inf_{x \in E_{r \backslash \delta} } f(x,y) \\
& \geq r,
\end{align} 
where (a) follows from the fact that $P_X( S_n (E_{r \backslash \delta})) > 0$, (b) follows from the fact that $S_n (E_{r \backslash \delta}) \cap E_{r \backslash \delta} \neq \varnothing$, and (c) follows from the fact that $f_n(x,y)$ converges uniformly to $f$ on $E_{r \backslash \delta}$. Finally, since $r $ was chosen arbitrarily from $B$, we get
\begin{align}
\lim_{n \rightarrow \infty} \sup_{x: P_{X}(S_n(x)) >0} f_n(x,y)  \geq \sup B = \essup_{P_X} f(X,y),
\end{align}
as desired. 

{\bf Proof of 2):}
If absolute continuity does not hold, then $I(X;Y) = +\infty$, and there exists a sequence of discretizations $(X_n,Y_n)$ such that $I(X_n;Y_n) \rightarrow +\infty$ (e.g.,~\cite[p.~37]{Gallager:IT}). The result then follows by noting that $\ml{X}{Y} \geq \ml{X_n}{Y_n} \geq I(X_n;Y_n)$.

\subsection{Proof of Lemma~\ref{lemmageneralformula}} \label{app:lemmageneral}
Suppose $P_{X_1 Y_1} \ll P_{X_1} \times P_{Y_1}$ and let $ d P_{X_1 Y_1} = f_1(x,y) d (P_{X_1} \times P_{Y_1})$. 
Then for every $A \in \sigma_X$ and $B \in \sigma_Y$,
\begin{equation*}
\int_\mathcal{X} \left[ \int_\mathcal{Y} 1(y \in B) d\mu(x,dy) \right]
 1(x \in A) dP_{X_1}(dx)
  = \int_\mathcal{X} \left[ \int_\mathcal{Y} 1(y \in B) f_1(x,y)
    dP_{Y_1}(dy) \right] 1(x \in A) dP_{X_1}(dx).
\end{equation*}
Since this holds true for all $A$ we must have
\begin{equation}
\label{eq:generalproofderivative}
\int_\mathcal{Y} 1(y \in B) d\mu(x,dy) = \int_{\mathcal{Y}}
      1(y \in B) f_1(x,y) dP_{Y_1} (dy) \quad \text{$P_{X_1}$--a.s.}
\end{equation}
Hence $\mu(x,\cdot) \ll P_{Y_1}$ and $f_1(x,y) = \frac{d \mu(x,\cdot)}{dP_{Y_1}}(y)$. Let $P_X$ be an arbitrary representative of the equivalence class of $P_{X_1}$ and $Q_Y$ be any measure satisfying $P_{Y_1} \ll Q_Y$. Then
\begin{align*}
\ml{X_1}{Y_1} & \stackrel{\text{(a)}} = \log \int_{\mathcal{Y}} \essup_{P_{X_1}} \left( \frac{d \mu(X,\cdot)}{dP_{Y_1}}(y) \right)  P_{Y_1}(dy) \\
& =  \log \int_{\mathcal{Y}} \essup_{P_{X_1}} \left( \frac{d \mu(X,\cdot)}{dP_{Y_1}}(y)  \right)  \frac{dP_{Y_1}}{dQ_{Y}}(y)   Q_{Y}(dy) \\
& = \log \int_{\mathcal{Y}} \essup_{P_{X_1}} \left( \frac{d \mu(X,\cdot)}{dQ_{Y}}(y) \right) Q_{Y}(dy) \\
& \stackrel{\text{(b)}} = \log \int_{\mathcal{Y}} \essup_{P_{X}} \left( \frac{d \mu(X,\cdot)}{dQ_{Y}}(y) \right) Q_{Y}(dy),
\end{align*} where (a) follows from Theorem~\ref{thmgeneralformula}, and (b) follows from the fact that for any function $h: \mathcal{X} \rightarrow \mathbb{R}$, $\essup_{P_{X}} h(X) = \essup_{P_{X_1}} h(X)$ when $P_{X_1} \equiv P_{X}$.
Now consider $P_{X_2}$ satisfying $P_{X_2} \equiv P_{X_1}$ and let $g(x) = \frac{d P_{X_2}}{d P_{X_1}}$. For any set $A \in \sigma_{XY}$,
\begin{align*}
P_{X_2 Y_2} (A)= \int_{\mathcal{X}} \int_{\mathcal{Y}} \mathbb{I} \{ (x,y) \in A \} \mu(x,dy) P_{X_2} (dx)  = \int_{\mathcal{X}} \int_{\mathcal{Y}} g(x) \mathbb{I} \{ (x,y) \in A \} \mu(x,dy) P_{X_1} (dx), 
\end{align*}
hence $P_{X_2 Y_2} \ll P_{X_1 Y_1}$. Similarly, for any set $B \in \sigma_Y$,
\begin{align*}
P_{Y_2} (B) = \int_{\mathcal{X}} \int_{\mathcal{Y}} \mathbb{I} \{ y \in B \} \mu(x,dy) P_{X_2} (dx)    = \int_{\mathcal{X}} \int_{\mathcal{Y}} \mathbb{I} \{ y \in B \} f_1(x,y)   P_{Y_1} (dy) P_{X_2} (dx) , 
\end{align*} and
\begin{align*}
P_{Y_1} (B) = \int_{\mathcal{X}} \int_{\mathcal{Y}} \mathbb{I} \{ y \in B \} f_1(x,y)   P_{Y_1} (dy) P_{X_1} (dx).
\end{align*}
Hence $P_{Y_2}(B) = 0$ implies that for $(P_{X_2} \times P_{Y_1})$-almost all $(x,y)$, $\mathbb{I} \{ y \in B \} f_1(x,y) = 0$. Since $P_{X_1} \ll P_{X_2}$, this implies that for $(P_{X_1} \times P_{Y_1})$-almost all $(x,y)$, $\mathbb{I} \{ y \in B \} f_1(x,y) = 0$~\cite[p.~22, Ex.~19]{Kallenberg:Prob2}. Hence $P_{Y_1} (B) = 0$, which implies that $ P_{Y_1} \ll P_{Y_2}$. Therefore, we get
\begin{align*}
P_{X_2 Y_2} \ll P_{X_1 Y_1} \ll P_{X_1} \times P_{Y_1} \stackrel{\text{(a)}} \ll P_{X_2} \times P_{Y_2},
\end{align*} where (a) follows from the fact that $P_{X_1} \ll P_{X_2}$ and $P_{Y_1} \ll P_{Y_2}$. 
By symmetry we also get $P_{Y_2} \ll P_{Y_1}$, hence $P_{Y_1} \equiv P_{Y_2}$. By choosing $P_{X_1}$ to be the representative of the equivalence classes of $P_{X_2}$ and noting that $P_{Y_2} \ll P_{Y_1}$, the first part of the lemma yields 
\begin{align*}
 \ml{X_2}{Y_2} = \log \int_{\mathcal{Y}} \essup_{P_{X_1}} \left( \frac{d \mu(X,\cdot)}{dP_{Y_1}} \right) P_{Y_1}(dy) = \ml{X_1}{Y_1}.
\end{align*}


\section{Proof of Equation~\eqref{eqHammingresult}} \label{app:maxleakhamming}

Let $P_{Y|X} = \begin{bmatrix} 1-W_{10} & W_{10} \\
 W_{01} & 1-W_{01} \end{bmatrix}  $ (where the first column corresponds to $y=0$, the second to $y=1$). Dropping the logarithm, we can rewrite the problem as:
\begin{align} \label{eqHamming2}
\text{minimize}~ & \max\{1-W_{10}, W_{01}\}+\max\{W_{10}, 1-W_{01}\}  \\
\text{subject to}~ & (1-p)W_{10} +pW_{01} \leq D,~~ 0 \leq W_{10}, W_{01} \leq 1. \notag
\end{align}
Now note that 
\begin{align} \label{eqHammingbound}
W_{10}+W_{01} \stackrel{\text{(a)}} \leq \frac{1-p}{p} W_{10} + W_{01} = \frac{1}{p} ((1-p)W_{10}+pW_{01}) \stackrel{\text{(b)}} \leq  \frac{D}{p} \stackrel{\text{(c)}} \leq 1,
\end{align}
where (a) follows because $p \leq 1/2$, (b) follows from the constraint in~\eqref{eqHamming2}, and (c) follows because $D \leq p$. Using~\eqref{eqHammingbound}, we can rewrite~\eqref{eqHamming2} as
\begin{align}
\text{minimize}~ & 2-(W_{10}+W_{01})  \\
\text{subject to}~ & (1-p)W_{10} +pW_{01} \leq D,~~ 0 \leq W_{10}, W_{01} \leq 1. \notag
\end{align}
Therefore, we need to maximize $(W_{10}+W_{01})$. By~\eqref{eqHammingbound}, the sum is upper-bounded by $D/p$. The upper bound can be achieved by setting 
\begin{align}
W^\star_{10} = 0 \text{ and } W^\star_{01} = D/p,
\end{align}
which clearly satisfies the constraint in~\eqref{eqHamming2}. Therefore, 
\begin{align} \label{eqHammingresultapp}
 \min_{ \substack{ P_{Y | X}: \\ \E[d(X,Y)] \leq D}} \ml{X}{Y}= \log_2 (2-D/p) \quad (bits).
\end{align}

%

\section{Proofs for Sections~\ref{seccostleakage}}

\subsection{Proof of Theorem~\ref{thmcostleakage}} \label{app:costleakthm}
To show that the left-hand side is upper-bounded by the right-hand side, fix $U$, $\hat{\mathcal{U}}$ and $d$, and consider:
\begin{align*}
\inf_{\hat{u}(\cdot) } \E[d(U,\hat{u}(Y))]  & = \sum_{y \in \mathcal{Y}} \inf_{\hat{u}} P_Y(y) \E[d(U, \hat{u}) | Y=y] \\
& = \sum_{y \in \mathcal{Y}} \inf_{\hat{u}}   \sum_{x \in \mathrm{supp}(X)} P_X(x,y) \E[d(U, \hat{u}) | X=x,Y=y] \\
& =  \sum_{y \in \mathcal{Y}} \inf_{\hat{u}}  \sum_{x \in \mathrm{supp}(X)} P_X(x) P_{Y|X}(y|x) \E[d(U, \hat{u}) | X=x] \\
& \geq \sum_{y \in \mathcal{Y}}  \min_{ \tilde{x} \in \mathrm{supp}(X)} P_{Y|X}(y|\tilde{x}) \inf_{\hat{u}} \sum_{x \in \mathrm{supp}(X)} P_X(x) \E[d(U, \hat{u}) | X=x] \\
& =  \sum_{y \in \mathcal{Y}}  \min_{ \tilde{x} \in \mathrm{supp}(X)} P_{Y|X}(y|\tilde{x}) \left( \inf_{ \hat{u}} \E[d(U, \hat{u})] \right),
\end{align*}
where the third equality follows from the Markov chain $U-X-Y$.
For the reverse direction, let $U=X$, $\hat{\mathcal{X}} = \mathrm{supp}(X) $, and
\begin{align} \label{eqdistmax}
d(x,\hat{x}) = \begin{cases}
\frac{1}{P_X(x)}, & x = \hat{x}, \\
0, & x \neq \hat{x}.
\end{cases}
\end{align}
Then
\begin{align}
\min_{\hat{x} \in \mathrm{supp}(X)} \E[d(X, \hat{x} )] =  \min_{\hat{x} \in \mathrm{supp}(X)} \sum_{x \in \mathrm{supp}(X)}  P_X(x) d(x,\hat{x}) = \min_{\hat{x} \in \mathrm{supp}(X) }  P_X(\hat{x}) d(\hat{x},\hat{x}) = 1,
\end{align} and for a given $y \in \mathcal{Y}$,
\begin{align}
\min_{\hat{x} \in \mathrm{supp}(X)} \sum_{x \in \mathrm{supp}(X)}  P_X(x) P_{Y|X}(y|x) d(x,\hat{x}) =  \min_{\hat{x} \in \mathrm{supp}(X)} P_{Y|X} (y|\hat{x}),
\end{align}
which concludes the proof.

\subsection{Proof of equation~\eqref{eq:MCLequiv}} \label{app:Cost2}
Fix any distribution $Q_Y$ on $\mathcal{Y}$. Then
\begin{align*}
\exp\{D_{\infty} (P_X \times Q_Y || P_{XY})\} = \max_{ \substack{x,y: \\ P_X(x)Q_Y(y) > 0}} \frac{Q_Y(y)}{P_{Y|X}(y|x)} = \max_{y:Q_Y(y)>0} \frac{Q_Y(y)}{\min\limits_{x: P_X(x)>0} P_{Y|X}(y|x)}. 
\end{align*}
If for every $y \in \mathcal{Y}$ there exists $x \in \mathrm{supp}(P_X)$ such that $P_{Y|X}(y|x) = 0$, then for any $Q_Y$ the above quantity is $\infty$. By Theorem~\ref{thmcostleakage}, $\mcl{X}{Y}$ is also $\infty$ in this case. Now assume $ \sum_{y \in \mathcal{Y}} \min_{x: P_X(x)>0} P_{Y|X}(y|x) > 0$. We have
\begin{align*}
\exp\{D_{\infty} (P_X \times Q_Y || P_{XY})\} = \max_{y:Q_Y(y)>0} \frac{Q_Y(y)}{\min\limits_{x: P_X(x)>0} P_{Y|X}(y|x)} 
& \geq \frac{\sum_{y \in \mathcal{Y}} Q_Y(y)}{ \sum\limits_{y \in \mathcal{Y}} \min\limits_{x: P_X(x)>0} P_{Y|X}(y|x)}.
\end{align*}
Noting that $\sum_{y} Q_Y(y) = 1$, we get $\inf_{Q_Y} D_{\infty} (P_X \times Q_Y || P_{XY})\} \geq \mcl{X}{Y}$. One can readily verify that the lower bound is achievable by setting 
\[ Q_Y(y) = \frac{ \min_{x: P_X(x)>0} P_{Y|X}(y|x) }{ \sum_{y' \in \mathcal{Y}} \min_{x: P_X(x)>0} P_{Y|X}(y'|x) }. \]
\begin{Remark}
In the case of $I_{\infty} (X;Y)$ (cf.~\eqref{eq:SibsonMI}), one can readily verify that
\begin{equation}
\label{eq:optQY}
Q_Y(y) = \frac{ \max_{x: P_X(x)>0} P_{Y|X}(y|x) }{ \sum_{y' \in \mathcal{Y}} \max_{x: P_X(x)>0} P_{Y|X}(y'|x) }
\end{equation}
achieves the infimum in (\ref{eq:SibsonMIinfDinf}).
\end{Remark}

\subsection{Proof of Corollary~\ref{corrcostleakage}} \label{app:costleakcorr}

In the following, assume $X$ has full support. 
\begin{enumerate}
	\item The data processing inequality follows directly from the definition. 
	\item The ``if'' direction is straightforward. The ``only if'' direction follows from the fact that, for each $y$, $\min_{x } P_{Y|X}(y|x) \leq  P_Y(y)$. Thus, $\sum_y \min_x P_{Y|X} (y|x) =1 \Rightarrow \forall y, \min_x P_{Y|X} (y|x) = P_Y (y) \Rightarrow X$ and $Y$ are independent.
	\item[3-5)] The additivity property and the equality in 4) can be readily verified. Example~\ref{exampleasy2} illustrates 5).
	\item[6)] Local-differential privacy upper-bounds maximal cost leakage since:
	\begin{align*}
	\frac{1}{ \sum_y \min_x P_{Y|X} (y|x)} = \frac{ \sum_y P_Y(y) }{ \sum_y \min_x P_{Y|X} (y|x)} \leq \max_y \frac{P_Y(y)}{\min_x P_{Y|X}(y|x)} \leq \max_{x,x',y} \frac{P_{Y|X} (y|x')}{P_{Y|X}(y|x) }.
	\end{align*}
	\item[7)] Convexity follows from the fact that $\min_x P_{Y|X} (y|x)$ is concave in $P_{Y|X}$, and $(-\log)$ is a non-increasing convex function.
\end{enumerate}

\subsection{Proof of Theorem~\ref{thmcostrealizable}} \label{app:costrealizable}

Without loss of generality, assume $X$ and $Y$ have full marginal support. To show $\mathcal{L}^{rc}(X \!\! \to \!\! Y) \leq D_\infty (P_{X} \times P_Y || P_{XY})$, fix any $\hat{\mathcal{X}}$, $d$ and $y \in \mathcal{Y}$, and consider:
\begin{align*}
\inf_{ \hat{u} \in \hat{\mathcal{U}}}  \E[d(U,\hat{u}) | Y=y] & =  \inf_{ \hat{u} \in \hat{\mathcal{U}}} \sum_{ u \in \mathcal{U}} P_{U|Y}(u|y) d(u, \hat{u}) \\
& = \inf_{ \hat{u} \in \hat{\mathcal{U}}} \sum_{ u \in \mathcal{U}} \sum_{ x \in \mathcal{X}} P_{X|Y}(x|y) P_{U|X} (u|x) d(u, \hat{u}) \\
&  = \inf_{ \hat{u} \in \hat{\mathcal{U}}} \sum_{ u \in \mathcal{U}} \sum_{ x \in \mathcal{X}} \frac{P_{Y|X}(y|x)}{P_Y(y)}   P_X(x) P_{U|X}(u|x) d(u, \hat{u}) \\  
& \geq \inf_{ \hat{u} \in \hat{\mathcal{U}}} \frac{\min_{x'} P_{Y|X}(y|x')}{P_Y(y)} \sum_{ u \in \mathcal{U}} \sum_{ x \in \mathcal{X}} P_X(x) P_{U|X}(u|x) d(u, \hat{u}) \\
& = \frac{\min_{x'} P_{Y|X}(y|x')}{P_Y(y)} \inf_{\hat{u} \in \hat{\mathcal{U}}} \E[d(U,\hat{u})].
\end{align*}
The reverse direction follows by using the same $d$ as in~\eqref{eqdistmax}.

\subsection{Proof of Corollary~\ref{corrLDPmaxcostrealized}} \label{app:corrcostrealized}

To show $\sup_{P_X} \mathcal{L}^{rc}(X \!\! \to \!\! Y) \leq L^{dp} ( X \!\! \to \!\! Y)$, note that $P_Y(y) \leq \max_{x} P_{Y|X}(y|x)$. For the reverse direction, consider the following. Let $y^\star$ be an element achieving the max of $L^{dp}$. Let $x_0 \in \argmin_x P_{Y|X}(y^\star|x)$ and $x_1 \in \argmax_x P_{Y|X}(y^\star|x)$. Finally, for a given $\alpha > 0$, let  $P_X(x_0) = 1- \alpha$ and $P_X(x_1) = \alpha$. Then,
\begin{align*}
 \sup_{P_X} \mathcal{L}^{rc}(X  \!\!  \to \!\!   Y) 
& \geq \log  \frac{P_{Y}(y^\star) }{ P_{Y|X} (y^\star|x_0) } \\
& = \log \frac{ (1-\alpha) P_{Y|X}(y^\star|x_0) + \alpha P_{Y|X}(y^\star|x_1) } { P_{Y|X} (y^\star|x_0) } \\
& \xrightarrow{\alpha \rightarrow 1} \log \frac{P_{Y|X}(y^\star|x_1)}{P_{Y|X}(y^\star|x_0)} = L^{dp} (X    \rightarrow    Y). \tag*{$\blacksquare$}
\end{align*} 

\end{appendices}

\bibliographystyle{IEEEtran}
\bibliography{IEEEabrv,database}

\end{document}